\documentclass[a4paper]{article}
\usepackage{titling}
\usepackage{a4wide}
\usepackage{amsmath}
\usepackage{amsfonts}
\usepackage{amssymb}
\usepackage{amsthm}
\usepackage{bbm}
\usepackage{hyperref}
\usepackage{cleveref}
\usepackage[utf8]{inputenc}
\usepackage{color}
\usepackage{url}
\usepackage{appendix}
\usepackage{csquotes}

\MakeOuterQuote{"}

\newtheorem{theorem}{Theorem}[section]
\newtheorem{definition}[theorem]{Definition}
\newtheorem{corollary}[theorem]{Corollary}
\newtheorem{lemma}[theorem]{Lemma}

\newtheorem{assumption}[theorem]{Assumption}

\numberwithin{equation}{subsection}
\numberwithin{theorem}{subsection}

\newcommand{\UHaddress}{\em University of Helsinki, Department of Mathematics 
and Statistics\\
\em P.O. Box 68, FI-00014 Helsingin yliopisto, Finland}
\newcommand{\email}[1]{E-mail: \tt #1}
\newcommand{\emailkalle}{\email{kalle.koskinen@helsinki.fi}}

\title{Infinite volume Gibbs states and metastates of the random field mean-field spherical model}
\author{Kalle Koskinen\thanks{\emailkalle} \\[1em] \UHaddress}
\date{\today}

\begin{document}
\begin{titlingpage}
\maketitle
\begin{abstract}
\noindent
For the discrete random field Curie-Weiss models, the infinite volume Gibbs states and metastates have been investigated and determined for specific instances of random external fields. In general, there are not many examples in the literature of non-trivial limiting metastates for discrete or continuous spin systems. We analyze the infinite volume Gibbs states of the mean-field spherical model, a model of continuous spins, in a general random external field with independent identically distributed components with finite moments of some order larger than four and non-vanishing variances of the second moments. Depending on the parameters of the model, we show that there exist three distinct phases: ordered ferromagnetic, ordered paramagnetic, and spin glass. In the ordered ferromagnetic and ordered paramagnetic phases, we show that there exists a unique infinite volume Gibbs state almost surely. In the spin glass phase, we show the existence of chaotic size dependence, provide a construction of the Aizenman-Wehr metastate, and consider both the convergence in distribution and almost sure convergence of the Newman-Stein metastates. The limiting metastates are non-trivial and their structure is universal due to the presence of Gaussian fluctuations and the spherical constraint.  
\end{abstract}
\end{titlingpage}
\tableofcontents
\addtocontents{toc}{\protect\thispagestyle{empty}}
\newpage
\section{Introduction}
The mean-field spherical model in a random external field, which can also be called the random field mean-field spherical model(RFMFS) is a variant of the mean-field spherical model(MFS), introduced in \cite{Kastner2006}, where the homogeneous external field is replaced by a random external field. Due to the introduction of external randomness, as opposed to the internal randomness of the associated Gibbs state, the finite volume Gibbs states(FVGS) of this model form a collection of random probability measures. In this paper, we will prove theorems concerning the infinite volume limits of the collection of FVGS for a number of different modes of convergence introduced in the disordered systems literature. Due to model specific features, we are able to prove these theorems in a significantly more general setting than other similar models which can be found in the literature.
\\
\\
The Curie-Weiss model in a random external field with independent Bernoulli distributed components was first introduced in \cite{Schneider1977}. In \cite{Schneider1977}\cite{Aharony1978}\cite{Salinas1985} the thermodynamics and the phase diagram of this model are determined. In \cite{Matos1991}, the fluctuations of magnetization are considered for the Curie-Weiss model in a random external field where the Bernoulli distributed components are replaced by independent identically distributed components with finite absolute first moments. We will refer to the general random external field case as the random field Curie-Weiss model(RFCW), and the Bernoulli distributed external field as the Bernoulli field Curie-Weiss model(BFCW). One of the key results in \cite{Matos1991} is the classification of the magnetization fluctuations in the RFCW model in terms of the "type" and "strength" of the global maximizing points of the free energy. This classification is similar to the classification given by Ellis and Newman in \cite{Ellis1978} for the magnetization fluctuations of the generalized Curie-Weiss model(GCW), but the external randomness modifies the scaling of the fluctuations.
\\
\\
The infinite volume limits of the finite marginal distributions of the GCW model are given in \cite{Ellis1978}, and the result is that the they are convex combinations of  product measures where the number of such product measures is related to the amount, type, and strength of the global maximizing points of the free energy of the model. Due to the similarities in results and proof techniques concerning the magnetization fluctuations between the GCW model and the RFCW model, one might suspect that the infinite volume limits would also be similar. This, however, is not necessarily the case.
\\
\\
For the RFCW model, in \cite{Matos1992}, the authors show that the infinite volume Gibbs states(IVGS), i.e. any element of the collection of limit points of the collection of FVGS, always belongs to the set of convex combinations of some product measures almost surely, where the number of product measures is determined by the amount of global maximizing points of the free energy. However, for the BFCW model where the parameters of the model are chosen so that there are exactly two global maximizing points, they show that there exists a set of probability $1$ such that a countable collection of convex combinations of the two associated product measures can be obtained as convergent subsequences almost surely. This phenomenon is known as chaotic size dependence due to Newman and Stein \cite{Newman1997}. The implication is that there do not exist a definitive unique almost sure limits of FVGS, instead, they are dependent on the subsequence of volumes chosen.
\\
\\
To rectify the problem of chaotic size dependence, other forms of convergence in the infinite volume limit must be considered. Let us briefly motivate two such forms of convergence and introduce the notion of metastates. Since the FVGS are random probability measures, if one considers their distributions as random variables, then the distribution is essentially a probability measure of probability measures, and we call it a metastate. One can then consider the convergence in distribution of the external field and the FVGS simultaneously. For the resulting limit in distribution, if it exists, one can then take the regular conditional distribution of the random probability measure given the external field, and this is then a metastate which is defined almost surely. This procedure is due to Aizenman and Wehr\cite{Aizenman1990}, and it is typically called the Aizenman-Wehr metastate(A-W) or the conditioned metastate. For the other form of convergence, one considers the empirical measure of the sequence of FVGS and its limit either almost surely or in distribution. The  empirical measures are called the Newman-Stein metastates(N-S) due to Newman and Stein \cite{Newman1997}. We refer to \cite[Chapter 6]{Bovier2006} for more details and definitions concerning the metastates.
\\
\\
In \cite{Matos1992}, for the BFCW in the case where there are two global maximizing points, it was shown that the limit in distribution of the random FVGS corresponds to a random probability measure which is split between the associated product measures with probability $\frac{1}{2}$ independent of the external field. A similar result, concerning the A-W metastate, is obtained in \cite{Kuelske1997}. In \cite{Kuelske1997}, the author also provides a proof that the N-S metastates corresponding to the FVGS do not converge almost surely, but they do converge in distribution. In doing so, it is also shown that the limit in distribution of the N-S metastates is more general in some sense since it also contains the result concerning the A-W metastate. 
\\
\\
Let us now comment on some of the technical details in the RFCW model and the BFCW model so that we have a reference point for some of the results in this paper. The characterization given in \cite{Matos1991} for the free energy of the RFCW suggests that one can find random external fields which would yield any amount, type, and strength of the global maximizing points of the free energy. In principle, one could perform a similar analysis for the IVGS of the RFCW model as in \cite{Matos1992}\cite{Kuelske1997}, but both of these works focused on the BFCW model, and there is an explicit mention of the difficulties associated with the general RFCW models in \cite{Matos1992}. There are two features which make the RFCW models technically simple to deal with. The first feature is that the FVGS are probability measures on a compact space, and, as such, the space of probability measures on this compact space is compact itself. As a result, one does not need to provide uniform tightness results for the metastates since they are automatic. The second feature is that the FVGS can be immediately written as a continuous mixture of product states via the Hubbard-Stratonivich transform. Using these ingredients, most of the analysis of these models comes down to the analysis of the free energy and the limiting procedure by which it was obtained.
\\
\\
In this paper, we will rigorously determine the IVGS and the limiting metastates of the RFMFS in the general setting where the external field has components which are independent identically distributed random variables with finite moments of some order larger than four and non-vanishing variances of the second moments. We show that the IVGS of this model is always either unique or a convex combination of two pure states. We will provide a characterization of the phases of this model, and we will show that in the spin glass phase the model exhibits CSD. For the spin glass phase, we will provide a construction of the A-W metastate and consider the limiting properties of the N-S metastates both in distribution and almost surely. The results obtained are universal in the sense that they hold for any random external fields with the assumptions given.
\\
\\
The main results of this paper that concern the case of the random external field are presented in \Cref{main_result3}, \Cref{aw_construction}, \Cref{main_result4}, and \Cref{main_result5}. In order, the results are the existence of CSD and the determination of unique almost sure limits of the  FVGS, the construction of the A-W metastate from the limit in distribution of the conditioned metastate probability measures, the almost sure divergence of the N-S metastates and the almost sure convergence of a random subsequence of N-S metastates to the A-W metastate, and, finally, the limit in distribution of the N-S metastates. More minor, but contextually important results, are the characterization of the spin glass phase provided in \Cref{spinglass_table} and the triviality of the metastates in \Cref{triviality}.
\\
\\
The methods of proof for this model differ substantially from those for the RFCW. We will first prove results for a deterministic inhomogeneous external field with convergent sample means and non-vanishing convergent sample standard deviations. Such a deterministic inhomogeneous external field will be called strongly varying. For strongly varying external fields, we are able to fully determine the infinite volume Gibbs states by specifying the rates of convergence of the sample means and sample standard deviations. The IVGS and metastate results are then applications of specific instances of strongly varying fields. The technical difficulties of this model are two-fold. The first difficulty is that the state space is not compact and thus we must provide a number of uniform tightness results for the different metastates. The second difficulty is that the inclusion of the external field breaks the permutation invariance of the standard mean-field spherical model, and, as a result, we must provide methods of proof for the resolution of finite marginals of certain singular probability distributions which do not require such symmetries.
\\
\\
To our knowledge, the results of this paper for the mean-field spherical model in a strongly varying field and the RFMFS are novel contributions to the literature. In particular, as noted in \cite{Bovier2006}, there are very few explicit constructions of metastates for non-trivial lattice models, and those that do exist, such as the results of the BFCW model, are for specific choices of random external fields. In this paper, we are able to provide an explicit construction of the metastates of the RFMFS model for a large collection of general random external fields, and show that they have a relatively simple and universal structure despite the generality of the random external field.
\subsection{Related works}
Recently, the free energy, overlap structures, and thermodynamic fluctuations were studied for the spherical Sherrington-Kirkpatrick model(SSK) and a variant of it in \cite{Baik2016}\cite{Baik2018}\cite{Baik2021}. The SSK model is similar to our model in the sense that it is a lattice model with a spherical constraint and some sort of external randomness added to the model. However, the critical difference in these models is that the external randomness typically appears in the interactions between spins. Furthermore, these works are not concerned with the construction of metastates. In general, there is a vast literature concerning these types of models too long to cite, instead we refer to the book \cite{Talagrand2011}. 
\\
\\
Let us also mention the large and moderate deviation results for the RFCW models given in \cite{Loewe2013} and \cite{Loewe2012} which initially spurred our interest in these types of models and provided valuable references in their introductions. In particular, there are examples given in \cite{Loewe2013} which consider random external fields which do not have independent identically distributed components.
\\
\\
To accompany the original introduction of the MFS model in \cite{Kastner2006}, the IVGS of the MFS model can be deduced from results presented in \cite{Koskinen2020}. They are not explicitly mentioned as IVGS, but the finite marginal distributions are deduced there. In general, the work of \cite{Koskinen2020} can be used to understand some of the rigorous methods of proof used in this paper as well, but the core methods of this paper differ significantly due to the lack of permutation invariance of the models. 
\subsection{Reading guide}
\noindent
This paper is split into two major sections which are \Cref{main_results_sec} and \Cref{proofs_of_results}. The main results are contained in \Cref{main_results_sec}, and they are presented along with a significant amount of exposition, intermediate results, and proof sketches. When necessary, details for non-trivial proofs are provided in \Cref{proofs_of_results}. In general, the main results of this paper and their proof sketches should be able to be understood by reading only \Cref{main_results_sec}.
\\
\\
Although this paper is primarily concerned with the random external field, \Cref{main_results_sec} is split into two subsections which are \Cref{deterministic} and \Cref{random}. Many key results are first developed for a deterministic inhomogeneous external field in \Cref{deterministic} and then applied to the case of the random external field in \Cref{random}. The main results concerning the random external field are provided in \Cref{random}, but most of their proofs rely directly on results developed in \Cref{deterministic}.
\\
\\
In \Cref{proofs_of_results}, we will provide the non-trivial proofs of main results and intermediate results from \Cref{main_results_sec}. There are also intermediate results and their proofs which may not be found in \Cref{main_results_sec}. To save space and provide clarity, we will typically omit extraneous sup- and subindices if it is clear what object is being referred to in the proof. If there is a relevant variable dependence in these indices, we will keep them in the proofs. In general, we have tried to remain consistent with the naming of certain key objects so that they are the same throughout the paper.
\\
\\
The appendix contains general definitions and proofs related to the uniform tightness, approximation, and weak convergence of probability measures related to the metastates. 
\section{Main results} \label{main_results_sec}
\subsection{Preliminaries} \label{preliminaries}
\noindent
The model of interest in this paper is an equilibrium statistical mechanical model of an unbounded continuous spin system with long-range interactions, a spherical constraint, and an external field which is initially deterministic. Later on, we will consider the case where the external field is random. We will refer to this model and its constituents, which are to be defined, as the RFMFS model.  
\\
\\
To define this model, we begin with the mean-field Hamiltonian in an external field $H_n^{J,h} : \mathbb{R}^n \to \mathbb{R}$ given by
\begin{align} \label{hamiltonian}
H_n^{J,h}[\phi] := - \frac{J}{2n} \sum_{i,j=1}^n \phi_i \phi_j - \sum_{i=1}^n h_i \phi_i,
\end{align}  
where $J > 0$ is a coupling constant, and $h \in \mathbb{R}^\mathbb{N}$ is an external field. The state space on which this Hamiltonian is defined is continuous and unbounded as opposed to discrete and bounded like the Ising model. Since each spin interacts with all other spins, the interaction defined by this Hamiltonian is a long-range interaction, and it is often also referred to as a mean-field interaction. The uniform probability measure $\omega_n$ on the $n-1$-dimensional sphere with radius $\sqrt{n}$ is formally given by its density $\omega_n(d \phi)$ on $\mathbb{R}^n$ which has the formula
\begin{align*}
\omega_n (d  \phi) := \frac{\delta(|| \phi ||^2 - n)}{\int_{\mathbb{R}^n} d \phi \ \delta (|| \phi ||^2 - n)} d \phi,
\end{align*}
where $\delta (\cdot)$ is the Dirac-delta, $\| \cdot \|$ is the Euclidean norm, and $d \phi$ is the Lebesgue measure. We refer to the utilization of this uniform measure as the spherical constraint in this model. The Gibbs state $\mu_n^{\beta,J,h}$ of this model is the probability measure on $\mathbb{R}^n$ formally given by its density $\mu_n^{\beta,J,h} (d \phi)$ on $\mathbb{R}^n$ which has the formula
\begin{align} \label{fvgs}
\mu_n^{\beta,J,h} (d \phi) := \frac{e^{- \beta H_n^{J,h} [\phi]}}{Z_n(\beta,J,h)} \omega_n (d \phi),
\end{align}
where $\beta > 0$ is the inverse temperature, and $Z_n(\beta,J,h)$ is the partition function formally given by
\begin{align} \label{partition}
Z_n (\beta,J,h) := \int_{\mathbb{R}^n} \omega_n (d \phi) \ e^{- \beta H_n^{J,h} [\phi]} .
\end{align}
It is common to scale or set certain parameters to fixed values in these kinds of models. We insist on leaving the parameters without rescaling so that their contributions can be seen more transparently in results to come such as \Cref{max_point2} or the results presented in \Cref{phase_table}, where the ratio of the coupling constant $J$ and the, to be introduced, limiting sample standard deviation $m^\perp$ is of interest. 
\\
\\ 
In this paper, a majority of the practical calculations are first done using the formal calculation properties of the $\delta$-functions, and these calculations are then given rigorous proofs later on. The Gibbs state is rigorously redefined by its action on $f \in C_b (\mathbb{R}^n)$, where $C_b (\mathbb{R}^n)$ is the space of continuous bounded functions on $\mathbb{R}^n$, given by
\begin{align} \label{fvgs2}
\mu_n^{\beta,J,h} [f] := \frac{1}{Z_n (\beta,J,h)} \frac{2}{n^{\frac{n}{2} - 1}} \frac{1}{|\mathbb{S}^{n - 3}|} \int_{\mathbb{S}^{n-1}} d \Omega \  e^{\frac{\beta J}{2} \left( \sum_{i=1}^n \Omega_i \right)^2 + \beta \sqrt{n} \sum_{i=1}^n h_i \Omega_i} f \left( \sqrt{n} \Omega \right) ,
\end{align}
where
\begin{align} \label{partition1}
Z_n(\beta, J, h) :=  \frac{2}{n^{\frac{n}{2} - 1}} \frac{1}{|\mathbb{S}^{n - 3}|} \int_{\mathbb{S}^{n-1}} d \Omega \  e^{\frac{\beta J}{2} \left( \sum_{i=1}^n \Omega_i \right)^2 + \beta \sqrt{n} \sum_{i=1}^n h_i \Omega_i} ,
\end{align}
and $d \Omega$ is the integral over the angular part of the hyperspherical coordinates on $\mathbb{R}^n$. Note that this is a redefinition, and the $n$-dependent factors in this definition have been chosen for convenience. 
\\
\\
Note that we have explicitly defined the Gibbs state by its action on continuous bounded functions and we reserve the term expectation for real-valued random variables that will appear later on in this paper. We will encounter probabilistic objects such as probability measure-valued random variables and $\mathbb{R}^\mathbb{N}$-valued random variables which are defined on a common underlying probability space. For such objects, it is convenient to consider many of their properties as expectations with respect to this underlying probability measure, and this is why we want to have distinguished definitions of the Gibbs state and the distribution of a random variable although both are probability measures.  
\\
\\
For technical reasons, we will extend the underlying space on which $\mu_n^{\beta, J,h}$ acts on as a probability measure to $\mathbb{R}^\mathbb{N}$ by "tensoring on $0$" to the remaining $\mathbb{R}^{\mathbb{N} \setminus \{ 1,2,...,n\}}$ portion of the space. For the final time, we redefine $\mu_n^{\beta,J,h} := \mu_n^{\beta,J,h} \otimes \delta_0$, where $\delta_0$ is the Dirac measure on the point with all components $0$ in $\mathbb{R}^{\mathbb{N} \setminus \{ 1,2,...,n\}}$. The probability measure on $\mathbb{R}^\mathbb{N}$ constructed from probability measures on $\mathbb{R}^n$  by this method will be referred to as $0$-tensored versions. A function $f : \mathbb{R}^\mathbb{N} \to \mathbb{R}$ is said to be local if there exists a finite index set $I \subset \mathbb{N}$ and a function $f' : \mathbb{R}^I \to \mathbb{R}$ such that $f(x) = f' (\pi_I (x))$, where $\pi_I$ is the canonical projection from $\mathbb{R}^\mathbb{N}$ to $\mathbb{R}^I$. When dealing with local functions, we will typically refer to the representation function $f'$ as $f$, and the index set on which the function is local will be called $I$. The reason why this extension procedure is purely technical is that for a fixed local function and large enough $n$, the expectation of this local function is the same for the normal Gibbs state and the redefined extended Gibbs state.  
\\
\\
For a fixed $n$, such a probability measure $\mu_n^{\beta,J,h}$ is referred to as a finite volume Gibbs state(FVGS), and the collection of FVGS  will be denoted by $\mathcal{G}(\beta,J,h) := \{ \mu_n^{\beta,J,h} \}_{n \in \mathbb{N}}$. The entire collection $\mathcal{G} (\beta,J,h)$ is then a collection of probability measures on $\mathbb{R}^\mathbb{N}$.
\\
\\
The two principle objects of interest in this paper are the limiting free energy and the infinite volume Gibbs states(IVGS). We define them here explicitly due to their importance.
\begin{definition}[Limiting free energy] \label{lim_free}
The limiting free energy $f(\beta,J,h)$, when it exists, is given by 
\begin{align*} 
f(\beta,J,h) := \lim_{n \to \infty} \frac{1}{n} \ln Z_n (\beta,J,h) .
\end{align*}
\end{definition}
\noindent
Note that the limiting free energy is defined here with a different sign and unit convention than the physicist's limiting free energy. 
\\
\\
To define the IVGS, recall first that the ambient space $\mathbb{R}^\mathbb{N}$ can be equipped with the product topology to make it a topological space which is separable, metrizable, and complete with respect to this metric. Such a topological space is called a Polish space. Recall also that the space of probability measures $\mathcal{M}_1(\mathbb{R}^\mathbb{N})$ on $\mathbb{R}^\mathbb{N}$, by virtue of $\mathbb{R}^\mathbb{N}$ being a Polish space, can be equipped with a topology which is separable, metrizable by the Prokhorov metric $d$, and complete with respect to this metric, see \cite[Chapter 3]{Ethier1986}.
\begin{definition}[Infinite volume Gibbs states] The collection infinite volume Gibbs states $\mathcal{G}_\infty (\beta,J,h)$, when it is non-empty, is the collection of probability measures on $\mathbb{R}^\mathbb{N}$ given by
\begin{align*}
\mathcal{G}_\infty (\beta,J,h) := L (\mathcal{G}(\beta,J,h)),
\end{align*} 
where $L(\mathcal{G}(\beta,J,h))$ is the collection of limits points of the finite volume Gibbs states.   
\end{definition}
\noindent
In this definition, it is useful to use the metric space characterization of a limit point. A limit point of a set $A$ in a metric space $X$ is any point not in $A$ which can be obtained as the limit of a convergent subsequence of elements in the set.
\\
\\
Since mean-field models typically do not involve boundary conditions or local specifications in the same sense as the classical Gibbsian formalism for lattice spin systems, see \cite[Chapter 4]{Bovier2006}, the definition for the IVGS is less involved and does not initially require any further study of the convexity properties of the Gibbs states. This definition of the IVGS reflects the fact that a typical mean-field interaction is invariant under permutations of the underlying spins, and that the notion of increasing volumes in the thermodynamic limit can be replaced by increasing number of lattice sites, since the structure of the sequence of increasing volumes is irrelevant to the mean-field interaction. This is why the FVGS are labelled as a sequence of probability measures rather than a collection of probability measures indexed by subsets corresponding to volumes. 
\\
\\
The metrization of $\mathcal{M}_1(\mathbb{R}^\mathbb{N})$ by the Prokhorov metric $d$ is somewhat intractable for practical calculations. Instead, we will use the fact that collection of local bounded Lipschitz functions $\operatorname{LBL} (\mathbb{R}^\mathbb{N})$ from $\mathbb{R}^\mathbb{N}$ to $\mathbb{R}$ is convergence determining on $\mathcal{M}_1(\mathbb{R}^\mathbb{N})$, see \cite[Chapter 3]{Ethier1986} and \cite[Chapter 13]{Klenke2020}. A local function $f$ is Lipschitz if it is Lipschitz on $\mathbb{R}^I$ with the standard Euclidean norm. In terms of the Prokhorov metric, it follows that a sequence of probability measures $\{ \mu_n \}_{n \in \mathbb{N}}$ converges to a probability measure $\mu$ with respect to the Prokhorov metric if and only if $\mu_n[f] \to \mu[f]$ in the limit as $n \to \infty$ for any $f \in \operatorname{LBL} (\mathbb{R}^\mathbb{N})$. By using the collection $\operatorname{LBL}(\mathbb{R}^\mathbb{N})$, it follows that $\mu \in \mathcal{G}_\infty (\beta,J,h)$ if and only if there exists a subsequence $\{ n_k \}_{k \in \mathbb{N}}$ such that $\mu_{n_k}^{\beta,J,h} [f] \to \mu[f]$ for any $f \in \operatorname{LBL} (\mathbb{R}^\mathbb{N})$. Convergence with respect to the Prokhorov metric is often called weak convergence and we will sometimes write $\mu_n \to \mu$ weakly in  the limit as $n \to \infty$ by which we mean that $d(\mu_n, \mu) \to 0$ in the limit as $n \to \infty$. If it is clear from the context, we will omit "weakly" and simply say that $\mu_n \to \mu$ in the limit as $n \to \infty$.  
\\
\\
In the literature, the mean-field spherical model was introduced in \cite{Kastner2006}, and it corresponds to the model presented in this paper with the specific choice of external field $h$ in which all components are equal. In \cite{Kastner2006}, the authors computed the limiting free energy, with a different sign and unit convention, of the mean-field spherical model. In \cite{Koskinen2020}, the IVGS are indirectly identified for the mean-field spherical model. Historically mean-field models were introduced as simplified models of more complicated models which retain enough non-trivial features to be of use to study. For example, the Curie-Weiss model, which is the mean-field version of the classical Ising model, was introduced to exactly study a simplified model which still exhibits interesting thermodynamic phenomena such as phase transitions, anomalous thermodynamic fluctuations, etc. We refer to \cite{Friedli2017} for more details on the Curie-Weiss model and the classical Ising model. In addition, various features of certain generalizations of the classical Curie-Weiss model were studied in depth in \cite{Ellis1978}. From this perspective, the mean-field spherical model is the mean-field version of the continuous spin version of the Ising model known as the Berlin-Kac model introduced in \cite{Berlin1952}. 
\subsection{Deterministic inhomogeneous external field} \label{deterministic}
\noindent
We will now begin the exposition and presentation of the main results concerning the deterministic inhomogeneous external field.  Denote $m_n^h := (m_n^{h, \parallel}, m_n^{h, \perp}) \in \mathbb{R}^2$, where $m_n^{h,\parallel}$ is the finite sample mean and $m_n^{h,\perp}$ is the finite sample standard deviation of the external field $h$ given by
\begin{align} \label{sample_magnetization}
m_n^{h,\parallel} := \frac{1}{n} \sum_{i=1}^n h_i, \ m_n^{h,\perp} := \sqrt{\frac{1}{n} \sum_{i=1}^n h_i^2 - \left( m_n^{h, \parallel} \right)^2} . 
\end{align}  
The magnetization $M_n : \mathbb{R}^n \to \mathbb{R}$ is given by
\begin{align} \label{magnetization}
M_n[\phi] := \sum_{i=1}^n \phi_i = \left<1_n, \phi \right>, 
\end{align} 
where $1 \in \mathbb{R}^\mathbb{N}$ is the vector with all components $1$, $1_n := \pi_{\{ 1,2,...,n\}} (1)$, and $\left< \cdot, \cdot \right>$ is the Euclidean inner-product. Observe that the Hamiltonian can be written in the following form
\begin{align*}
H_n^{J,h} [\phi] = - \frac{J}{2 n} M_n[\phi]^2 - \sum_{i=1}^n h_i \phi_i = - \frac{J}{2 n} \left< 1_n, \phi \right>^2 - \left< h_n, \phi \right> ,
\end{align*}
where $h_n := \pi_{\{ 1,2,...,n\}} (h)$. Suppose now that $h$ satisfies $m_n^{h,\perp} \not = 0$, and let us consider the plane $W_n^h \subset \mathbb{R}^n$ spanned by the vectors $1_n$ and $h_n$. An orthonormal basis of $W_n^h$ is given by the span of the unit vectors $\{ w_{1,n}, w^h_{2,n} \}$ given by 
\begin{align} \label{basis}
w_{1,n} := \frac{1_n}{\sqrt{n}}, \ w^h_{2,n} := \frac{h_n - m_n^{h,\parallel} 1_n}{\sqrt{n} m_n^{h,\perp}} .
\end{align} 
The Hamiltonian can be written as
\begin{align} \label{hamilton_calc}
H_n^{J,h} [\phi] = - \frac{J}{2} \left< w_{1,n}, \phi \right>^2  - \sqrt{n} m_n^{h, \parallel} \left< w_{1,n}, \phi \right> - \sqrt{n} m_n^{h, \perp}\left< w^h_{2,n}, \phi \right> .
\end{align}
Let $\{ v_{k,n}^h \}_{k=3}^n$ be an orthonormal basis of $\left( W_n^h \right)^\perp$, where $\left( W_n^h \right)^\perp$ is the orthogonal complement of $W_n^h $. Let $O^h_n$ be the orthogonal change of coordinates $O^h_n : \mathbb{R}^n \to \mathbb{R}^n$ given by $\left( O^h_n(\phi) \right)_1 = \left< w_{1,n}, \phi \right>$, $\left( O^h_n(\phi) \right)_2 = \left< w^h_{2,n}, \phi \right>$, and $\left( O^h_n(\phi) \right)_k = \left< v_{k,n}^h, \phi \right>$ for $k=3,...,n$. Using this change of coordinates and the hyperspherical change of coordinates, formally, we have
\begin{align} \label{partition_formal}
\int_{\mathbb{R}^n} d \phi \ e^{- \beta H_n^{J,h} [\phi]} \delta \left( \lVert \phi \rVert^2 - n \right) &=  \frac{|\mathbb{S}^{n - 3}| n^{\frac{n}{2} - 1}}{2}  \int_{B(0,1)} dz \ e^{n \left( \frac{\beta J}{2} x^2 + \beta \left< m_n^h, z \right> \right)} (1 - || z ||^2)^{\frac{n - 4}{2}} ,
\end{align}
where $z = (x,y) \in \mathbb{R}^2$, $\lvert \mathbb{S}^{n - 3}\rvert$ is the integral over the angular coordinates of the $n-3$-dimensional unit sphere, and $B(0,1) \subset \mathbb{R}^2$ is the $2$-dimensional unit sphere. This formal calculation leading to \Cref{partition_formal} is given a rigorous proof in \Cref{can_rep}.
\\
\\
Let us remark that if $m_n^{h, \perp} = 0$, then it follows that $h_n$ is a vector with equal components and the calculation is the same as for the standard mean-field spherical model in \cite{Kastner2006}. This is why the standard mean-field spherical model concerns a homogeneous external field while the model present in the paper concerns an inhomogeneous external field. From here on out, whenever we refer to an external field, we mean an inhomogeneous external field unless otherwise stated. 
\subsubsection{Limiting free energy}  
\noindent
To continue, we state the first condition that the external field must satisfy as a definition. 
\begin{definition} \label{strongly_varying} An external field $h$ is strongly varying if 
\begin{align*}
\lim_{n \to \infty} m_n^h  = m := (m^\parallel, m^\perp) \in \mathbb{R} \times (0, \infty) .
\end{align*}
\end{definition}
\noindent
Subsequently, one could define the weakly varying external field as an external field where the limit $m \in \mathbb{R} \times \{ 0 \}$. In other words, the "strongly varying" part of the field is associated with a non-vanishing limiting sample standard deviation. Note also that the strongly varying condition ensures that $m_n^{h, \perp} \not = 0$ for large enough $n$. Note also that whenever an external field is strongly varying it is also thus inhomogeneous. From this point onwards, unless explicitly stated otherwise, we will assume that the external field is strongly varying.
\\
\\
Denote the exponential tilting function $\psi_n^{\beta,J,h}$ to be the function $\psi_n^{\beta,J,h} : B(0,1) \to \mathbb{R}$ given by
\begin{align} \label{tilting}
\psi_n^{\beta,J,h} (z) := \frac{\beta J}{2} x^2 + \beta \left<m_n^h,z \right> - \frac{1}{2} \ln (1 - || z ||^2) .
\end{align} 
By setting $f \equiv 1$ in \Cref{can_rep}, we see that the exponential tilting function is related to the partition function $Z_n(\beta,J,h)$ by the following formula
\begin{align}
Z_n(\beta,J,h) = \int_{B(0,1)} dz \ e^{2 \beta J x^2 + 4 \beta \left< m_n^h, z \right>} e^{(n - 4) \psi_n^{\beta,J,h} (z)} .
\end{align} 
Motivated by this representation, we denote the limiting exponential tilting function $\psi^{\beta,J,m}$ to be the function $\psi^{\beta,J,m} : B(0,1) \to \mathbb{R}$ given by
\begin{align} \label{limiting_tilting}
\psi^{\beta,J,m} (z) := \frac{\beta J}{2} x^2 + \beta \left<m,z \right> - \frac{1}{2} \ln (1 - || z ||^2) .
\end{align}
For the exponential tilting functions, it follows that
\begin{align*}
\left| \psi_n^{\beta, J, h} (z) - \psi^{\beta,J,m} (z) \right| \leq \beta || m_n^h - m || ,
\end{align*}
from which it follows that
\begin{align*}
\lim_{n \to \infty} \sup_{z \in B(0,1)} \left| \psi_n^{\beta, J, h} (z) - \psi^{\beta,J,m} (z) \right| = 0 .
\end{align*} 
By using this uniform convergence, we have the first main result. 
\begin{theorem} \label{argmax1} Let $h$ be a strongly varying external field.
\\
\\
It follows that
\begin{align*}
f (\beta, J, h) := \lim_{n \to \infty} \frac{1}{n} \ln Z_n (\beta, J, h) = \sup_{z \in B(0,1)} \psi^{\beta,J,m}(z) .
\end{align*}
\end{theorem}
\noindent
The proof of this result, see \Cref{free_tight}, is a routine large deviations calculation once we show that set of global maximizing points $M^*(\beta,J,m)$ of $\psi^{\beta,J,m}$ is compact and non-empty, see \Cref{argmax2} for the proof.  
\\
\\
One can see from this theorem that if we allowed for a weakly varying external field, i.e. we have $m^\perp = 0$, then we would recover the same free energy as for the standard mean-field spherical model presented in \cite{Kastner2006}. With respect to the manipulation of the Hamiltonian in \Cref{hamilton_calc}, when one applies a homogeneous external field, which is characterized by the external field $h$ having equal components, we see that there is a single magnetization "component" which corresponds to the projection of $\phi$ along the normalized unit vector $w_{1,n}$. If the external field is not homogeneous, there exists a second magnetization component corresponding to the projection of $\phi$ along the normalized unit vector $w^h_{2,n}$. Subject to the strongly varying condition, this perpendicular magnetization component is non-vanishing in the limit and it is the distinguishing feature between the standard mean-field spherical model in a homogeneous external field and the mean-field spherical model in a strongly varying field.   
\subsubsection{Partial classification of infinite volume Gibbs states}  
\noindent
Our next main result concerns a partial classification of the IVGS. With reference to \Cref{can_rep}, we begin by defining the mixture probability measure $\rho_n^{\beta,J,h}$ on $B(0,1)$ by its density
\begin{align} \label{mixture}
\rho_n^{\beta,J,h} (dz) =  \frac{e^{2 \beta J x^2 + 4 \beta \left< m_n^h, z \right>} e^{(n - 4) \psi_n^{\beta,J,h} (z)}}{Z_n(\beta,J,h)} dz ,
\end{align}
where $dz$ is the Lebesgue measure on $B(0,1)$. Next, we define the "microcanonical" probability measure on $\mathbb{R}^\mathbb{N}$ as the $0$-tensored version of the probability measure given via its action on continuous bounded functions $f \in C_b (\mathbb{R}^n)$ by
\begin{align} \label{microcanonical}
f \mapsto \nu_n^{z,h}[f] :=  \frac{1}{|\mathbb{S}^{n - 3}|}\int_{\mathbb{S}^{n - 3}} d \Omega \ f \left( \sqrt{n} x w_{1,n} + \sqrt{n} y w^h_{2,n} + \sqrt{1 - || z ||^2}\sqrt{n} \sum_{j=3}^n \Omega_j v^h_{j,n} \right) .
\end{align}
As a direct application of \Cref{can_rep} along with the definitions given by \Cref{mixture} and \Cref{microcanonical}, we have the following central representation result for the FVGS.
\begin{lemma} \label{can_rep2} Let $h$ be a strongly varying external field.
\\
\\
It follows that
\begin{align} 
\mu^{\beta, J, h}_n =\int_{B(0,1)} \rho_n^{\beta, J, h}(dz) \ \nu_n^{z,h} .
\end{align}
\end{lemma}
\noindent
By using a similar calculation to the one used for the proof of the form of the limiting free energy, we show that the collection of mixture probability measures $\{ \rho_n^{\beta,J,h}\}_{n \in \mathbb{N}}$ is uniformly tight and its limit points are probability measures supported on the set $M^*(\beta,J,m)$. 
\begin{lemma} \label{mix_supp} Let $h$ be a strongly varying external field.
\\
\\
It follows that the collection of probability measures $\{ \rho^{\beta,J,h}_n\}_{n \in \mathbb{N}}$ is uniformly tight and 
\begin{align*}
L \left( \{ \rho_n^{\beta, J, h} \}_{n \in \mathbb{N}} \right) \subset \left\{ \rho \in \mathcal{M}_1 (B(0,1)) : \operatorname{supp} (\rho) \subset M^*(\beta,J,m)\right \} .
\end{align*} 
\end{lemma}
\noindent
For the full proof, see \Cref{free_tight}. Note that the definition of the support we are using is $\operatorname{supp}(\rho) := \{ x \in B(0,1) : \forall \delta > 0, \ \rho(B(x, \delta)) > 0 \}$.
\\
\\
For this model, we show that the structure of $M^*(\beta,J,m)$ is simple and completely characterizable. It is either a set with one element or two elements depending on the parameters of the model, see \Cref{max_point1} and \Cref{max_point2} for the proofs.  For the parameter range where $m^\parallel = 0$ and $m^\perp < J$ simultaneously, the transition from a single element to two elements is marked by a critical inverse temperature $\beta_c := \frac{J}{(J - m^\perp)(J + m^\perp)}$, such that the set consists of a single element when $\beta \leq \beta_c$ and it consists of two elements when $\beta > \beta_c$. When the mean is non-vanishing i.e. $m^\parallel \not = 0$, the set consists of a single element $z^*$ such that $x^* > 0$ if $m^\parallel > 0$ and $x^* < 0$ if $m^\parallel < 0$. For the other parameter ranges, the set consists of a single element $z^0$ such that $x^0 = 0$.  
\newpage
\noindent
We summarize these results in the following table.
\begin{table}[h!]
\begin{center}
\caption{Parametric ranges for the structure of the set of global maximizing points of the limiting exponential tilting function for the strongly varying deterministic inhomogeneous external field}
\label{phase_table}
\begin{tabular}{ |c|c|c|c| } 
\hline
Mean & Standard deviation and coupling constant & Inverse temperature & $M^*(\beta,J,m)$ \\
\hline
$m^\parallel \not = 0$ & $J > 0$ & $\beta > 0$ & $\{ z^*  \}$ \\
$m^\parallel = 0$ & $m^\perp \geq J$ & $\beta > 0$ & $\{ z^0 \}$ \\
$m^\parallel = 0$ & $m^\perp < J$ & $\beta \leq \beta_c$ & $\{ z^0 \}$ \\
$m^\parallel = 0$ & $m^\perp < J$ & $\beta > \beta_c$ & $\{ z^+, z^- \}$ \\
\hline
\end{tabular}
\end{center}
\end{table}
\\
Let us now classify these parameter ranges in the following way. We will say that we are in the pure state(PS) parameter range if the set $M^*(\beta,J,m)$ consists of a single element. This parameter range can be deduced from the first three rows of the above table. We will say that we are in the mixed state(MS) parameter range if the set $M^*(\beta,J,m)$ consists of two elements. This parameter range can be deduced from the last row of the above table. We will always refer to the two elements of this set as $z^+$ and $z^-$, where $x^+ > 0$.  
\\
\\
As for the probability measure $\nu_n^{z,h}$, we show that it satisfies a type of uniform convergence in the variable $z$ to a sufficiently regular probability measure on $\mathbb{R}^\mathbb{N}$ in the limit as $n \to \infty$. To begin with,  denote $\eta$ to be the probability measure on $\mathbb{R}^\mathbb{N}$ given via its action on $f \in \operatorname{LBL} (\mathbb{R}^\mathbb{N})$ by
\begin{align} \label{infinity_gaussian}
f \mapsto \eta[f] := \left( \int_{\mathbb{R}^{I}} d \phi \ e^{- \frac{|| \phi ||^2}{2}}\right)^{-1} \int_{\mathbb{R}^{I}} d \phi \ e^{- \frac{|| \phi ||^2}{2}} f (\phi) .
\end{align}
where $I \subset \mathbb{N}$ is the local index set of $f$. This probability measure on $\mathbb{R}^\mathbb{N}$ is the countable dimensional version of the standard finite dimensional Gaussian probability measure. 
\\
\\
We will need the projection $P_n^h : \mathbb{R}^n \to \mathbb{R}^n$ to the subspace $W_n^h$ given by the formula
\begin{align*}
P_n^h (\phi) = \left< \phi, w_{1,n}\right> w_{1,n} + \left< \phi, w^h_{2,n}\right> w^h_{2,n} .
\end{align*} 
Let $T_n^{z,h} : \mathbb{R}^\mathbb{N} \to \mathbb{R}^{n} \times \mathbb{R}^{\mathbb{N} \setminus [n]} \subset \mathbb{R}^\mathbb{N}$ be the transport map given by 
\begin{align*}
(T_n^{z,h}(\phi))_1 = \sqrt{n} x w_{1,n} + \sqrt{n} y w^h_{2,n} + \sqrt{1 - || z ||^2} \sqrt{n} \frac{\phi - P^h_n (\phi)}{|| \phi - P^h_n (\phi)||}
\end{align*}
and $(T_n^{z,h} (\phi))_2 = 0$, where $z = (x,y) \in B(0,1)$. For the given probability measures and transport maps, we show that $\nu_n^{z,h} = { T_{n}^{z,h} }_* \eta$, see \Cref{transport1}, where ${ T_{n}^{z,h} }_* \eta$ is the pushforward measure of $\eta$ by $T_n^{z,h}$.
\\
\\
Continuing, note that
\begin{align*}
\lim_{n \to \infty} \sqrt{n} \pi_I (w_{1,n}) = 1_I, \ \lim_{n \to \infty} \sqrt{n} \pi_I (w^h_{2,n}) = \lim_{n \to \infty} \frac{h_I - m_n^{h, \parallel} 1_I}{\sqrt{m_n^{h, \perp}}} = \frac{h_I - m^\parallel 1_I}{\sqrt{m^\perp}} ,
\end{align*}  
where $I \subset \mathbb{N}$ is any finite index set. Let $T_\infty^{z,h} : \mathbb{R}^\mathbb{N} \to \mathbb{R}^\mathbb{N}$ be the transport map given by
\begin{align} \label{gaussian_pushforward}
T^{z,h} (\phi) := \sqrt{1 - || z ||^2} \phi + x 1 + y \frac{h - m^\parallel 1}{m^\perp} ,
\end{align}
where $z = (x,y) \in B(0,1)$ and $m^\perp \not = 0$. Using this transport map, consider the probability measure $\nu_\infty^{z,h}$ on $\mathbb{R}^\mathbb{N}$ given by $\nu_\infty^{z,h} :={T^{z,h}_\infty}_* \eta$. We show the following uniform convergence result concerning the convergence of $\nu_n^{z,h}$ to $\nu_\infty^{z,h}$.
\begin{lemma} \label{uni_conv} Let $h$ be a strongly varying external field.
\\
\\
It follows that
\begin{align*}
\lim_{n \to \infty} \sup_{z \in B(0,1)} \left| \nu_n^{z,h} [f] - \nu_\infty^{z,h} [f] \right| = 0
\end{align*}
for any $f \in \operatorname{LBL} (\mathbb{R}^\mathbb{N})$.
\end{lemma}
\noindent
For the full proof, see \Cref{uni_conv_sec}. The proof of the uniform convergence result relies on specific asymptotic properties of the spherical constraint. The main technical difficulties involve the lack of symmetries such as permutation invariance or translation invariance which are present in the MFS model.
\\
\\
By combining together the limiting point result \Cref{mix_supp} for the collection of mixture probability measures $\{ \rho_n^{\beta,J,h} \}_{n \in \mathbb{N}}$, the uniform convergence result from \Cref{uni_conv}, and repeated use of Prokhorov's theorem, see \cite[Chapter 3]{Ethier1986}, we have the following partial classification of the IVGS.
\begin{theorem} \label{main_result1}
Let $h$ be a strongly varying external field.
\begin{enumerate}
\item For the pure state parameter range, we have 
\begin{align*}
\mathcal{G}_\infty (\beta,J,h) = \left\{ \nu_\infty^{z^*,h}  \right\} .
\end{align*}
\item For the mixed state parameter range, we have
\begin{align*}
\mathcal{G}_\infty (\beta,J,h) \subset \operatorname{conv} \left(  \nu_\infty^{z^+,h}, \  \nu_\infty^{z^-,h} \right) := \{ \lambda \nu_\infty^{z^+,h} + (1 - \lambda) \nu_\infty^{z^-,h} : \lambda \in [0,1] \} .
\end{align*}
\end{enumerate}
\end{theorem}  
\noindent
For the full proofs, see \Cref{single_point} and \Cref{class2}. 
\\
\\
In the PS parameter range, we see that there exists a unique IVGS. Since the proof of this result relies on Prokhorov's theorem, we do not obtain a complete classification of the collection of IVGS for the MS parameter range.  Indeed, it is possible that there could still exist a single unique limit, but this method of proof does not provide it.
\\
\\
Let us emphasize that the method of proof used in \Cref{main_result1} relies heavily on Prokhorov's theorem, and the uniqueness of the IVGS in the PS parameter range is due to the fact that the only probability measure supported by a single point is the Dirac measure at that point. In general, when $M^*(\beta,J,m)$ does not consist of a single point, to fully characterize the limit points of the mixtures measures as in \Cref{mix_supp}, one would need to characterize the probability measures supported by $M^*(\beta,J,m)$ that can be obtained as weakly convergent subsequence of the mixtures measures $\{ \rho_n^{\beta,J,h}\}_{n \in \mathbb{N}}$.  Later on in this paper, we prove two theorems concerning this phenomenon. For the deterministic inhomogeneous external field, subject to additional assumptions to the magnetization vectors $\{ m_n^h \}_{n \in \mathbb{N}}$, we show that in the MS parameter range, the IVGS can still consist of a single unique point, see \Cref{conv_class}. For the random external field, we show that all convex combinations of the pure states can be obtained, see \Cref{main_result3}. These results show that in order to give results concerning the opposite inclusion or characterization of weakly converging subsequences of the mixture measures, one must control the, to be introduced in \Cref{weights}, relative weights associated with the FVGS.  We will discuss these results in depth once they are proven.
\subsubsection{Full classification of the infinite volume Gibbs states}  
\noindent
In the MS parameter range, there are two global maximizing points $z^\pm \in M^*(\beta,J,m)$. They are related by the fact that $x^- = - x^+ < 0$ and $y^- = y^+$. This suggests studying the mixture probability measure $\rho_n^{\beta,J,h}$ conditioned to the positive and negative quadrants $B_{+} (0,1) := B(0,1) \cap ((0, \infty) \times \mathbb{R})$ and $B_- (0,1) := B(0,1) \cap ((-\infty, 0) \times \mathbb{R})$, which is equivalent to conditioning the FVGS on the set where the magnetization is positive for the positive quadrant and negative for the negative quadrant. The conditioned FVGS $\mu_n^{\beta,J,h, \pm}$ act on $f \in C_b (\mathbb{R}^\mathbb{N})$ by
\begin{align} \label{fvgs_conditioned}
\mu_n^{\beta,J,h,\pm}[f] := \frac{\mu_n^{\beta,J,h} [ \mathbbm{1}(\pm M_n > 0) f]}{\mu_n^{\beta,J,h} [ \mathbbm{1}(\pm M_n > 0)]} .
\end{align} 
To accompany the conditioned FVGS, the weights $W_n^{\beta,J,h, \pm}$ are given by
\begin{align} \label{weights}
W_n^{\beta,J,h, \pm} := \mu_n^{\beta,J,h} [ \mathbbm{1}(\pm M_n > 0)] .
\end{align}
This conditioning yields a representation of the form
\begin{align} \label{condition_rep}
\mu_n^{\beta,J,h} = W_n^{\beta,J,h,+} \mu_n^{\beta,J,h,+} + (1- W_n^{\beta,J,h,+}) \mu_n^{\beta,J,h,-} .
\end{align}
By reusing the proof of \Cref{main_result1} for the PS parameter range, we show that the probability measures $\nu_n^{\beta,J,h, \pm}$ converge weakly in the limit as $n \to \infty$.
\begin{lemma} \label{condition_conv} Let $h$ be a strongly varying external field.
\\
\\
For the mixed state parameter range, it follows that
\begin{align*}
\lim_{n \to \infty} \mu_n^{\beta,J,h, \pm} = \nu_{\infty}^{z^\pm, h} .
\end{align*}
\end{lemma}
\noindent
For the full proof, see \Cref{asymp_analysis}. 
\\
\\
It thus follows that the convergence properties of the weights $W_n^{\beta,J,h, \pm}$ determine the limiting structure of the FVGS. By rearranging the form of the weights, we see that
\begin{align} \label{weight_rep}
W_n^{\beta, J, h, +} = \frac{1}{1 + \frac{\int_{B_-(0,1)} dz \ e^{2 \beta J x^2 + 4 \beta \left< m_n^h, z \right>} e^{(n - 4)\psi^{\beta,J,m}_n(z)}}{\int_{B_+(0,1)} dz \ e^{2 \beta J x^2 + 4 \beta \left< m_n^h, z \right>} e^{(n - 4)\psi^{\beta,J,m}_n(z)}}} .
\end{align} 
To resolve the convergence properties, we introduce a sequence of local maximizing points $\{ z_n^* \}_{n \in \mathbb{N}}$ such that each $z_n^*$ is a local maximizing point of $\psi_n^{\beta,J,h}$, $z_n^*$ satisfies the critical point equation $\nabla \psi_n^{\beta,J,h} (z_n^*) = 0$, and $z_n^* \to z^*$ in the limit as $n \to \infty$, where $z^*$ is a global maximizing point of $\psi^{\beta,J,m}$, see \Cref{max_seq} for the construction. From the proof presented for \Cref{max_point2}, we know that the Hessian $H[\psi^{\beta,J,m}]$ of $\psi^{\beta,J,m}$ is negative definite at the points $z^\pm$. With these observations, we show that
\begin{align}
\lim_{n \to \infty} \frac{n \int_{B_{\pm} (0,1)} dz \ e^{2 \beta J x^2 + 4 \beta \left< m_n^h, z \right>} e^{(n - 4)\psi^{\beta,J,m}_n(z)}}{e^{n \psi^{\beta,J,m}_n(z_n^{\pm})}} = \frac{1}{(1 - || z^\pm ||^2)^2}\int_{\mathbb{R}^2} dz \ e^{\frac{1}{2} \left< z, H[\psi^{\beta,J,m}](z^\pm) z \right>} ,
\end{align}
where $\{ z_n^\pm\}_{n \in \mathbb{N}}$ is the collection of local maximizing points in their own respective quadrants of $B(0,1)$. The proof, see \Cref{laplace_asymp}, is essentially a modification of a similar proof for Laplace-type integrals which can be found in \cite[Chapter 2]{Wong2001}. The idea and method of constructing sequences of local maximizing points in this way is a model specific adaptation of the same method presented in \cite{Matos1991} \cite{Matos1992}.
\\
\\
Let us also remark that all of the proofs that are presented in \Cref{asymp_analysis} use the same notation and sequence of local maximizing points $\{ z_n^* \}_{n \in \mathbb{N}}$, and that the proofs and techniques presented here are only valid when the Hessian of $\psi^{\beta,J,m}$ at $z^*$ is negative definite. This point is also emphasized below the proof of \Cref{max_seq}. 
\\
\\
In the following exposition, we will present a series of results concerning the rates of convergence of various quantities present in these calculations. We will use the symbol $\approx$ to imply that the results hold in the large $n$ limit with a suitable error term for the desired application. In this notation, the weights satisfy
\begin{align*}
W_n^{\beta,J,h,+} \approx \frac{1}{1 + e^{n (\psi_n^{\beta,J,h} (z_n^-) - \psi_n^{\beta,J,h} (z_n^+))}} .
\end{align*} 
Using the critical point equations, we show that the local maximizing points satisfy
\begin{align*}
z_n^\pm - z^\pm \approx - \beta H[\psi^{\beta,J,m}]^{-1} (z^\pm) (m_n^h - m) ,
\end{align*}
see \Cref{diff_eq} for the proof. As a direct application of this result, we show that the exponential tilting functions evaluated at the local maximizing points satisfy
\begin{align*}
\psi_n^{\beta,J,h} (z_n^\pm) - \psi^{\beta,J,m} (z^\pm) \approx \beta \left< m_n^h - m, z^\pm \right> - \frac{\beta^2}{2} \left< m_n^h - m, H[\psi^{\beta,J,m}]^{-1} (z^\pm) (m_n^h - m) \right>,
\end{align*}
see \Cref{diff_eq2} for the proof. Finally, using the fact that $x^- = - x^+$, $y^+ = y^-$, and $\psi^{\beta,J,m} (z^+) = \psi^{\beta,J,m} (z^-)$, and the previous result, we show that
\begin{align*}
\psi_n^{\beta,J,h} (z_n^-) - \psi_n^{\beta,J,h} (z_n^+) \approx - 2 \beta (m_n^{h, \parallel} - m^\parallel) x^+ .
\end{align*}
By combining all of these results together, we can compute the limit of the weights by specifying the rates of convergence of the sample mean and sample standard deviation.
\begin{lemma} \label{weight_asymp}
Let $h$ be a strongly varying external field, and suppose that we are in the mixed state parameter range. 
\\
\\
Suppose there is $\delta \in [0, \infty)$ such that $n^\delta (m_n^h - m) \to \gamma := (\gamma^{\parallel}, \gamma^\perp) \in  \mathbb{R}^2$ in the limit as $n \to \infty$.
\begin{enumerate}
\item If $\delta \in [0,1)$ and $\gamma^{\parallel} \not = 0$, it follows that
\begin{align*}
\lim_{n \to \infty} W_n^{\beta,J,h,+} = \mathbbm{1}(\gamma^\parallel > 0) .
\end{align*}
\item If $\delta = 1$, it follows that
\begin{align*}
\lim_{n \to \infty} W_n^{\beta,J,h,+} = \frac{1}{1 + e^{- 2 \beta x^+ \gamma^\parallel}} .
\end{align*}
\item If $\delta \in (1, \infty)$, it follows that
\begin{align*}
\lim_{n \to \infty} W_n^{\beta,J,h,+} = \frac{1}{2} .
\end{align*}
\end{enumerate}
\end{lemma}
\noindent
For the full proof, see \Cref{asymp_analysis}.
\\
\\
Combining together the conditioned representation from \Cref{condition_rep}, the weak convergence of the conditioned probability measures from \Cref{condition_conv}, and the asymptotics of the weights from \Cref{weight_asymp}, we present the full classification of the IVGS given sublinear, linear, and superlinear rates of convergence of $m_n^h$ to $m$ in the limit as $n \to \infty$.
\begin{theorem} \label{conv_class}
Let $h$ be a strongly varying external field, and suppose that we are in the mixed state parameter range. 
\\
\\
Suppose there is $\delta \in [0, \infty)$ such that $n^\delta (m_n^h - m) \to \gamma := (\gamma^{\parallel}, \gamma^\perp) \in  \mathbb{R}^2$ in the limit as $n \to \infty$.
\begin{enumerate}
\item If $\delta \in [0,1)$ and $\gamma^{\parallel} \not = 0$, it follows that
\begin{align*}
\lim_{n \to \infty} \mu_n^{\beta, J, h} = \mathbbm{1}(\gamma^\parallel > 0) \nu_\infty^{z^+,h} +  \mathbbm{1}(\gamma^\parallel < 0) \nu_\infty^{z^-,h}.
\end{align*}
\item If $\delta = 1$, it follows that
\begin{align*}
\lim_{n \to \infty} \mu_n^{\beta, J, h} = \frac{1}{1 + e^{- 2 \beta x^+ \gamma^\parallel}} \nu_\infty^{z^+,h} + \frac{1}{1 + e^{ 2 \beta x^+ \gamma^\parallel}} \nu_\infty^{z^-,h} .
\end{align*}
\item If $\delta \in (1, \infty)$, it follows that
\begin{align*}
\lim_{n \to \infty} \mu_n^{\beta, J, h} = \frac{1}{2} \nu_\infty^{z^+,h} + \frac{1}{2}\nu_\infty^{z^-,h} .
\end{align*}
\end{enumerate}
\end{theorem}
\noindent
This result improves on \Cref{main_result1} for the mixed state parameter range. In particular, this result shows that when the external field is strongly varying, one can construct any IVGS by a specific choice of the rate of convergence of the sample mean and sample standard deviation. Later on in this paper, in \Cref{main_result3} and \Cref{aw_conv1}, we will give concrete probabilistic examples which apply \Cref{conv_class} for $\delta= \frac{1}{2}$ and $\delta = 1$ for some $\gamma$. However, we were unable to find a general way to construct deterministic inhomogeneous external fields which would realize the other possible asymptotic rates presented in \Cref{conv_class}. In general, this problem concerns the simultaneous control of the Cesàro sums of the sequences $\{ h_i \}_{i \in \mathbb{N}}$ and $\{ h_i^2 \}_{i \in \mathbb{N}}$, which seems difficult. A result in this direction for the RFCW model is the so-called generalized quasi-average method utilized in \cite{Matos1992}, but this method, from our perspective, involves a perturbed external field rather than a fixed external field.
\subsubsection{Summary and remarks}  
\noindent
Let us now summarize these results for the strongly varying external field. When the parallel magnetization component is non-vanishing i.e. $m^\parallel \not = 0$, irrespective of all other details of the model, there is a single unique IVGS. When the parallel magnetization component is vanishing i.e. $m^\parallel = 0$, and the perpendicular magnetization component is large enough i.e. $m^\perp \geq J$, there is a single unique IVGS. When the parallel magnetization component is vanishing, the perpendicular magnetization component is small enough i.e. $m^\perp < J$, and the inverse temperature is small enough i.e. $\beta \leq \beta_c$, there is a single unique IVGS. Finally, when the parallel magnetization component is vanishing, the perpendicular magnetization component is small enough, and the inverse temperature is large enough i.e. $\beta > \beta_c$, the IVGS can be realized as any convex combination of the pure states $\nu_\infty^{z^+,h}$ and $\nu_\infty^{z^-,h}$ subject to the sublinear, linear, and superlinear rates of convergence assumptions for the convergence of the sample mean and sample standard deviation. 
\\
\\
To our knowledge, the results concerning the limiting free energy, classification of IVGS, and rate of convergence analysis are novel contributions in the literature for this specific model, and, in general, this level of detail and specification is rare even for similar models. The last point about details and specification will be discussed further for the random external field. 
\\
\\
For a deterministic inhomogeneous external field, a result in the literature in the same spirit is the classification of the IVGS for the classical Curie-Weiss model presented in \cite{Brankov1986}. The major difference to our work is that the intent of this paper is to realize convex combinations of pure states in the standard Curie-Weiss model by perturbing the Curie-Weiss Hamiltonian with a small symmetry breaking external field. This would be similar to studying the weakly varying case for our model, which we explicitly exclude. Very briefly, the method of solution for this model involves using the Hubbard-Stratonovich transform, to write the perturbed finite volume Gibbs states as a mixture of product states. In our case, we do not, and cannot, in some sense, apply the Hubbard-Stratonivich transform, nor is there an immediate product structure. This is why one needs the uniform convergence lemma, which is \Cref{uni_conv} in the paper, used in the proof of one of the main results. In \cite{Koskinen2020}, the authors were able to utilize permutation invariance of the standard mean-field spherical model for a relatively simple proof using Wasserstein distances for weak convergence. In this work, due to the permutation invariance breaking external field, we needed a different technique which is the uniform convergence lemma.
\\
\\
Let us also remark that our paper does not make use of the method of steepest descent utilized by the authors in the original work which introduced the spherical model \cite{Berlin1952}. We only mention this since there are several instances of the utilization of this method for spherical models, yet we opted for a direct approach since it was possible. The method of steepest descent has been applied to study the spherical model in a specific deterministic non-homogeneous external field in \cite{Patrick1993} and the limiting Gibbs states of spherical model with a small homogeneous external field in \cite{Brankov1987}.
\subsection{Random external field} \label{random}
\noindent
Next, we begin the presentation and specification our results when the deterministic inhomogeneous external field is replaced by a random external field. A measurable map $h : (\Omega, \mathcal{F}, \mathbb{P}) \to (\mathbb{R}^\mathbb{N}, \mathcal{B}(\mathbb{R}^\mathbb{N}))$ is said to be a random external field, where $(\Omega, \mathcal{F}, \mathbb{P})$ is a probability triple, and $(\mathbb{R}^\mathbb{N}, \mathcal{B}(\mathbb{R}^\mathbb{N}))$ is a measurable space, where $\mathcal{B}(\mathbb{R}^\mathbb{N})$ is the Borel $\sigma$-algebra associated with the product topology on $\mathbb{R}^\mathbb{N}$. Since the map $h \mapsto \mu_n^{\beta,J,h}$ is continuous and thus measurable, it follows that $\omega \mapsto h(\omega) \mapsto \mu_n^{\beta,J,h(\omega)}$ is also measurable, and thus $\mu_n^{\beta,J,h}$ interpreted as a probability measure-valued random variable $\mu_n^{\beta,J,h(\cdot)} : (\Omega, \mathcal{F}, \mathbb{P}) \to (\mathcal{M}_1 (\mathbb{R}^\mathbb{N}), \mathcal{B}(\mathbb{R}^\mathbb{N}))$ is a random probability measure, see \Cref{cont} for this justification. The collection of FVGS is then a collection of probability measure-valued random variables, and we are interested in studying the limiting properties of this collection subject to additional assumptions to the random external field.
\\
\\
Let us also briefly remark on the distinction of convergence in distribution and weak convergence of random probability measures. If $\mathcal{S}$ is a Polish space and $\{ X_n \}_{n \in \mathbb{N}}$ and $X$ are $\mathcal{S}$-valued random variables, we say that $X_n$ converges to $X$ in distribution if the probability distributions of $X_n$ converge weakly to the probability distribution of $X$. For random probability measures $\{ \mu_n \}_{n \in \mathbb{N}}$ and $\mu$, when we say that $\mu_n$ converges to $\mu$ in distribution, we mean it in the sense that we just explained. If $\mu_n$ converges to $\mu$ almost surely, we mean that $d(\mu_n, \mu) \to 0$ in the limit as $n \to \infty$ almost surely, which is equivalent to saying that $\mu_n$ converges weakly to $\mu$ almost surely. We will try to stay consistent with this terminology so that weak convergence is  reserved for probability measures and convergence in distribution is reserved for random variables.
\\
\\
For the following assumptions to the random external field, we will need the concept of possible values of a random walk. Let $\{ X_i  \}_{i \in \mathbb{N}}$ be a collection of independent identically $X$-distributed $\mathbb{R}^d$-valued random variables. Denote $\{ S_n' \}_{n \in \mathbb{N}}$ to be the centred random walk with step length $X - \mathbb{E} X$ given by $S_n' := \sum_{i=1}^n (X_i - \mathbb{E} X_i)$. We say that a point $x \in \mathbb{R}^d$ is a possible value of $S_n'$ if for any $\varepsilon > 0$ there exists $n \in \mathbb{N}$ such that $\mathbb{P} (|| S_n' - x || < \varepsilon) > 0$. Denote the collection of possible values by $P$. We say that a point $x \in \mathbb{R}^d$ is a recurrent value of $S_n'$ if for any $\varepsilon > 0$, we have $\mathbb{P} (|| S_n' - x|| < \varepsilon \text{ infinitely often}) = 1$.
\\
\\
We present the following further assumptions for the random external field.
\begin{assumption} \
\begin{itemize} 
\item[(A1)] The components of $h$ are independent $h_0$-distributed real-valued random variables such that $\mathbb{E} h_0^2 < \infty$ and $\mathbb{V} h_0 > 0$, where $\mathbb{V} h_0 := \mathbb{E} h_0^2 - \left( \mathbb{E} h_0 \right)^2$
\item[(A2)] The random variable $h_0$ satisfies $\mathbb{E} h_0^4 < \infty$ and $\mathbb{V} h_0^2 > 0$
\item[(A3)] The set of possible values $P$ of the centred random walk with step length $(h_0 - \mathbb{E} h_0, h_0^2 - \mathbb{E} h_0^2)$ satisfies $\pi_1 (P) = \mathbb{R}$, where $\pi_1 (\cdot)$ is the canonical projection to the first coordinate
\item[(A4)] The random variable $h_0$ satisfies $\mathbb{E} h_0^{4 + \xi}$ for some $\xi > 0$
\end{itemize}
\end{assumption}
\noindent
Note that the moment conditions of (A2) imply the moment conditions of (A1). For the rest of this paper, we will denote $\{ S_n \}_{n \in \mathbb{N}}$ to be the random walk with step length $(h_0, h_0^2)$ and $\{ S_n'\}_{n \in \mathbb{N}}$ to be the centred random walk with step length $(h_0 - \mathbb{E} h_0, h_0^2 - \mathbb{E} h_0^2)$.
\subsubsection{Self-averaging of the limiting free energy} 
\noindent
In terms of the random walk $\{ S_n \}_{n \in \mathbb{N}}$, we can write the vector $m_n^h$ as
\begin{align*}
m^{h,\parallel}_n = \frac{{(S_n)}_1}{n}, \ m_n^{h, \perp} = \sqrt{\frac{{(S_n)}_2}{n} - \left( \frac{{(S_n)}_1}{n}\right)^2} .
\end{align*}
From this simple observation, as an application of the strong law of large numbers for the random walk $\{ S_n \}_{n \in \mathbb{N}}$, we show that the sequence of vectors $\{ m_n^h \}_{n \in \mathbb{N}}$ satisfies a strong law of large numbers.
\begin{lemma} Let $h$ be a random external field which satisfies (A1).
\\
\\
It follows that
\begin{align*}
\lim_{n \to \infty} m_n^h = \left( \mathbb{E} h_0, \sqrt{\mathbb{V} h_0}\right) := m ,
\end{align*}
almost surely.
\end{lemma}
\noindent
For the proof, see \Cref{limit_theorems}.
\\
\\
We see that the condition $\mathbb{V} h_0 > 0$ ensures that the random external field $h$ is strongly varying almost surely. Subject to assumption (A1), it immediately follows that the limiting free energy is given by
\begin{align*}
f(\beta,J,h) = \lim_{n \to \infty} \frac{1}{n} \ln Z_n (\beta,J,h) = \sup_{z \in B(0,1)} \psi^{\beta,J,m} (z) 
\end{align*} 
almost surely. Note that although the partition functions $Z_n (\beta,J,h)$ are random variables, the limiting free energy is a deterministic quantity. We show that the collection of finite volume free energies is uniformly integrable and thus we have the following result concerning the self-averaging of the limiting free energy.
\begin{theorem} \label{self_average}
Let $h$ be a random external field which satisfies (A1).
\\
\\
It follows that
\begin{align*}
\lim_{n \to \infty} \frac{1}{n} \mathbb{E} \ln Z_n (\beta,J,h) = \lim_{n \to \infty} \frac{1}{n} \ln Z_n (\beta,J,h) = \sup_{z \in B(0,1)} \psi^{\beta,J,m} (z)
\end{align*}
almost surely.
\end{theorem}
\noindent
For the proof, see \Cref{csd_sec}.
\newpage
\noindent
Since the random external field is almost surely strongly varying, the classification of the parameter ranges for the pure states and mixed states is the same as for the deterministic inhomogeneous external field. The updated values of $m$ and $\beta_c$ are $m^\parallel = \mathbb{E} h_0$, $m^\perp = \sqrt{\mathbb{V} h_0}$, and $\beta_c = \frac{J}{(J - \sqrt{\mathbb{V} h_0}) (J + \sqrt{\mathbb{V} h_0})}$. We present this table here for completeness.
\begin{table}[h!] 
\begin{center}
\caption{Parametric ranges for the structure of the set of global maximizing points of the limiting exponential tilting function for the random external field} \label{random_table}
\begin{tabular}{ |c|c|c|c| } 
\hline
Mean & Standard deviation and coupling constant & Inverse temperature & $M^*(\beta,J,m)$ \\
\hline
$\mathbb{E} h_0 \not = 0$ & $J > 0$ & $\beta > 0$ & $\{ z^*  \}$ \\
$\mathbb{E} h_0 = 0$ & $\sqrt{\mathbb{E} h_0^2}\geq J$ & $\beta > 0$ & $\{ z^0 \}$ \\
$\mathbb{E} h_0 = 0$ & $\sqrt{\mathbb{E} h_0^2} < J$ & $\beta \leq \beta_c$ & $\{ z^0 \}$ \\
$\mathbb{E} h_0 = 0$ & $\sqrt{\mathbb{E} h_0^2} < J$ & $\beta > \beta_c$ & $\{ z^+, z^- \}$ \\
\hline
\end{tabular}
\end{center}
\end{table}
\\
As a direct application of \Cref{main_result1}, we have the same partial classification of the IVGS for the random external field as for the deterministic field, the only difference being that the classification only holds almost surely. For the random external field, the IVGS can be characterized almost surely also in the case where we are in the MS parameter range subject to further assumptions to the random external field. Recall that the first three rows of \Cref{random_table} describe the PS parameter range and the last row describes the MS parameter range.
\subsubsection{Chaotic size dependence}
\noindent
For the PS parameter range, the IVGS is unique almost surely and its proof only relies on assumption (A1). For the MS parameter range, we begin by noting that the sequence of vectors $\{ m_n^h \}_{n \in \mathbb{N}}$ can be written entirely in terms of the centred random walk $\{ S_n' \}_{n \in \mathbb{N}}$ by
\begin{align} \label{walk_rep}
n(m_n^{h, \parallel} - m^\parallel) = \left( S_n' \right)_1, \ n(m_n^{h, \perp} - m^\perp) = \frac{1}{m_n^{h, \perp} + m^\perp} \left( S_n' \right)_2 - \frac{m_n^{h, \parallel}}{m_n^{h, \perp} + m^\perp}  \left( S_n' \right)_1 .
\end{align}
Subject to assumption (A2), it follows that $\frac{1}{\sqrt{n}} S_n' \to G$ in distribution in the limit as $n \to \infty$, where $G$ is a non-degenerate $2$-dimensional Gaussian, and, as a result, the recurrent and possible values of the centred random walk $\{ S_n' \}_{n \in \mathbb{N}}$ are the same, see \cite[Chapter 5]{Durrett2019}. Using the recurrence of the centred random walk, we show that the sequence of vectors $\{ m_n^h \}_{n \in \mathbb{N}}$ satisfies a similar recurrence result.
\begin{lemma} Let $h$ be a random external field which satisfies (A1) and (A2).
\\
\\
It follows that
\begin{align*}
\left\{ \left( p_1, \frac{1}{2 \sqrt{\mathbb{E} h_0^2}} p_2 \right) : p := (p_1,p_2) \in P \right\}  \subset  L \left(\left\{ n (m_n^h - m)\right\}_{n \in \mathbb{N}} \right) 
\end{align*}
almost surely.
\end{lemma}
\noindent
For the proof, see \Cref{limit_theorems}. Let us also briefly remark and clarify on proofs of this kind which involve many steps which hold almost surely. If there is a collection, with at most countable size, of statements which all hold almost surely, then the intersection of these statements also holds almost surely. In the proofs, there will typically be a number of such almost sure statements which are used in the order they appear. These proofs should be read so that one collects all almost sure statements made in the proof, and the set of probability $1$ for which the theorem holds is the intersection of all of these statements.
\\
\\
As an application of \Cref{conv_class} in the case where $\delta = 1$ and $\gamma = \left(p_1,  \frac{1}{2 \sqrt{\mathbb{E} h_0^2}} p_2\right)$, we have the following complete classification of the IVGS almost surely.
\begin{theorem} \label{main_result3}
Let $h$ be a random external field which satisfies (A1).
\\
\\
For the pure state parameter range, it follows that
\begin{align*}
\mathcal{G}_\infty (\beta,J,h) = \left\{ \nu_\infty^{z^*,h} \right\}
\end{align*}
almost surely.
\\
\\
If $h$ also satisfies (A2) and (A3), for the mixed state parameter range, it follows that
\begin{align*}
\mathcal{G}_\infty (\beta,J,h) = \operatorname{conv} \left( \nu_\infty^{z^+,h}, \nu_\infty^{z^-,h}\right)
\end{align*}
almost surely.
\end{theorem} 
\noindent 
The proof of this result, see \Cref{csd_sec}, is given for the case where $\pi_1 (P)$ is not necessarily the whole space.
\\
\\
The phenomenon proven in this result is referred to as chaotic size dependence(CSD) due to Newman and Stein \cite{Newman1996}. For more physically relevant models, this property is distressing in the sense that it predicts that the IVGS depend on the way the subsequences of finite volumes are selected when the external field is random. This property has been studied for a variety of systems including the BFCW in \cite{Matos1992}, and \cite{Kuelske1997}. For the Ising model with random boundary conditions, a result of this type was obtained in \cite{Enter2005}. To our knowledge, our result is novel for this particular model, and the generality of the result is greater than similar results obtained for other models with random external fields. In particular, the works of \cite{Salinas1985}, \cite{Kuelske1997}, and \cite{Matos1992} give a certain emphasis to the case where the random external field has components which are Bernoulli distributed. In addition, in \cite{Kuelske1997}, it is remarked that a random external field with continuous components, as opposed to the discrete components of the Bernoulli field, is expected to realize all convex combinations of the pure states. This is indeed the case in this model for any $h_0$ satisfying (A1), (A2) and (A3).
\subsubsection{Construction of the Aizenman-Wehr metastate}  
\noindent
Since almost sure convergence is too strong of a form of convergence for FVGS, we will instead consider weaker forms of convergence which ultimately result in constructions of limiting objects similar to the IVGS. These constructions have been introduced in the disordered systems literature, and we will reference them as they appear in this paper. For more details and exposition, we refer to \cite[Chapter 6]{Bovier2006}. In addition, since we are dealing with random probability measures, for convergence properties of random probability measures, we refer to \cite[Chapter 4]{Kallenberg2017}.
\\
\\
We begin with the collection of joint probability measures $\{ K_n^{\beta,J}\}_{n \in \mathbb{N}}$ which act on $f \in C_b (\mathbb{R}^\mathbb{N} \times \mathbb{R}^\mathbb{N})$ by
\begin{align*}
K_n^{\beta,J} [f] := \mathbb{E} \mu_n^{\beta,J,h} [f(h, \cdot)],
\end{align*}
where the expectation $\mu_n^{\beta,J,h} [f(h, \cdot)]$ is taken with respect to the second argument. Note that the marginal distribution of the first component is simply the distribution of $h$, and the marginal distribution of the second component is given by the intensity measure $\mathbb{E} \mu_n^{\beta,J,h}$ which acts on $f \in C_b (\mathbb{R}^\mathbb{N})$ by $f \mapsto \mathbb{E} \mu_n^{\beta,J,h} [f]$. We denote the weak limit, when it exists, of the joint probability measures by $K^{\beta,J}$.
\\
\\
Next, we will consider the collection of metastate probability measures $\{ \mathcal{K}_{n}^{\beta,J}\}_{n \in \mathbb{N}}$ which are the probability distributions of the $\mathbb{R}^\mathbb{N} \times \mathcal{M}_1 (\mathbb{R}^\mathbb{N})$-valued random variables $(h, \mu_n^{\beta,J,h})$. Note that the marginal distribution of the first component of the metastate probability measure is the distribution of $h$, and the marginal distribution of the second component is the distribution of $\mu_n^{\beta,J,h}$. We denote the probability measure corresponding to the limit in distribution, when it exists, of the metastate probability measures by $\mathcal{K}^{\beta,J}$.  Since $\mathcal{K}^{\beta,J}$ is a probability measure on $\mathbb{R}^\mathbb{N} \times \mathcal{M}_1 (\mathbb{R}^\mathbb{N})$, we can obtain a random probability measure on probability measures by taking the regular conditional distribution of the second argument given the first argument.
\begin{definition} \label{metastate_def} A conditioned metastate probability measure or Aizenman-Wehr metastate $\kappa^{\beta,J,h}$, when it exists, is a measurable map $\kappa^{\beta,J, \cdot} : \mathbb{R}^\mathbb{N} \to \mathcal{M}_1 (\mathcal{M}_1 (\mathbb{R}^\mathbb{N}))$ which satisfies
\begin{align*}
\int_{\mathbb{R}^\mathbb{N} \times \mathcal{M}_1 (\mathbb{R}^\mathbb{N})} \mathcal{K}^{\beta,J} (dh, d \mu) \ f (h, \mu) = \mathbb{E} \kappa^{\beta,J,h} [f(h, \cdot)]
\end{align*}
for all $f \in C_b (\mathbb{R}^\mathbb{N} \times \mathcal{M}_1 (\mathbb{R}^\mathbb{N}))$.
\end{definition} 
\noindent
This construction is due to Aizenman and Wehr \cite{Aizenman1990}, and the name metastate refers to the fact that the resulting object is a probability measure on probability measures. Although we gave here the definition of a conditioned metastate probability measure, we will still refer to upcoming construction as the conditioned metastate probability measure.
\\
\\
Let us now remark on some properties of the joint probability measures and the metastate probability measures. If $f \in C_b (\mathbb{R}^\mathbb{N} \times \mathbb{R}^\mathbb{N})$, then it follows that the map $(h, \mu) \mapsto \mu[f(h, \cdot)]$ is continuous and bounded. As a result, if the weak limit of the metastate probability measures exists in the limit as $n \to \infty$, we must have
\begin{align*}
\int_{\mathbb{R}^\mathbb{N} \times \mathcal{M}_1 (\mathbb{R}^\mathbb{N})} \mathcal{K}^{\beta,J} (dh, d \mu) \ \mu [f(h, \cdot)] = \lim_{n \to \infty} \mathbb{E} \mu_n^{\beta,J,h} [f(h, \cdot)] =  \int_{\mathbb{R}^\mathbb{N} \times \mathbb{R}^\mathbb{N}} K^{\beta,J} (dh, d \phi) \ f (h, \phi) .
\end{align*}
It follows that the weak limit of the metastate probability measures completely determines the weak limit of the joint probability measures. As a result, the joint probability measures are in some sense redundant if one has limiting results pertaining to the metastate probability measures. In this paper, we will use the joint probability measures primarily as a tool to prove uniform tightness results of the metastate probability measures. To that end, although we will not make immediate use of the following result, we show the uniform tightness of the collection of intensity measures $\{ \mathbb{E} \mu_n^{\beta,J,h} \}_{n \in \mathbb{N}}$.
\begin{lemma} \label{intensity_tight} Let $h$ be a random external field which satisfies (A1).
\\
\\
It follows that the collection of intensity measures $ \left\{ \mathbb{E} \mu_n^{\beta, J, h} \right\}_{n \in \mathbb{N}}$ is uniformly tight.
\end{lemma}
\noindent
For the full proof, see \Cref{aw_conv_sec}. The uniform tightness of the various metastate probability measures follow by combining this result with \Cref{joint_tight} and \Cref{ns_distr_tight}.
\\
\\
For the MS parameter range, recall that the random variable $m_n^h$ can be written in terms of the $2$-dimensional centred random walk $S_n'$ presented in \Cref{walk_rep}. By using the multivariate delta method, we have the following central limit theorem for the sequence of vectors $\{ m_n^h \}_{n \in \mathbb{N}}$.
\begin{lemma} Let $h$ be a random external field which satisfies (A1) and (A2).
\\
\\
It follows that 
\begin{align*}
\lim_{n \to \infty} (h, \sqrt{n} (m_n^h - m)) = (h,G)
\end{align*} 
in distribution, where $G$ is a non-degenerate $2$-dimensional Gaussian random variable independent of $h$. 
\end{lemma}
\noindent
For the proof, see \Cref{limit_theorems}. The multivariate delta method is a standard tool of statistics, one can see \cite{Pinelis2016} for some more direct references and discussion.
\\
\\
By using Skorohod's representation theorem, see \cite[Chapter 17]{Klenke2020}, we can construct another probability space on which the convergence in distribution of $(h, \sqrt{n} (m_n^h - m)) \to (h,G)$ is elevated to almost sure convergence. On this new probability space, subject to a slight abuse of notation, we can apply the previous main result \Cref{conv_class} in the case where $\delta = \frac{1}{2}$ and $\gamma = G$ almost surely. Using these methods, we have the following result concerning the weak limit of the metastate probability measures. 
\begin{theorem} \label{aw_conv1}
Let $h$ be a random external field which satisfies (A1) and (A2).
\\
\\
For the mixed state parameter range, it follows
\begin{align*}
\lim_{n \to \infty} \mathcal{K}_n^{\beta, J}[f] = \frac{1}{2} \int_{\Omega} d \mathbb{P} \ f(h, \nu_\infty^{z^+,h}) + \frac{1}{2} \int_{\Omega} d \mathbb{P} \ f(h, \nu_\infty^{z^-,h}) := \mathcal{K}^{\beta,J} [f]
\end{align*}
for any $f \in C_b (\mathbb{R}^\mathbb{N} \times \mathcal{M}_1 (\mathbb{R}^\mathbb{N}))$.
\end{theorem}
\noindent
For the full proof, see \Cref{aw_conv_sec}.
\\
\\
As a direct corollary, using the fact that $h \mapsto \nu_\infty^{z^\pm,h}$ is a continuous mapping, we construct the Aizenman-Wehr metastate of this model. 
\begin{theorem} \label{aw_construction}
Let $h$ be a random external field which satisfies (A1) and (A2).
\\
\\
For the mixed state parameter range, the Aizenman-Wehr metastate is given by
\begin{align*}
\kappa^{\beta,J,h} := \frac{1}{2} \delta_{\nu_\infty^{z^+,h}} + \frac{1}{2} \delta_{\nu_\infty^{z^-,h}} .
\end{align*}
\end{theorem}
\noindent
To our knowledge this is a novel result for this particular model and it surpasses other similar models in its level of generality. Similar results have been obtained in \cite{Matos1992} and \cite{Kuelske1997} for the BFCW model. We also emphasize that the proof of weak convergence of the metastate probability measures is almost a direct corollary of the previous main result \Cref{conv_class} by using Skorohod's representation theorem. This proof strategy does not seem to be utilized in either \cite{Matos1992} or \cite{Kuelske1997}.
\subsubsection{Phase characterization}  
\noindent
To better understand this result, we will characterize this model in  terms of the phase characterization of disordered systems given in \cite{Pastur1978}. This characterization describes the phases in terms of the expectation and variance of the magnetization density. We can give an equivalent characterization of the RFMFS model.
\\
\\
If we return to the representation \Cref{can_rep}, we see that the magnetization density of this model is given by
\begin{align*}
\mu_n^{\beta,J,h} \left[ \frac{M_n}{n} \right] = \int_{B(0,1)} \rho_n^{\beta,J,h} (dz) \ x .
\end{align*}
We have the following result.
\begin{lemma}
Let $h$ be a random external field which satisfies (A1) and (A2).
\\
\\
For the pure state parameter range, it follows that
\begin{align*}
\lim_{n \to \infty} \mathbb{E}  \mu_n^{\beta,J,h} \left[ \frac{M_n}{n} \right] = x^*, \ \lim_{n \to \infty} \mathbb{V} \mu_n^{\beta,J,h} \left[ \frac{M_n}{n} \right] = 0 .
\end{align*}
For the mixed state parameter range, it follows that
\begin{align*}
\lim_{n \to \infty} \mathbb{E}  \mu_n^{\beta,J,h} \left[ \frac{M_n}{n} \right] = 0, \ \lim_{n \to \infty} \mathbb{V} \mu_n^{\beta,J,h} \left[ \frac{M_n}{n} \right] = \left( x^+ \right)^2 > 0 .
\end{align*}
\end{lemma}
\noindent
The proof of this statement is a direct application of the convergence in distribution of the expectation of magnetization, see \Cref{mag_dist}.
\\
\\ 
We present the characterization for the RFMFS model in the following table.
\begin{table}[h!]
\begin{center}
\caption{Characterization of the phases of the RFMFS model} \label{spinglass_table}
\begin{tabular}{ |c|c|c|c| } 
\hline
Mean & Variance and coupling constant & Inverse temperature & Phase \\
\hline
$\mathbb{E} h_0 \not = 0$ & $J > 0$ & $\beta > 0$ & ordered ferromagnet \\
$\mathbb{E} h_0 = 0$ & $\sqrt{\mathbb{E} h_0^2} \geq J$ & $\beta > 0$ & ordered paramagnet \\
$\mathbb{E} h_0 = 0$ & $\sqrt{\mathbb{E} h_0^2} < J$ & $\beta \leq \beta_c$ & ordered paramagnet \\
$\mathbb{E} h_0 = 0$ & $\sqrt{\mathbb{E} h_0^2} < J$ & $\beta > \beta_c$ & spin glass \\
\hline
\end{tabular}
\end{center}
\end{table}
\\
We see that the transition from the PS parameter range to the MS parameter range involves the phase transition from the ordered paramagnetic phase to the spin glass phase. In the spin glass phase, the model exhibits chaotic size dependence and a unique "splitting" of the pure states unlike that of the deterministic inhomogeneous external field model. In particular, the A-W metastate describes a model where limiting states are positive magnetization states with equal probability to the negative magnetization states. 
\subsubsection{Convergence of the Newman-Stein metastates}  
\noindent
We introduce the collection of empirical metastates or Newman-Stein metastates  $\{ \overline{\kappa}_N^{\beta,J,h}\}_{N \in \mathbb{N}}$ which are given by
\begin{align} \label{ns_metastate_def}
\overline{\kappa}_N^{\beta,J,h} := \frac{1}{N} \sum_{n=1}^N \delta_{\mu_n^{\beta,J,h}} 
\end{align} 
almost surely. The collection of probability distributions of the N-S metastates is denoted by $\{ \overline{\mathcal{K}}_N^{\beta,J,h}\}_{N \in \mathbb{N}}$. The N-S metastates were introduced in \cite{Newman1997} by Newman and Stein.
\\
\\
To begin the study of their convergence properties, we first consider their uniform tightness. The almost sure uniform tightness of the collection of N-S metastates follows from the almost sure uniform tightness of $\mathcal{G}(\beta,J,h)$, which can be deduced from either \Cref{main_result3} or \Cref{class1}, and \Cref{ns_tight}. The uniform tightness of the collection of probability measures of N-S metastates follows from the uniform tightness of the collection of intensity measures given in \Cref{intensity_tight} and \Cref{ns_distr_tight}. We state these two results as a lemma.
\begin{lemma} \label{ns_tight_2} Let $h$ be a random external field which satisfies (A1).
\\
\\
It follows that the collection of Newman-Stein metastates $\{ \overline{\kappa}_N^{\beta,J,h}\}_{N \in \mathbb{N}}$ is uniformly tight almost surely, and the collection of probability measures of Newman-Stein metastates $\{ \overline{\mathcal{K}}^{\beta,J}_N\}_{N \in \mathbb{N}}$ is uniformly tight.
\end{lemma}
\noindent
Given the uniform tightness of these collections, it is enough to study expectations of the form
\begin{align*}
\overline{\kappa}_N^{\beta,J,h} [P] := \frac{1}{N} \sum_{n=1}^N P (\mu_n^{\beta,J,h} [f_1],..., \mu_n^{\beta,J,h} [f_m]),
\end{align*}
where $P : \mathbb{R}^m \to \mathbb{R}$ is a finite degree polynomial of $m$-variables and $\{ f_i \}_{i=1}^m$ is a finite collection belonging to $\operatorname{LBL} (\mathbb{R}^\mathbb{N})$, see \Cref{subalgebra1}. In the case of almost sure convergence, one considers such expectations in the limit as $N \to \infty$ almost surely, and, in the case of weak convergence of the probability distributions of the N-S metastates, one considers such expectations in the limit in distribution in the limit as $N \to \infty$ .
\\
\\
To study the limits, we begin by introducing a collection of sets $\{ A_{n, \delta} \}_{n \in \mathbb{N}}$ given by
\begin{align*}
A_{n, \delta} := \left( \left[ - n^{- \frac{1}{2} + \delta}, - n^{-\frac{1}{2} - \delta}\right] \cup \left[ - n^{- \frac{1}{2} + \delta}, n^{-\frac{1}{2} + \delta}\right] \right) \subset \mathbb{R}^2 ,
\end{align*} 
where $0 < \delta < \frac{1}{6}$. We will use this collection of sets as "conditioning sets" for the N-S metastates. We show the following three results for this collection. First, subject to the addition of assumption (A4), we show that
\begin{align*}
\lim_{N \to \infty} \frac{1}{N} \sum_{n=1}^{N} \mathbbm{1}(m_n^h - m \not \in A_{n, \delta}) = 0 ,
\end{align*}
almost surely, see \Cref{subseq} for the proof. 
\\
\\
This result allows one to consider the Newman-Stein metastates only "along" the sets $A_{n, \delta}$. Using the control of the fluctuation of $m_n^h - m$ provided by conditioning on the sets $A_{n, \delta}$, along with the asymptotics developed for the weights $W_n^{\beta,J,h,+}$ in \Cref{diff_eq2}, we show that
\begin{align*}
\lim_{n \to \infty} \mathbbm{1}(m_n^h - m \in A_{n, \delta}) \left| W_n^{\beta,J,h,+} - \mathbbm{1} \left( \sum_{i=1}^n h_i > 0 \right) \right| = 0 ,
\end{align*}
and
\begin{align*}
\lim_{n \to \infty} \mathbbm{1}(m_n^h - m \in A_{n, \delta}) \left| \mu_n^{\beta,J,h, \pm} [f] - \nu_\infty^{z^\pm,h} [f] \right| = 0 ,
\end{align*}
where $f \in \operatorname{LBL} (\mathbb{R}^\mathbb{N})$, see \Cref{ns_conv} for the full proof. Combining these results together, we have the following result concerning the pathwise asymptotics of the N-S metastates.
\begin{lemma} \label{ns_conv2} Let $h$ be a random external field which satisfies (A1), (A2), and (A4).
\\
\\
For the mixed state parameter range, it follows that
\begin{align*}
 \overline{\kappa}_{N}^{\beta,J,h} [P] = T_{N}^+ P (\nu_\infty^{z^+,h} [f_1],..., \nu_\infty^{z^+,h} [f_m]) +  (1 - T_{N}^+) P (\nu_\infty^{z^-,h} [f_1],..., \nu_\infty^{z^-,h} [f_m])  + o(1)
\end{align*}
almost surely, where
\begin{align*}
T_N^+ := \frac{1}{N} \sum_{n=1}^N \mathbbm{1} \left( \sum_{i=1}^n h_i > 0 \right) .
\end{align*}
\end{lemma}
\noindent
For the full proof and notation, see \Cref{ns_metastate}. The results presented and proved here for conditioning sets are model specific adaptations of the ideas and methods concerning "regular sets" presented in \cite{Kuelske1997}. In the above result, we are using the standard little-$o$ asymptotic notation. 
\\
\\
The limiting structure of the N-S metastates is then determined by the properties of the collection of random variables $\{ T_N^+ \}_{N \in \mathbb{N}}$, which correspond to the portion of time that a $1$-dimensional random walk with step-length $h_0$ spends in the upper half-plane. Results concerning this collection of random variables are classical, and we refer to \cite[Chapter 4]{Spitzer1964}. We will utilize the following two results. The first results concerns the convergence in distribution of $\{ T_N^+\}_{N \in \mathbb{N}}$ to an arcsine distributed random variable. The second related result concerning the characterization of the set of limit points of $\{ T_N^+\}_{N \in \mathbb{N}}$, follows by using the Hewitt-Savage $0$-$1$ law, see \cite[Chapter 12]{Klenke2020}, and the convergence in distribution to an arcsine distributed random variable.
\begin{lemma} \label{arcsine_dense} Let $h$ be a random external field which satisfies (A1) and (A2).
\\
\\
It follows that
\begin{align*}
\lim_{N \to \infty} (h, T_N^+) = (h, \alpha)
\end{align*}
in distribution, where $\alpha$ is an arcsine distributed random variable independent of $h$ given by its distribution function
\begin{align*}
\mathbb{P}(\alpha \leq x) = \frac{\arcsin(2x - 1)}{\pi} + \frac{1}{2}
\end{align*}
for $x \in [0,1]$, and
\begin{align*} 
L \left( \{ T_N^+ \}_{N \in \mathbb{N}} \right) = [0,1]
\end{align*}
almost surely.
\end{lemma}
\noindent
For the full proof, see \Cref{ns_metastate}.
\\
\\
By combining the almost sure uniform tightness of the N-S metastates from \Cref{ns_tight_2}, the pathwise asymptotics from \Cref{ns_conv2}, and the limit point density from \Cref{arcsine_dense}, we have the following CSD result for the almost sure convergence of the N-S metastates.
\begin{theorem} \label{main_result4} Let $h$ be a random external field which satisfies (A1), (A2), and (A4).
\\
\\
For the mixed state parameter range, it follows that
\begin{align*}
\operatorname{conv} (\delta_{\nu_\infty^{z^+,h}}, \delta_{\nu_\infty^{z^-,h}}) =L \left( \left\{ \overline{\kappa}^{\beta,J,h}_N \right\}_{N \in \mathbb{N}}\right)  
\end{align*}
almost surely.
\\
\\
In particular, it follows that $\{ \overline{\kappa}^{\beta,J,h}_N \}_{N \in \mathbb{N}}$ does not converge almost surely but there does exist a random subsequence $\{ N_{k} \}_{k \in \mathbb{N}}$ such that
\begin{align*}
\overline{\kappa}_{N_{k}}^{\beta,J,h} \to \kappa^{\beta,J,h}  
\end{align*}
almost surely.
\end{theorem}
\noindent
See \Cref{ns_metastate} for the full proof.
\\
\\
By combining the almost sure uniform tightness of the N-S metastates from \Cref{ns_tight_2}, the pathwise asymptotics from \Cref{ns_conv2}, and the convergence in distribution to the arcsine distributed random variable, we have the following convergence in distribution result.
\begin{theorem} \label{main_result5} Let $h$ be a random external field which satisfies (A1), (A2), and (A4).
\\
\\
For the mixed state parameter range, it follows that
\begin{align*}
\lim_{n \to \infty} \overline{\kappa}^{\beta,J,h}_N = \alpha \delta_{\nu_\infty^{z^+,h}} + (1 - \alpha) \delta_{\nu_{\infty}^{z^-,h}} := \overline{\kappa}^{\beta,J,h},
\end{align*}
in distribution, where $\alpha$ is an arcsine distributed random variable independent of $h$.
\end{theorem}
\noindent
See \Cref{ns_metastate} for the full proof.
\\
\\
The Newman-Stein metastates were introduced as a way to obtain some form of almost sure convergence for the FVGS. However, as can be seen from these results, this is not the case, but at least one can realize the A-W metastate as random subsequence of the N-S metastates.
\\
\\
The convergence in distribution of the N-S metastates clearly shows the pathwise dependence of the model. In some sense, the presence of the arcsine random variable is a result of the pathwise dependence of the weights of the FVGS. Since the weights behave like indicator functions for large enough $n$, the result is somewhat expected. 
\\
\\
This result is novel for this specific model, and similar, almost identical, results have been obtained by \cite{Kuelske1997} for the BFCW model. The biggest difference between the proof technique of these results is that for Bernoulli components, one can work directly with the $1$-dimensional simple random walks. For this model, it seems necessary to use some methods of non-linear statistics for $2$-dimensional random walks as in the proof of \Cref{subseq}.
\subsubsection{Triviality of metastates in the pure state parameter range}  
\noindent
We have not yet discussed the metastates for the pure state parameter range. This is because they are trivial due to to the almost sure convergence of the FVGS from \Cref{main_result3}. As a direct application of \Cref{trivial_conv}, we have the following result.
\begin{theorem} \label{triviality}
Let $h$ be a random external field which satisfies (A1).
\\
\\
For the pure state parameter range, we have
\begin{align*}
\lim_{n \to \infty} (h, \mu_n^{\beta,J,h}) = (h,\nu_\infty^{z^*,h})
\end{align*}
in distribution, and 
\begin{align*}
\lim_{N \to \infty} \overline{\kappa}_N^{\beta,J,h} = \delta_{\nu_\infty^{z^*,h}}
\end{align*}
almost surely.
\end{theorem}
\subsubsection{Summary and remarks}
\noindent
As we earlier remarked, the results obtained are novel for this particular model and they are universal for random external fields in the precise sense given by the assumptions of the theorems concerning the random variable $h_0$.
\\
\\
In the introduction, we noted that the RFCW model, in principle, should allow one to study limiting free energies with many global maximizing points with varying strengths. For the RFMFS model, due possibly to the spherical constraint, the only options are that the global maximizing point is unique, in which case the structure of the metastates is trivial, or there are exactly two global maximizing points of quadratic type which constitute the simplest form of global maximizing points. The results then for the case of two global maximizing points are almost identical to the results concerning the metastate obtained for the BFCW model. As such, it is natural to see many methods repeated such as the construction of sequence of local maximizing points of the exponential tilting functions, and the conditioning sets used for the analysis of the N-S metastate. The most significant difference between the BFCW model and the RFMFS model is that main calculations for the BFCW model concerning the limiting free energy consider the analysis of a smooth function on an unbounded set where all derivatives contain some form of randomness while the analogous analysis for the RFMFS model considers a two-dimensional smooth function on a bounded set such that the randomness vanishes for derivatives of order two or higher.  
\\
\\
As a final remark, let us comment on the methods and usability of the arguments presented here for cases other than the independent identically distributed components case. The proofs concerning the CSD phenomenon rely heavily on the recurrence properties of random walks. In a similar fashion, the proofs concerning the convergence of the N-S metastate rely on specific properties of random walks leading to the arcsine law. On the other hand, the construction of the A-W metastate relies on the strong law of large numbers, or generic almost sure convergence, to resolve the almost sure convergence of the "microcanonical" measures, permutation invariance of $h$ to prove uniform tightness, and the weak convergence relies on the convergence in distribution of the magnetization vectors $\{ \sqrt{n} (m_n^h - m) \}_{n \in \mathbb{N}}$, which in turn comes down to proving results concerning the scaled sums of the sequence $\{ (h_i, h_i^2)\}_{i \in \mathbb{N}}$. Since the results concerning weak convergence use Skorohod's representation theorem, one can, in principle, study external fields satisfying the other required properties but with slower or possibly faster rates of asymptotic convergence in distribution of the magnetization vectors. To accommodate this type of analysis, one can also use the methods of proof for the asymptotics of the magnetization vectors to consider cases where the rate of convergence is not the same for each individual component of the sequence of vectors, but different. For examples of this type of analysis for large deviations of the RFCW model, see \cite{Loewe2013}.
\section{Proofs of results} \label{proofs_of_results}
\subsection{Rigorous delta function calculation} \label{rigorous_delta}
\noindent
For this first calculation, we will need the following two properties concerning the integral over the sphere which are presented and proved in the appendix of \cite{Lukkarinen2019}. These two properties are orthogonal invariance and the decomposition of the sphere into subdimensional spheres. 
\\
\\
If $f \in C_b (\mathbb{R}^n)$ and $O : \mathbb{R}^n \to \mathbb{R}^n$ is an orthogonal transformation, then it follows that
\begin{align*}
\int_{\mathbb{R}^n} d \phi \ \delta (|| \phi||^2 - n) f (\phi)  = \int_{\mathbb{R}^n} d \phi \ \delta (|| \phi||^2 - n) (f \circ O^{-1} ) (\phi) .
\end{align*}
If $f \in C_b (\mathbb{R}^n)$ and $1 < k < n$, then it follows that
\begin{align*}
\int_{\mathbb{R}^n} d \phi \ \delta (|| \phi||^2 - n) f (\phi) = \frac{1}{2} \int_{\mathbb{R}^k} d \phi' \left( n - || \phi '||^2 \right)^{\frac{n - k}{2} - 1} \mathbbm{1}(|| \phi '||^2 < n) \int_{\mathbb{S}^{n - k - 1}} d \Omega \ f \left( \phi', \sqrt{n - || \phi'||^2} \Omega \right) ,
\end{align*}
where we have identified $\mathbb{R}^n = \mathbb{R}^k \oplus \mathbb{R}^{n - k}$ for the argument of $f$.
\\
\\
Recall the collection of vectors $w_{1,n}$, $w_{2,n}^h$, and $\{ v_{j,n}^h \}_{j=3}^n$ and the orthogonal change of coordinates $O_n^h$ given in the surrounding text of \Cref{basis}, \Cref{hamiltonian}, and \Cref{partition_formal}. The following lemma is a rigorous version of the formal calculation presented in \Cref{partition_formal}. 
\begin{lemma} \label{can_rep}
Suppose that $h$ satisfies $m_n^{h, \perp} \not = 0$.
\\
\\
It follows that
\begin{align*}
&\frac{2}{n^{\frac{n}{2} - 1}} \frac{1}{|\mathbb{S}^{n - 3}|}\int_{\mathbb{S}^{n-1}} d \Omega \  e^{\frac{\beta J}{2} \left( \sum_{i=1}^n \Omega_i \right)^2 + \beta \sqrt{n} \sum_{i=1}^n h_i \Omega_i} f (\sqrt{n} \Omega) \\
&= \int_{B(0,1)} dz \ e^{2 \beta J x^2 + 4 \beta \left< m_n^h, z \right>} e^{(n - 4) \left( \frac{\beta J}{2} x^2 + \beta \left< m_n^h, z \right> + \frac{1}{2} \ln (1 - || z ||^2)\right)} \\
&\times \frac{1}{|\mathbb{S}^{n - 3}|}\int_{\mathbb{S}^{n - 3}} d \Omega \ f \left( \sqrt{n}x w_{1,n} + \sqrt{n}y w^h_{2,n} + \sqrt{1 - || z ||^2}\sqrt{n} \sum_{j=3}^n \Omega_j v_{j,n}^h \right) 
\end{align*} 
for $f \in C_b (\mathbb{R}^n)$.
\end{lemma}
\begin{proof}
Using the orthogonal invariance property, we have
\begin{align*}
&\int_{\mathbb{S}^{n-1}} d \Omega \  e^{\frac{\beta J}{2} \left( \sum_{i=1}^n \Omega_i \right)^2 + \beta \sqrt{n} \sum_{i=1}^n h_i \Omega_i} f (\sqrt{n} \Omega) \\
&= \int_{\mathbb{S}^{n-1}} d \Omega \  e^{\frac{\beta J n}{2} \Omega_1^2 + \beta n m_n^\parallel \Omega_1 + \beta n m^\perp \Omega_2} f \left( \sqrt{n} \left( \Omega_1 w_{1,n} + \Omega_2 w_{2,n} + \sum_{j=3}^n \Omega_j v_{j,n} \right) \right) .
\end{align*}
Using the subdimensional sphere decomposition, we have
\begin{align*}
&\frac{2}{n^{\frac{n}{2} - 1}} \frac{1}{|\mathbb{S}^{n - 3}|}\int_{\mathbb{S}^{n-1}} d \Omega \  e^{\frac{\beta J n}{2} \Omega_1^2 + \beta n m_n^{h, \parallel} \Omega_1 + \beta n m^\perp \Omega_2} f \left( \sqrt{n} \left( \Omega_1 w_{1,n} + \Omega_2 w_{2,n} + \sum_{j=3}^n \Omega_j v_{j,n} \right) \right) \\
&= \int_{B(0,1)} dz \ e^{2 \beta J x^2 + 4 \beta \left< m_n, z \right>} e^{(n - 4) \left( \frac{\beta J}{2} x^2 + \beta \left< m_n, z \right> + \frac{1}{2} \ln (1 - || z ||^2)\right)} \\
&\times \frac{1}{|\mathbb{S}^{n - 3}|}\int_{\mathbb{S}^{n - 3}} d \Omega \ f \left( \sqrt{n} x w_{1,n} + \sqrt{n} y w_{2,n} + \sqrt{1 - || z ||^2}\sqrt{n} \sum_{j=3}^n \Omega_j v_{j,n} \right) ,
\end{align*}
as desired.
\end{proof}
\subsection{Uniform convergence of microcanonical probability measures} \label{uni_conv_sec}
\noindent
Denote $\eta_n$ to be the probability measure on $\mathbb{R}^n$ obtained by setting $I = [n]$ in \Cref{infinity_gaussian}. We have the following lemma concerning the relationship between $\eta$, $T_n^{z,h}$, and $\nu_n^{z,h}$.
\begin{lemma} \label{transport1} Suppose that $h$ satisfies $m_n^{h, \perp} \not = 0$.
\\
\\
It follows that $\nu_n^{z,h} = {T_n^{z,h}}_* \eta$.
\end{lemma}
\begin{proof}
It is enough to prove that if $f \in C_b(\mathbb{R}^n)$, then ${T_n}_* \eta_n [f] = \nu_n [f \circ \pi_n]$. To that end, note that
\begin{align*}
&\int_{\mathbb{R}^n} d\phi \ e^{- \frac{|| \phi ||^2}{2}} \left(f \circ \left(T_n \right)_1 \right) (\phi) \\ &= \int_{\mathbb{R}^n} d\phi \ e^{- \frac{|| P_n(\phi)||^2 + || \phi - P_n(\phi)||^2}{2}} f\left(\sqrt{n}x w_{1,n} + \sqrt{n}y w_{2,n} + \sqrt{1 - || z ||^2} \sqrt{n} \frac{\phi - P_n (\phi)}{|| \phi - P_n (\phi)||} \right) .
\end{align*}
Furthermore, we have
\begin{align*}
\frac{\phi - P_n (\phi)}{|| \phi - P_n (\phi)||} := \frac{1}{\sqrt{\sum_{j=3}^n \left< v_{j,n}, \phi  \right>^2}}\sum_{j=3}^n \left< v_{j,n}, \phi \right> v_{j,n} .
\end{align*}
Using the orthogonal change of coordinates $O_n$, it follows that
\begin{align*}
&\int_{\mathbb{R}^n} d\phi \ e^{- \frac{|| P_n(\phi)||^2 + || \phi - P_n(\phi)||^2}{2}} f\left(\sqrt{n}x w_{1,n} + \sqrt{n} y w_{2,n} + \sqrt{1 - || z ||^2} \sqrt{n} \frac{\phi - P_n (\phi)}{|| \phi - P_n (\phi)||} \right) \\
&= \int_{\mathbb{R}^2} d \phi' \ e^{- \frac{|| \phi'||^2}{2}} \int_{\mathbb{R}^{n - 2}} d \phi \ e^{- \frac{|| \phi ||^2}{2}} f\left(\sqrt{n} x w_{1,n} + \sqrt{n} y w_{2,n} + \sqrt{1 - || z ||^2} \sqrt{n} \sum_{j=3}^n \frac{\phi_j}{|| \phi ||} v_{j,n} \right) .
\end{align*}
The integral over $\mathbb{R}^2$ is redundant since the integrand does not depend on $\phi'$. By change of coordinates to the hyperspherical coordinates for the integral over $\mathbb{R}^{n-2}$, we have
\begin{align*}
&\left(  \int_{\mathbb{R}^n} d\phi \ e^{- \frac{|| \phi ||^2}{2}} \right)^{-1} \int_{\mathbb{R}^n} d\phi \ e^{- \frac{|| \phi ||^2}{2}} \left(f \circ \left(T_n \right)_1 \right) (\phi) \\
&= \frac{1}{|\mathbb{S}^{n - 3}|}\int_{\mathbb{S}^{n - 3}} d \Omega \ f \left( \sqrt{n} x w_{1,n} + \sqrt{n} y w_{2,n} + \sqrt{1 - || z ||^2}\sqrt{n} \sum_{j=3}^n \Omega_j v_{j,n} \right),
\end{align*} 
from which the result follows.
\end{proof}
\noindent
We can now prove the result concerning the uniform convergence of $\nu_n^{z,h}$ to $\nu_\infty^{z,h}$.
\begin{proof}[Proof of \Cref{uni_conv}]
By rescaling, it is enough to prove this claim for $f \in \operatorname{LBL} (\mathbb{R}^\mathbb{N})$ such that the function and its Lipschitz constant are bounded above by $1$. Let $I$ be the index set that $f$ is local on. For large enough $n$, using the Lipschitz property, we have
\begin{align*}
\left| \nu_n [f] - \nu_\infty[f] \right| \leq \eta_n \left[ \left| \left| (\pi_I \circ T_n) (\phi) -  (\pi_I \circ T_\infty) (\phi) \right| \right| \right]
\end{align*} 
Using the explicit form of the transport maps, we have
\begin{align*}
\left| \left| (\pi_I \circ T_n) (\phi) -  (\pi_I \circ T_\infty) (\phi) \right| \right| &\leq \sqrt{|I|} \left| \frac{m_n^{\parallel}}{m_n^\perp} - \frac{m^\parallel}{m^\perp} \right| \\
&+ \sqrt{|I|} \max_{i \in I} |h_i| \left| \frac{1}{m_n^\perp} - \frac{1}{m^\perp} \right| \\
&+ \left| \left| \pi_I \left( \phi - \sqrt{n} \frac{\phi - P_n^h (\phi)}{\left| \left|\phi - P_n^h (\phi)  \right| \right|}\right)\right|  \right| .
\end{align*}
For the last term, we begin by noting that
\begin{align*}
\phi - \sqrt{n} \frac{\phi - P_n (\phi)}{\left| \left|\phi - P_n (\phi)  \right| \right|} = P_n(\phi) + (\phi - P_n(\phi)) \left(1 - \frac{\sqrt{n}}{|| \phi - P_n(\phi)||} \right) .
\end{align*} 
It follows that
\begin{align*}
&\left| \left| \pi_I \left( \phi - \sqrt{n} \frac{\phi - P_n (\phi)}{\left| \left|\phi - P_n (\phi)  \right| \right|}\right)\right|  \right| \\ &\leq || \pi_I (P_n (\phi))|| + \left| 1 - \frac{\sqrt{n}}{|| \phi - P_n(\phi)||} \right| \left( || \pi_I (P_n (\phi))|| + || \pi_I (\phi) ||\right)
\end{align*}
Now, note that
\begin{align*}
\pi_I (P_n (\phi)) = \frac{1_I}{\sqrt{n}} \left< w_{1,n}, \phi \right> + \frac{h_I - m_n^\parallel 1_I}{\sqrt{n} m_n^\perp} \left< w_{2,n}, \phi \right> . 
\end{align*}
Using the projection to form an orthogonal change of coordinates, we have
\begin{align*}
&\eta_n \left[ || \pi_I (P_n (\phi))|| \right] \\ &= \frac{1}{\sqrt{n}}\left( \int_{\mathbb{R}^2} dxdy \ e^{- \frac{x^2 + y^2}{2}}\right)^{-1} \int_{\mathbb{R}^2} dxdy \ e^{- \frac{x^2 + y^2}{2}} \left| \left| 1_I x + \frac{h_I - m_n^\parallel 1_I}{ m_n^\perp} y \right| \right| .
\end{align*}
Using dominated convergence, this term is vanishing in the limit as $n \to \infty$. Continuing, by the Cauchy-Schwartz inequality, we have
\begin{align*}
&\eta_n \left[\left| 1 - \frac{\sqrt{n}}{|| \phi - P_n(\phi)||} \right| || \pi_I (P_n(\phi))||\right] \\ &\leq \left( \eta_n \left[\left| 1 - \frac{\sqrt{n}}{|| \phi - P_n(\phi)||} \right|^2 \right] \right)^\frac{1}{2} \left( \eta_n \left[|| \pi_I (P_n(\phi))||^2 \right] \right)^\frac{1}{2} .
\end{align*}
The second term in the product on the right-hand side of the inequality will converge to something finite using the same calculation as in the previous integral. For the first term, using the projection to form an orthogonal change of coordinates and using hyperspherical coordinates, we have
\begin{align*}
\eta_n \left[\left| 1 - \frac{\sqrt{n}}{|| \phi - P_n(\phi)||} \right|^2 \right] = \left( \int_0^\infty dr \ r^{n - 3} e^{- \frac{n r^2}{2}}\right)^{-1} \int_0^\infty dr \ r^{n - 3} e^{- \frac{n r^2}{2}} \left| 1 - \frac{1}{r} \right|^2 .
\end{align*}
By using the Laplace method, see \cite[Chapter 2]{Wong2001}, one can conclude that this term vanishes in the limit as $n \to \infty$. By essentially repeating steps used for the previous term, one can see that the term
\begin{align*}
\eta_n \left[\left| 1 - \frac{\sqrt{n}}{|| \phi - P_n(\phi)||} \right| || \pi_I (\phi))||\right]
\end{align*}
is vanishing in the limit as $n \to \infty$. Collecting the inequalities and vanishing terms, it follows that
\begin{align*}
\lim_{n \to \infty} \eta_n \left[ \left| \left| \pi_I \left( \phi - \sqrt{n} \frac{\phi - P_n (\phi)}{\left| \left|\phi - P_n (\phi)  \right| \right|}\right)\right|  \right| \right] = 0 .
\end{align*}
Now, returning to the first inequality concerning the transport maps, it follows that
\begin{align*}
\left| \nu_n [f] - \nu_\infty [f] \right| &\leq \eta_n \left[ \left| \left| (\pi_I \circ T_n) (\phi) -  (\pi_I \circ T_\infty) (\phi) \right| \right| \right] \\
&\leq \sqrt{|I|} \left| \frac{m_n^{\parallel}}{m_n^\perp} - \frac{m^\parallel}{m^\perp} \right| \\
&+ \sqrt{|I|} \max_{i \in I} |h_i| \left| \frac{1}{m_n^\perp} - \frac{1}{m^\perp} \right| \\
&+ \eta_n \left[ \left| \left| \pi_I \left( \phi - \sqrt{n} \frac{\phi - P_n (\phi)}{\left| \left|\phi - P_n (\phi)  \right| \right|}\right)\right|  \right| \right] .
\end{align*}
The right hand side of this inequality is vanishing in the limit as $n \to \infty$ and it does not depend on $z \in B(0,1)$. The result follows. 
\end{proof}
\subsection{Limiting free energy and uniform tightness of infinite volume Gibbs states} \label{free_tight}
\noindent
The following result shows that the set $M^*(\beta,J,m)$ of global maximizing points of $\psi^{\beta,J,m}$ is non-empty and compact.
\begin{lemma}\label{argmax2} Let $h$ be a strongly varying external field.
\\
\\
It follows that $M^*(\beta,J,m)$  is non-empty and compact.
\end{lemma}
\begin{proof}
First, by direct calculation, for any $z \in B(0,1)$, we have
\begin{align*}
\psi(z) < \frac{\beta J}{2} + \beta |m^\parallel| + \beta m^\perp + \frac{1}{2} \ln (1 - || z ||^2) .
\end{align*}
Because of the logarithmic term, there exists $0 < R < 1$ such that for any $z \in B(0,1)$ satisfying $|| z ||^2 > R^2$, we have $\psi (z) < -1$. The set $\overline{B}(0,R)$ is compact, the mapping $\psi$ is continuous there, and $\psi (0) = 0 > - 1$. This implies that $\psi$ attains a maximum value at some point in $\overline{B}(0,R)$ which is greater than or equal to $0$, but it necessarily strictly larger than all elements in $B(0,1) \setminus \overline{B}(0,R)$, which implies that it is actually a global maximum, thus proving that $M^*$ is non-empty. In fact, this shows that $M^* \subset \overline{B}(0,R) \subset B(0,1)$.
\\
\\
It is now enough to prove that $M^*$ is closed since it is contained in a compact set. Let $\{ z_k^* \}_{k \in \mathbb{N}}$ be a sequence of elements in $M^*$ such that $z_k^* \to z^*$. Suppose that $z^*$ is not a global maximum point of $\psi$. It follows that there must exist some $z \in B(0,1)$ such that $\psi(z^*) < \psi(z)$. By pointwise convergence and continuity of $\psi$, for large enough $k$, we have
\begin{align*}
\psi(z_k^*) - \psi(z^*) < \psi(z) - \psi(z^*) \iff \psi(z_k^*) < \psi(z),
\end{align*}
which is a contradiction since $z_k^*$ is a global maximum point. It follows that $z^* \in M^*$ which shows that $M^*$ is closed.
\end{proof}
\noindent
Using the fact that $M^*(\beta,J,m)$ is non-empty and compact, we can compute the limiting free energy.
\begin{proof}[Proof of \Cref{argmax1}]
Observe that
\begin{align*}
&\left| \ln \frac{\int_{B(0,1)} dz \ e^{2 \beta J x^2 + 4 \beta \left< m_n, z \right>} e^{(n - 4)\psi_n(z)}}{\int_{B(0,1)} dz \ e^{(n - 4)\psi(z)}} \right| \\ &\leq 2 \beta J + 4 \beta || m_n || + (n - 4) \sup_{z \in B(0,1)} |\psi_n(z) - \psi(z)| .
\end{align*}
Let $\varepsilon > 0$ be arbitrary but small. By using compactness of $M^*$ and the continuity of $\psi$, it follows that there exists a set $A_\varepsilon \subset B(0,1)$ such that $\left| \psi(z) - \sup_{z' \in B(0,1)} \psi(z') \right| \leq \varepsilon$ for all $z \in A_\varepsilon$. By using such a set, it follows that
\begin{align*}
\frac{n-4}{n} \left(\sup_{z \in B(0,1)} \psi(z) - \varepsilon \right) + \frac{1}{n} \ln |A_\varepsilon|  \leq \frac{1}{n} \ln \int_{B(0,1)} dz \ e^{(n - 4)\psi(z)} \leq \frac{n-4}{n} \sup_{z \in B(0,1)} \psi(z), 
\end{align*} 
where $|A_\varepsilon|$ is the Lebesgue measure of $A_\varepsilon$. By first taking the limit as $n \to \infty$ followed by the limit as $\varepsilon \to 0^+$, it follows that
\begin{align*}
\lim_{n \to \infty} \frac{1}{n} \ln \int_{B(0,1)} dz \ e^{(n - 4)\psi(z)} = \sup_{z \in B(0,1)} \psi(z) .
\end{align*}
By combining the first observation with this limit, the result follows.
\end{proof}
\noindent
We can now prove the following result concerning the uniform tightness and limiting structure of the collection of probability measures $\{ \rho_n^{\beta, J, h} \}_{n \in \mathbb{N}}$.
\begin{proof}[Proof of \Cref{mix_supp}]
Fix $0 < R < 1$ and let $\varepsilon > 0$ be small enough such that $K := \{ z \in B(0,1) : d(z, M^*) \leq \varepsilon \} \subset \overline{B}(0, R)$. By essentially repeating the arguments of \Cref{argmax1} and \Cref{argmax2}, one can show that
\begin{align*}
\limsup_{n \to \infty} \frac{1}{n} \ln \rho_n (K^c) \leq \sup_{z \in \overline{K^c} \cap B(0,1)} \psi(z) - \sup_{z \in B(0,1)} \psi(z) < 0.
\end{align*}
By rewriting
\begin{align*}
 \rho_n(K) = 1 - e^{n  \left( \frac{1}{n} \ln \rho_n (K^c) \right)},
\end{align*}
it follows that
\begin{align*}
\inf_{k \geq n} \rho_k (K) \geq 1 - e^{n \sup_{k \geq n} \frac{1}{k} \ln \rho_k (K^c)} ,
\end{align*}
from which it follows that $\lim_{n \to \infty} \rho_n (K) = 1$ which proves uniform tightness.
\\
\\
Next, let $z \in \left( M^* \right)^c$ and $\delta_1, \delta_2 > 0$ small enough such that $\overline{B}(z, \delta_1) \subset B(z, \delta_2)$ and $B(z, \delta_2) \cap M^* = \emptyset$. By a similar argument to the uniform tightness argument above, one can show that
\begin{align*}
\liminf_{n \to \infty} \rho_n (B(z, \delta_1)) = 0 .
\end{align*}
Now, by Prokhorov's theorem, let $\rho$ be a probability measure obtained as a convergent subsequence $\{ \rho_{n_k} \}_{k \in \mathbb{N}}$ of the uniformly tight collection of probability measures $\{ \rho_n \}_{n \in \mathbb{N}}$. It follows that
\begin{align*}
\rho(B(z, \delta_1)) \leq \liminf_{k \to \infty} \rho_{n_k} (B(z, \delta_1)) = 0 .
\end{align*}
This implies that the complement of $M^*$ and the support of $\rho$ are disjoint which in turn implies that $\rho$ must be supported by $M^*$.
\end{proof}
\noindent
We conclude with the proof of the following partial classification of the IVGS.
\begin{lemma} \label{class1} Let $h$ be a strongly varying external field.
\\
\\
It follows that $\mathcal{G} (\beta,J,h)$ is uniformly tight and
\begin{align*}
 \mathcal{G}_\infty (\beta,J,h)  \subset \left\{ \mu \in \mathcal{M}_1 (\mathbb{R}^\mathbb{N}) : \operatorname{supp} (\rho) \subset M^*(\beta,J,m), \  \mu = \int_{M^*(\beta,J,m)} \rho(dz) \ \nu_\infty^{z,h}  \right\} .
\end{align*}
\end{lemma}
\begin{proof}
Let $\{ \mu_{n_k} \}_{k \in \mathbb{N}}$ be an arbitrary subsequence. By \Cref{mix_supp}, the collection of probability measures $\{ \rho_{n_k} \}_{k \in \mathbb{N}}$ is uniformly tight which implies, by Prokhorov's theorem, that for any subsequence $\{ \rho_{n_k} \}_{k \in \mathbb{N}}$ of this collection there exists a convergent subsubsequence $\{ \rho_{n_{k_j}}\}_{j \in \mathbb{N}}$ with a limit probability distribution $\rho$. We have
\begin{align*}
\left| \mu_{n_{k_j}}[f] - \int_{M^*} \rho(dz) \ \nu_\infty^z [f] \right| &\leq \sup_{z \in B(0,1)} \left| \nu_{n_{k_j}}^z [f] - \nu_\infty^z[f]\right| \\ &+ \left| \int_{M^*} \rho(dz) \ \nu_\infty^z [f] - \int_{B(0,1)} \rho_{n_{k_j}}(dz) \ \nu_\infty^z [f] \right| 
\end{align*} 
for any $f \in \operatorname{LBL}(\mathbb{R}^\mathbb{N})$. By \Cref{uni_conv} and \Cref{mix_supp}, it follows that
\begin{align*}
\lim_{j \to \infty} \mu_{n_{k_j}}[f] = \int_{M^*} \rho(dz) \ \nu_\infty^z [f] 
\end{align*}
for any $f \in \operatorname{LBL}(\mathbb{R}^\mathbb{N})$, which implies that $\mathcal{G}$ is uniformly tight. The same argument applied to a weakly convergent subsequence $\{ \mu_{n_k} \}_{k \in \mathbb{N}}$ shows that
\begin{align*}
\mathcal{G}_\infty \subset \left\{ \mu \in \mathcal{M}_1 (\mathbb{R}^\mathbb{N}) : \operatorname{supp} (\rho) \subset M^*, \  \mu = \int_{M^* } \rho(dz) \ \nu_\infty^{z}  \right\}  .
\end{align*}
\end{proof}
\noindent
The central applications of this result concern the case where $M^*(\beta,J,m)$ is finite. For that case, we have the following corollaries.
\begin{corollary} \label{class2}
Let $h$ be a strongly varying external field and suppose that $M^*(\beta,J,m)$ is finite.
\\
\\
It follows that
\begin{align*}
\mathcal{G}_\infty (\beta,J,h) \subset \operatorname{conv} \left( \left\{ \nu_\infty^{z^*,h} \right\}_{z^* \in M^*(\beta,J,m)}  \right) .
\end{align*}
\end{corollary}
\begin{proof}
If $M^*$ is finite and $\rho$ is a probability measure supported by $M^*$, it follows that $\rho$ is given by the normalized weighted sum of delta measures. To be exact, we have
\begin{align*}
\rho = \sum_{z^* \in M^*} \rho(z^*) \delta_{z^*} ,
\end{align*}
where $\rho(z^*) := \rho (\{ z^* \})$. Combining \Cref{class1} and this observation, it follows that
\begin{align*}
\int_{M^*} \rho(dx) \ \nu_\infty^z = \sum_{z^* \in M^*} \rho(z^*) \nu_{\infty}^{z^*} ,
\end{align*}
from which the result follows.
\end{proof}
\noindent
In the case where $M^*(\beta,J,m)$ is a single element, we have the following special case.
\begin{corollary} \label{single_point} Let $h$ be a strongly varying external field and suppose that $M^*(\beta,J,m) = \{ z^* \}$.
\\
\\
It follows that
\begin{align*}
\lim_{n \to \infty} \mu_n^{\beta,J,h} = \nu_\infty^{z^*,h} 
\end{align*}
weakly.
\end{corollary}
\begin{proof}
From \Cref{class1}, it follows that $\mathcal{G}$ is uniformly tight. From \Cref{class2}, it follows that every weakly convergent subsequence of $\mathcal{G}$ converges to $\nu_\infty^{z^*,h}$ which implies that the full sequence converges weakly to the same limit. 
\end{proof}
\subsection{Global maximizing points of the limiting exponential tilting function}
\noindent
We will determine the amount of points in $M^*(\beta,J,m)$ depending on the given parameters. This first result considers the case where $m^\parallel \not = 0$.
\begin{lemma} \label{max_point1}
Let $h$ be a strongly varying external field and suppose that $m^\parallel \not = 0$.
\\
\\
It follows that $\psi^{\beta,J,m}$ has a unique global maximizing point. 
\end{lemma}
\begin{proof}
We will use the fact that any local extrema of a differentiable function on an open set are necessarily critical points. This leads us to consider the critical point equation
\begin{align*}
\nabla [\psi] (x^*, y^*) = 0 \iff \beta J x^* + \beta m^\parallel - \frac{x^*}{1 - {x^*}^2 - {y^*}^2} = 0, \ \beta m^\perp - \frac{y^*}{1 - {x^*}^2 - {y^*}^2} = 0 .
\end{align*}
From the second component of the critical point equation, we can solve for $y^*$ in terms of $x^*$ by rearranging the equation to a quadratic in $y^*$ and solving for the root which satisfies $-1 < y^* < 1$. This root is given by
\begin{align*}
y^* = \sqrt{1 + \left( \frac{1}{2 \beta m^\perp} \right)^2 - {x^*}^2} - \frac{1}{2 \beta m^\perp} .
\end{align*}
By direct computation, we have
\begin{align*}
1 - {x^*}^2 - {y^*}^2 = 2 \left( \frac{1}{2 \beta m^\perp} \right) \left(\sqrt{1  + \left( \frac{1}{2 \beta m^\perp} \right)^2 - {x^*}^2} -  \frac{1}{2 \beta m^\perp} \right) .
\end{align*}  
Plugging in this value to the first component of the critical point equation, we obtain
\begin{align*}
\beta J x^* + \beta m^\parallel - \frac{x^*}{2 \left( \frac{1}{2 \beta m^\perp} \right) \left(\sqrt{1  + \left( \frac{1}{2 \beta m^\perp} \right)^2 - {x^*}^2} -  \frac{1}{2 \beta m^\perp} \right)} = 0 .
\end{align*}
Denote the function $F : (-1, 1) \to \mathbb{R}$ by
\begin{align*}
F(x) := \beta J x + \beta m^\parallel - \frac{x}{2 \left( \frac{1}{2 \beta m^\perp} \right) \left(\sqrt{1  + \left( \frac{1}{2 \beta m^\perp} \right)^2 - {x}^2} -  \frac{1}{2 \beta m^\perp} \right)} ,
\end{align*}
the solution for the value of $x^*$ will be given by the roots of $F$. We are, however, not interested in all the roots of this equation. Because we are searching for $(x^*, y^*)$ which correspond to the global maximum value of $\psi$, we can employ a symmetry argument for $\psi$ to show that the any maximum value point must lie in a certain quadrant of the unit ball. To be exact, suppose that $m^\parallel > 0$. If there existed a point $(x^*, y^*)$ such that $-1 < x^* < 0$ and the point corresponds to a global maximum point of $\psi$, then this would be a contradiction since $\psi (- x^*, y^*) > \psi (x^*, y^*)$ by direct calculation, which would imply that it is not actually a global maximum point. A similar argument holds for $m^\parallel < 0$, and a supposed maximum point which would satisfy $0 < x^* < 1$.  This argument shows that we are only interested in finding the roots of the equation $F$ on the interval $(0,1)$ if $m^\parallel > 0$, and $(-1,0)$ if $m^\parallel < 0$.
\\
\\
For the root analysis of $F$, let us begin by first remarking that $x = 0$ is never a root of $F$, since $F(0) = \beta m^\parallel \not = 0$. We only mention this since the symmetry argument used to show that the maximum point must reside on either positive or the negative interval only holds if $x \not = 0$. We will work case by case. First, suppose that $m^\parallel > 0$. We can now restrict ourselves to the interval $(0,1)$. By a direct computation, we have
\begin{align*}
F(0) = \beta m^\parallel > 0, \ \lim_{x \to 1^{-}}  F(x) = - \infty .
\end{align*}
Since $F$ is a continuous function on $[0,1)$, this implies that $F$ must have at least one root on the interval $(0,1)$. Next, by direct computation we have
\begin{align*}
F''(x) &= -\frac{1}{2a}\frac{x^3}{(1 + a^2 - x^2)^\frac{3}{2} (\sqrt{1 + a^2 - x^2} - a)^2} \\
&-\frac{1}{2a} \frac{2 x^3}{(1 + a^2 - x^2) (\sqrt{1 + a^2 - x^2} - a)^3} \\
&-\frac{1}{2a} \frac{3x}{(1 + a^2 - x^2)^\frac{1}{2} (\sqrt{1 + a^2 - x^2} - a)^2} \\
\end{align*}
where $a = \frac{1}{2 \beta m^\perp} > 0$. It is clear that $F''(x) < 0$ for all $x \in (0,1)$. The inequality for the second derivative implies that $F'$ is a strictly decreasing function on the interval $(0,1)$. If $F'(0) \leq 0$, then $F'(x) < 0$ for all $x \in (0,1)$, and the function $F$ must be strictly decreasing on $(0,1)$. In this case, $F$ has a single root in $(0,1)$. If $F'(0) > 0$, then by the monotonicity of $F'$, there must exist some intermediate value $z \in (0,1)$ such that $F'(z) = 0$. If there is no such value, then $F'(x) > 0$ for all $x \in (0,1)$ which would imply that $F$ is a strictly increasing function, but since $F$ starts at a positive value and must decrease without bound as it approaches $1$, this is a contradiction. Thus there must exist an intermediate value $z$ as described, and $F(z) > 0$. Now, by again applying the monotonicity of $F$ on $(z, 1)$, and that $F$ decreases without bound as it approaches $1$, we see that there must exist a unique root of $F$ on $(z, 1)$, and there are no other roots on $(0,z)$. This analysis shows that if $m^\parallel > 0$, then $F$ has a unique root on the interval $(0,1)$.
\\
\\
If $m^\parallel < 0$, define $G : (0,1) \to \mathbb{R}$ by 
\begin{align*}
G(x) := - F(- x) = \beta J x + \beta \left( - m^\parallel \right) - \frac{x}{2 \left( \frac{1}{2 \beta m^\perp} \right) \left(\sqrt{1  + \left( \frac{1}{2 \beta m^\perp} \right)^2 - {x}^2} -  \frac{1}{2 \beta m^\perp} \right)} .
\end{align*}
The analysis that we did for $F$ when $m^\parallel > 0$ holds verbatim for the function $G$ since $(- m^\parallel) > 0$ in this case. As such, the function $G$ has a unique root on $(0,1)$, which implies that the function $F$ has a unique root on $(-1,0)$ when $m^\parallel < 0$.
\\
\\
Since the points where the global maximum value of $\psi$ is attained are also critical points, the above analysis shows that if $m^\perp \not = 0$, then there exists only one unique point at which the maximum value is attained.
\end{proof}
\noindent
The second result concerns the case where $m^\parallel = 0$. 
\begin{lemma}\label{max_point2} Let $h$ be a strongly varying external field and suppose that $m^\parallel = 0$.
\\
\\
One of the following three holds:
\begin{enumerate}
\item If $m^\perp \geq J$, then for all $\beta > 0$, there exists a unique global maximum point of $\psi^{\beta,J,m}$ given by
\begin{align*}
x^0 = 0, \ y^0 = \sqrt{1 + \left( \frac{1}{2 \beta m^\perp}\right)^2} - \frac{1}{2 \beta m^\perp} .
\end{align*}
\item If $m^\perp < J$, then for all $\beta \leq \frac{J}{(J - m^\perp) (J + m^\perp)}$, there exists a unique global maximum point of $\psi^{\beta,J,m}$ given by
\begin{align*}
x^0 = 0, \ y^0 = \sqrt{1 + \left( \frac{1}{2 \beta m^\perp}\right)^2} - \frac{1}{2 \beta m^\perp} .
\end{align*}
\item If $m^\perp < J$, then for all $\beta > \frac{J}{(J - m^\perp) (J + m^\perp)}$, there exists two global maximum points of $\psi^{\beta,J,m}$ given by
\begin{align*}
x^{\pm} = \pm \sqrt{1 - \frac{1}{\beta J} - \frac{{m^\perp}^2}{J^2} }, \ y^{\pm} = \frac{m^\perp}{J} .
\end{align*} 
\end{enumerate}
\end{lemma}
\begin{proof}
We begin by noting that any local extrema of differentiable function on an open set are also critical points. The critical point equation is given by
\begin{align*}
 \nabla [\psi]  (x^*,y^*) = 0 \iff x^* \left( \beta J - \frac{1}{1 - {x^*}^2 - {y^*}^2} \right) = 0, \ \beta m^\perp - \frac{y^*}{1 - {x^*}^2 - {y^*}^2} = 0 .
\end{align*}
Let us now proceed case by case. Recall that the second component of the critical point equation can be used to solve for the value of $y^*$ in terms of $x^*$. The first component of the critical point equation can then be solved and the three possible candidates for a global maximum point on whether or not $x^*$ vanishes. Let us denote $x^0 = 0$, $x^+ \in (0,1)$, and $x^- \in (-1,0)$, for the possible values of $x^*$. Let us emphasize that the existence of these solutions depends on the parameters $\beta$, $J$, and $m^\perp$, and that the candidate global maximum points do not all exist simultaneously.
\\
\\
First, let us consider $x^* = x^0 = 0$. The corresponding $y^0$ value is given by
\begin{align*}
y^0 = \sqrt{1 + \left( \frac{1}{2 \beta m^\perp}\right)^2} - \frac{1}{2 \beta m^\perp} .
\end{align*} 
Furthermore, one can verify that $H_{12} [ \psi] (z^0) = 0 = H_{21} [ \psi ] (z^0)$, and $H_{22} [ \psi ] (z^0) < 0$, where $H[\psi] (z^0)$ is the Hessian of $\psi$ evaluated at the point $z_0$. This implies that the sign of $H_{11} [ \psi ] (z^0)$ determines whether or not this is a true local maximum. A direct computation shows that
\begin{align*}
H_{11} [ \psi ](z^0) = \beta J -  \sqrt{\beta^2 {m^\perp}^2 + \left( \frac{1}{2} \right)^2} - \frac{1}{2} .
\end{align*}
One can immediately see that if $m^\perp \geq J$, then $H_{11} [ \psi] (z^0) < 0$ for all $\beta > 0$. If $m^\perp < J$, we can solve the condition $H_{11} [ \psi] (z^0) \leq 0$ by considering it in terms of a quadratic in $\beta$. This analysis yields the following
\begin{align*}
H_{11} [ \psi ] (z^0) \leq 0 \iff 0 < \beta \leq \frac{J}{(J - m^\perp)(J + m^\perp)} .
\end{align*}
Note that in the above equation the equality only occurs when both sides of the equivalence have equality. Let us now summarize these properties. If $m^\perp > J$, $\psi$ has a local maximum for all $\beta > 0$ which can be verified by the negative definiteness of its Hessian at the the critical point. If $m^\perp < J$, $\psi$ has a local maximum which can be verified by the negative definiteness of its Hessian at the critical point, but this local maximum only exists for $\beta$ belonging to the finite interval shown before. If $m^\perp = J$, the Hessian matrix at the critical point has determinant $0$, and, without further analysis, we cannot conclude whether or not this point is a maximum, minimum, or a saddle point. We will return to the case of $m^\perp = J$ later.
\\
\\
Let us now consider the solutions $z^\pm$ which satisfy
\begin{align*}
\beta J - \frac{1}{1 - {x^{\pm}}^2 - {y^{\pm}}^2} = 0 .
\end{align*}
By plugging this equation into the second component of the critical point equation, we immediately find that
\begin{align*}
y^\pm = \frac{m^\perp}{J} .
\end{align*}
Because of this property, it follows that if $m^\perp \geq J$, then there are no solutions to the critical point equation in $B(0,1)$, where $x^* \not = 0$. If $m^\perp < J$, we can solve for the value of $x^{\pm}$ from the first component of the critical point equation. 
When it exists, this value is given by
\begin{align*}
x^{\pm} = \pm \sqrt{1 - \frac{1}{\beta J} - \frac{{m^\perp}^2}{J^2} } .
\end{align*}
In order for $x^{\pm} \in (-1,1)$, we must have
\begin{align*}
 \beta >  \frac{J}{(J - m^\perp) (J + m^\perp)} .
\end{align*}
We will consider the Hessian matrix at this point to verify that these are local maxima. A direct computation shows that 
\begin{align*}
\det \left( H [ \psi ] (z^\pm) \right) = \frac{1 + x^2 + y^2}{(1 - x^2 - y^2)^3} - \beta J \left(\frac{1}{1 - x^2 - y^2} + \frac{2 y^2}{(1 - x^2 - y^2)^2} \right) .
\end{align*} 
In our case, we have
\begin{align*}
\det \left( H [\psi] (z^\pm) \right)  = 2 (\beta J)^3 - 2 (\beta J)^2 - 2 \beta J (\beta m^\perp)^2  = 2 (\beta J)^2 \left( \beta J \left( 1 - \left( \frac{m^\perp}{J}\right)^2 \right) - 1 \right).
\end{align*}
It follows that
\begin{align*}
\det \left( H [ \psi ] (z^\pm) \right)  > 0 \iff \beta > \frac{J}{(J - m^\perp) (J + m^\perp)} .
\end{align*}
Since $H_{22}[\psi] (z^\pm) < 0$, this confirms that when $m^\perp < J$, and 
\begin{align*}
\beta > \frac{J}{(J - m^\perp) (J + m^\perp)} ,
\end{align*}
the given $z^\pm$ are local maxima of $\psi$.
\\
\\
Finally, one should observe that the parameter ranges for which the respective maximizing points are derived from the critical point equation are disjoint. Since each maximizing point of $\psi$ must be achieved at a critical point, it follows that these points are the only possible candidates for global maximizing points. It follows that each of the solutions to the critical point equations are thus all global maximizing points and the result follows.
\end{proof}
\subsection{Asymptotic analysis of sequences of local maximizing points} \label{asymp_analysis}
\noindent
In this section, we will use the multivariate Taylor's formula as presented in \cite[Chapter 8]{Folland1999}. 
\\
\\
The following result shows that the conditioned FVGS are always weakly convergent.
\begin{proof}[Proof of \Cref{condition_conv}]
Observe that
\begin{align*}
\mu_n^{\beta, J, h, \pm} = \frac{1}{\rho_n^{\beta, J, h} (B_{\pm} (0,1))} \int_{B_\pm(0,1)} \rho_n^{\beta, J, h}(dz) \nu_n^{z,h} .
\end{align*}
If we denote the probability measure $\rho_n^{\beta, J, h, \pm}$ on $B_{\pm} (0,1)$ given by its density
\begin{align*}
\rho_n^{\beta, J, h, \pm}(dz) := \frac{\rho_n^{\beta, J, h}(dz) \mathbbm{1}(z \in B_\pm (0,1))}{\rho_n^{\beta, J, h} (B_{\pm} (0,1))} , 
\end{align*}
then by replicating the proofs of \Cref{single_point}, it follows that
\begin{align*}
\lim_{n \to \infty} \rho_n^{\beta, J, h, \pm} = \delta_{z^{\pm}} .
\end{align*}
From this weak limit and the uniform convergence provided by \Cref{uni_conv}, the result follows.
\end{proof}
\noindent
The following result provides the construction of a suitably converging sequence of local maximizing points of the exponential tilting functions.
\begin{lemma}\label{max_seq} Let $h$ be a strongly varying external field.
\\
\\
It follows that for every global maximizing point $z^*$ of $\psi^{\beta,J,m}$ there exists a sequence of points  $\{ z_n^*\}_{n \in \mathbb{N}}$ such that $z_n^*$ are local maximum points of $\psi_n^{\beta, J, h}$ satisfying $z_n^* \to z^*$ as $n \to \infty$ and $\nabla [\psi_n^{\beta, J, h} ] (z_n^*) = 0$ for large enough $n$.
\end{lemma}
\begin{proof}
Let $B(z^*, \delta) \subset B(0,1)$ be a small enough so that $z^*$ is the only global maximum point of $\psi$ in $\overline{B}(z^*, \delta)$. By compactness, the functions $\psi_n$ must attain their maxima $z_n^*$ in $\overline{B}(z^*, \delta)$. Let $ \{ z^*_{n_k} \}_{k \in \mathbb{N}}$ be any convergent subsequence of $\{ z^*_n\}_{n \in \mathbb{N}}$ and let us denote $z^*_0 := \lim_{k \to \infty} z^*_{n_k}$. By definition, for all $k \in \mathbb{N}$, we have
\begin{align*}
\psi_{n_k}(z^*) \leq \psi_{n_k} (z_{n_k}^*) .
\end{align*}
Using the uniform continuity of $\psi_{n_k}$, taking the limit $k \to \infty$, we have
\begin{align*}
\psi(z^*) \leq \psi(z_0^*) .
\end{align*}    
Since $z^*$ is the unique global maximum point of $\psi$ in $\overline{B}(z^*, \delta)$, it follows that $z^* = z_0^*$. Since the sequence $\{ z_n^*\}_{n \in \mathbb{N}}$ is bounded, and this property holds for any convergent subsequence, it follows that $z_n^* \to z^*$. Because $z_n^* \to z^*$, it follows that for large enough $n$ the $z_n^*$ must belong to $B(z^*, \delta)$ and not $\partial B(z^*, \delta)$ which guarantees that $\nabla [\psi_n] (z_n^*) = 0$.
\end{proof}
\noindent
For the rest of the proofs in this section, $z^* \in B(0,1)$ will always be a global maximum point of $\psi^{\beta,J,m}$ such that the Hessian $H[\psi^{\beta,J,m}] (z^*)$ at $z^*$ is negative definite, and there exists a sequence of points $\{ z_n^* \}_{n \in \mathbb{N}}$ satisfying $z_n^* \to z^*$ in the limit as $n \to \infty$ and $\nabla [\psi_n^{\beta,J,h}] (z_n^*) = 0$ for large enough $n$. 
\\
\\
The following result characterizes the rate of convergence of the sequence of local maximizing points in terms of the sample mean and sample standard deviation.
\begin{lemma}\label{diff_eq} Let $h$ be a strongly varying external field. 
\\
\\
It follows that there exist a sequence of invertible matrices $\{ H_n \}_{n \in \mathbb{N}}$ such that $H_n \to H[\psi^{\beta, J,m}] (z^*)$ as $n \to \infty$ and we have
\begin{align*}
z_n^* - z^* = - \beta  H_n^{-1} (m_n^h - m) 
\end{align*}
for large enough $n$.
\end{lemma}
\begin{proof}
Using the respective critical point equations for $\psi$ and $\psi_n$, we can form the following pair of difference equations
\begin{align*}
\beta (m_n^\parallel - m^\parallel) &=  \frac{x_n^*}{1 - {x_n^*}^2 - {y_n^*}^2} - \beta J x_n^*  - \left( \frac{x^*}{1 - {x^*}^2 - {y^*}^2} - \beta J x^* \right) , \\
\beta(m_n^\perp - m^\perp) &=  \frac{y_n^*}{1 - {x^*_n}^2 - {y_n^*}^2} - \frac{y^*}{1 - {x^*}^2 - {y^*}^2}.
\end{align*}
We consider two functions $C_1, C_2 : B(0,1) \to \mathbb{R}$ given by
\begin{align*}
C_1(x,y) &:= \frac{x}{1 - x^2 - y^2} - \beta J x  - \left( \frac{x^*}{1 - {x^*}^2 - {y^*}^2} - \beta J x^* \right), \\ C_2(x,y) &:= \frac{y}{1 - x^2 - y^2} - \frac{y^*}{1 - {x^*}^2 - {y^*}^2} .
\end{align*}
Observe that for any multi-index $\alpha \in \mathbb{N}^2$ such that $|\alpha| \geq 1$, we have $(\partial^{\alpha} C_1)(x,y) = - \left( \partial^\alpha \partial_1 \psi \right)(x,y)$ and $(\partial^{\alpha} C_2)(x,y) = - \left( \partial^\alpha \partial_2 \psi \right)(x,y)$. Furthermore, if we define $C : B(0,1) \to \mathbb{R}^2$ by $C(x,y) := (C_1 (x,y), C_2(x,y))$, then $D[C] = - H \left[ \psi \right]$.
\\
\\
The difference equations can thus be expressed as
\begin{align*}
\beta (m_n - m) = C(z_n^*) . 
\end{align*}
If $\det \left( H \left[ \psi \right] (z^*) \right) < 0$, then $D[C] (z^*)$ is a real symmetric positive definite matrix. By considering the Taylor series of $C_1$ and $C_2$ developed around the point $z^*$, we have
\begin{align*}
\beta (m_n - m) = \left( - H[\psi](z^*) + S(z_n^* - z^*) \right) (z_n^* - z^*),
\end{align*} 
where
\begin{align*}
S_{ij} (z_n^* - z^*) = - \sum_{k=1}^2  \int_0^1 dt \ (1 - t) (\partial_k \partial_j \partial_i \psi) (z^* + t (z_n^* - z^*))  (z_n^* - z^*)_k . 
\end{align*}
Because $S(z_n^*) \to 0$ as $n \to \infty$, and $H[\psi](z^*)$ is negative definite, it follows that $- H[\psi](z^*) + S(z_n^*)$ is invertible for large enough $n \in \mathbb{N}$. For large enough $n \in \mathbb{N}$, we thus have
\begin{align*}
z_n^* - z^* = - \beta \left( H[\psi](z^*) - S(z_n^*) \right)^{-1} (m_n - m) .
\end{align*}
Now, setting
\begin{align*}
H_n := H[\psi](z^*) - S(z_n^*),
\end{align*}
the result follows.
\end{proof}
\noindent
The following result is an adaptation of the Laplace method to this specific model.
\begin{lemma} \label{laplace_asymp} Let $h$ be a strongly varying external field. 
\\
\\
Suppose that there exists an open set $B \subset B(0,1)$ such that $z^*$ is the unique maximum point of $\psi^{\beta,J,m}$ in $B$.
\\
\\
It follows that
\begin{align*}
\lim_{n \to \infty} \frac{n \int_{B} dz \ e^{2 \beta J x^2 + 4 \beta \left< m_n^h, z \right>} e^{(n - 4)\psi^{\beta,J,m}_n(z)}}{e^{n \psi^{\beta,J,m}_n(z_n^*)}} = \frac{1}{(1 - || z^* ||^2)^2}\int_{\mathbb{R}^2} dz \ e^{\frac{1}{2} \left< z, H[\psi^{\beta,J,m}](z^*) z \right>}  . 
\end{align*}
\end{lemma}
\begin{proof}
First, using the fact that $z_n^* \to z^*$, for large enough $n$, it follows that there exists $0 < c < \delta < C$ such that $B(z^*, c) \subset B(z_n^*, \delta) \subset \overline{B}(z^*, C) \subset B$, where $\delta$ and $c$ are fixed but small. Denote $B_{n,\delta} = B(z_n^*, \delta)$. We have
\begin{align*}
&\int_{B \setminus B_{n,\delta}} dz \ e^{2 \beta J x^2 + 4 \beta \left< m_n, z \right>} e^{(n - 4) (\psi_n(z) - \psi_n(z_n^*))} \\ &\leq |B (0,1)| e^{\beta \left( 2J + 4 |m_n^\parallel| + 4 m_n^\perp \right)} e^{(n - 4) \Delta_n},
\end{align*}
where
\begin{align*}
\Delta_n := \sup_{z \in B \setminus B_{n, \delta}} \psi_n(z) - \sup_{z \in B} \psi_n(z) < \sup_{z \in B \setminus B(z^*, c)} \psi_n(z) - \sup_{z \in B} \psi_n(z).
\end{align*}
Using the uniform convergence of $\psi_n \to \psi$, one sees that 
\begin{align*}
\lim_{n \to \infty} \Delta_n := \Delta < 0,
\end{align*}
and it follows that
\begin{align*}
\int_{B \setminus B_{n,\delta}} dz \ e^{2 \beta J x^2 + 4 \beta \left< m_n, z \right>} e^{(n - 4) (\psi_n(z) - \psi_n(z_n^*))} = O\left( e^{(n - 4) \Delta} \right) .
\end{align*}
\\
\\ 
Next, by changing variables, we have
\begin{align*}
&\int_{ B_{n, \delta}} dz \ e^{2 \beta J x^2 + 4 \beta \left< m_n, z \right>} e^{(n - 4) (\psi_n(z) - \psi_n(z_n^*))} \\
&= \frac{1}{n-4} \int_{B(0, \delta \sqrt{n-4})} dz \ e^{2 \beta J \left( x^* + \frac{x}{\sqrt{n - 4}}\right)^2 + 4 \beta \left< m_n, z^* + \frac{z}{\sqrt{n - 4}} \right>} e^{(n-4) (\psi_n (z_n^* + \frac{z}{\sqrt{n-4}}) - \psi_n(z_n^*))}  .
\end{align*}
Now, note that
\begin{align*}
(n - 4) \left( \psi_n \left(z_n^* + \frac{z}{\sqrt{n-4}} \right) - \psi_n(z_n^*)  \right) &= \frac{1}{2} \left< z, H[\psi] (z_n^*)z \right>  \\
&+ \frac{1}{(n - 4)^{\frac{1}{2}}} \sum_{|\alpha| = 3} R_\alpha \left( z_n^*, \frac{z}{\sqrt{n - 4}}\right) z^\alpha, 
\end{align*}
where $R_\alpha$ is given by
\begin{align*}
R_\alpha \left( z_n^*, \frac{z}{\sqrt{n - 4}}\right) = \frac{|\alpha|}{\alpha!} \int_0^1 dt \ (1 - t)^{|\alpha| - 1} \left( \partial^\alpha \psi \right) \left( z_n^* + t \frac{z}{\sqrt{n - 4}}\right) . 
\end{align*}
Because $z_n^* \to z^*$, for large enough $n$ there exists $\delta_1 > \delta$ such that $B(z^*, \delta_1) \subset B$ and we have
\begin{align*}
\left| \left( \partial^\alpha \psi \right) \left( z_n^* + t \frac{z}{\sqrt{n - 4}}\right) \right| \leq \max_{|\alpha| = 3, z' \in B(z^*, \delta_1)} |\partial^\alpha \psi (z')| 
\end{align*}
for $z \in B(0, \delta \sqrt{n - 4})$. It follows that
\begin{align*}
\left| \frac{1}{(n - 4)^{\frac{1}{2}}}\sum_{|\alpha| = 3} R_\alpha \left( z_n^* + \frac{z}{\sqrt{n-4}} \right) z^\alpha \right| &\leq \max_{|\alpha| = 3, z' \in B(z^*, \delta_1)} |\partial^\alpha \psi (z')| \frac{1}{(n - 4)^{\frac{1}{2}}} \sum_{|\alpha| = 3} |z^\alpha| \\
&\leq \delta  D \left( \max_{|\alpha| = 3, z' \in B(z^*, \delta_1)} |\partial^\alpha \psi (z')| \right) || z ||^2 .
\end{align*}
for some fixed constant $D > 0$ and every $z \in B(0, \delta \sqrt{n - 4})$. It follows that for large enough $n$ and small enough $\delta > 0$, we have
\begin{align*}
\mathbbm{1}(z \in B(0, \delta \sqrt{n - 4}))e^{(n-4) (\psi_n (z_n^* + \frac{z}{\sqrt{n-4}}) - \psi_n(z_n^*))} \leq e^{\frac{1}{2} \left< z, \left( H[\psi](z_n^*) + 2 \delta D \max_{|\alpha| = 3, z' \in B(z^*, \delta_1)} |\partial^\alpha \psi (z')| I \right) z \right>} ,
\end{align*}
where $I$ is the identity matrix of dimension $2 \times 2$.   Because $H[\psi](z_n^*) \to H[\psi] (z^*)$, and $\delta > 0$ can be chosen arbitrarily small but fixed, it follows that for large enough $n$ and small enough $\delta > 0$, there exists a positive definite quadratic form $Q$ such that
\begin{align*}
e^{\left< z, \left( H[\psi](z_n^*) + \delta D \max_{|\alpha| = 3, z' \in B(z^*, \delta_1)} |\partial^\alpha \psi (z')| I \right) z \right>} \leq e^{- \left<z, Q z \right>} .
\end{align*}
Finally, it is clear that for large enough $n$ and fixed $\delta > 0$, there exists a constant $E > 0$ such that
\begin{align*}
\mathbbm{1}(z \in B(0, \delta \sqrt{n - 4}))e^{2 \beta J \left( x^* + \frac{x}{\sqrt{n - 4}}\right)^2 + 4 \beta \left< m_n, z^* + \frac{z}{\sqrt{n - 4}} \right>} \leq e^E .
\end{align*}
From these observations, by dominated convergence, it follows that
\begin{align*}
&\lim_{n \to \infty}  \frac{(n-4) \int_{B_{n, \delta}} dz \ e^{2 \beta J x^2 + 4 \beta \left< m_n^h, z \right>} e^{(n - 4)\psi^{\beta,J,m}_n(z)}}{e^{(n-4) \psi^{\beta,J,m}_n(z_n^*)}} \\ &= e^{2 \beta J {x^*}^2 + 4 \beta \left< m, z^* \right>} \int_{\mathbb{R}^2} dz \ e^{\frac{1}{2} \left< z, H[\psi](z^*) z \right>} .
\end{align*}
Now, including the exponentially decreasing integral over $B \setminus B_{n, \delta}$, it follows that
\begin{align*}
&\lim_{n \to \infty}  \frac{(n-4) \int_{B} dz \ e^{2 \beta J x^2 + 4 \beta \left< m_n^h, z \right>} e^{(n - 4)\psi^{\beta,J,m}_n(z)}}{e^{(n-4) \psi^{\beta,J,m}_n(z_n^*)}} \\ &= e^{2 \beta J {x^*}^2 + 4 \beta \left< m, z^* \right>} \int_{\mathbb{R}^2} dz \ e^{\frac{1}{2} \left< z, H[\psi](z^*) z \right>} .
\end{align*}
By plugging this formula into the original desired limit, and cancelling out terms, the result follows.
\end{proof}
\noindent
The following result concerns the asymptotics of the difference $\psi_n^{\beta, J, h} (z_n^\pm) - \psi^{\beta,J,m} (z^\pm)$.
\begin{lemma} \label{diff_eq2} Let $h$ be a strongly varying external field. 
\\
\\
It follows that there exists two sequence of matrices $\{ Q_n \}_{n \in \mathbb{N}}$ and $\{ H_n \}_{n \in \mathbb{N}}$ such that $Q_n \to  H[\psi^{\beta,J,m}] (z^*)$  and $H_n \to H[\psi^{\beta,J,m}] (z^*)$ as $n \to \infty$ and we have
\begin{align*}
\psi^{\beta,J,m}_n(z_n^*) - \psi^{\beta,J,m}(z^*) &= \beta \left< m_n^h - m, z^* \right> - \frac{\beta^2}{2} \left< m_n^h - m, Q_n^{-1} (m_n^h - m) \right> \\ &+ \sum_{|\alpha| = 3} R_\alpha (-\beta H_n^{-1} (m_n^h - m)) (H_n^{-1} (m_n^h - m))^\alpha
\end{align*}
for large enough $n \in \mathbb{N}$, where $R_\alpha$ are functions which are continuous at $0$.
\end{lemma}
\begin{proof}
We begin with the simple observation
\begin{align*}
\psi_n(z_n^*) - \psi(z^*) = \psi_n(z_n^*) - \psi(z_n^*) 
+ \psi(z_n^*) - \psi(z^*) .
\end{align*}
The individual differences can be evaluated separately. For the first difference, we have
\begin{align*}
\psi_n(z_n^*) - \psi(z_n^*) &= \beta \left< m_n - m, z^* \right> + \beta \left< m_n - m, z_n^* - z^* \right> 
\end{align*}
For the second, we have
\begin{align*}
\psi(z_n^*) - \psi(z^*) = \frac{1}{2} \left< z_n^* - z^*, H[\psi](z^*) (z_n^* - z^*) \right> + \sum_{|\alpha| = 3} R_\alpha (z_n^* - z^*) (z_n^* - z^*)^\alpha,
\end{align*}
where
\begin{align*}
R_\alpha (z_n^* - z^*) = \frac{|\alpha|}{\alpha!} \int_0^1 dt \ (1 - t)^{|\alpha| - 1} (\partial^{\alpha} \psi) (z^* + t (z_n^* - z^*)) . 
\end{align*}
Combining the differences, we see that
\begin{align*}
\psi_n(z_n^*) - \psi(z^*) &= \beta \left< m_n - m, z^* \right> + \beta \left< m_n - m,  z_n - z^* \right> + \frac{1}{2} \left< z_n^* - z^*, H[\psi] (z^*) (z_n - z^*) \right>  \\ &+ \sum_{|\alpha| = 3} R_\alpha (z_n^* - z^*) (z_n^* - z^*)^\alpha .
\end{align*}
Now, we apply \Cref{diff_eq} to convert the $z_n - z^*$ terms to $m_n - m^*$ terms to get
\begin{align*}
\psi_n(z_n^*) - \psi(z^*) &= \beta \left< m_n - m, z^* \right> - \frac{\beta^2}{2} \left< m_n - m, Q_n^{-1} (m_n - m) \right> \\ &+ \sum_{|\alpha| = 3} R_\alpha (-\beta H_n^{-1} (m_n - m)) (H_n^{-1} (m_n - m))^\alpha
\end{align*}
where
\begin{align*}
Q_n := (2 H_n^{-1} - H_n^{-1} H[\psi] (z^*) H_n^{-1})^{-1} ,
\end{align*}
from which the result follows.
\end{proof}
\noindent
The following result characterizes the asymptotics of the weight $W_n^{\beta, J, h, +}$ via the rate of convergence of $m_n^h - m$.
\begin{proof}[Proof of \Cref{weight_asymp}]
By applying \Cref{laplace_asymp} to the form of the weight $W_n^{+}$ given in \Cref{weight_rep}, there exists of a sequence $\{ a_n \}_{n \in \mathbb{N}}$ such that $a_n \to 1$ and
\begin{align*}
W_n^{+} = \frac{1}{1 + a_n e^{n (\psi_n (z^-_n) - \psi_n (z_n^+))}} .
\end{align*}
Suppose that the asymptotic relation $\lim_{n \to \infty} n^\delta (m_n - m) \to \gamma \in \mathbb{R}^2$ holds for some $\delta > 0$. Using the asymptotics from \Cref{diff_eq2} applied the differences in the exponential of the previous limit, we have
\begin{align*}
\lim_{n \to \infty} n^\delta (\psi_n (z_n^+) - \psi (z^+) - (\psi_n(z_n^-) - \psi (z^-))) = - 2 \beta x^+ \gamma^\parallel .
\end{align*}
Now, by considering the specific $\delta$ and $\gamma$ given in the enumerated assumptions, the result follows from the following limit
\begin{align*}
\lim_{n \to \infty} W_n^+ = \frac{1}{1 + \lim_{n \to \infty} e^{n^{1 - \delta} n^\delta (\psi_n (z_n^+) - \psi (z^+) - (\psi_n(z_n^-) - \psi (z^-)))}} .
\end{align*}
\end{proof}
\subsection{Random external fields, random walks, and chaotic size dependence} \label{csd_sec}
\noindent
The following two results concern the measurability of the map $\omega \in \Omega \mapsto h \mapsto \mu_n^{\beta,J,h}$.
\begin{lemma} \label{meas}
The mapping $h \mapsto \mu_n^{\beta, J, h}$ is continuous.
\end{lemma}
\begin{proof}
Since the topologies on both spaces are given by metrics and the local continuous bounded functions on $\mathbb{R}^\mathbb{N}$ are convergence determining, it is enough to show that if $\{ h_k \}_{k \in \mathbb{N}}$ is a sequence of elements in $\mathbb{R}^\mathbb{N}$ such that $h_k \to h$ then we have $\mu_n^{\beta, J, h_k} [f] \to \mu_n^{\beta, J, h} [f]$ as $k \to \infty$ for any local continuous bounded function $f \in C_b(\mathbb{R}^\mathbb{N})$. This follows by using dominated convergence to change the order of the limit and integration in the definition of $\mu_n^{\beta,J,h}$ given in \Cref{fvgs2}.
\end{proof}
\begin{lemma} \label{cont}
Suppose that $h : \Omega \to \mathbb{R}^\mathbb{N}$ is measurable.
\\
\\
It follows that $\omega \mapsto h \mapsto \mu_n^{\beta, J, h} : \Omega \to \mathcal{M}_1(\mathbb{R}^\mathbb{N})$ is measurable. 
\end{lemma}
\begin{proof}
By \Cref{meas}, the mapping $h \mapsto \mu_n^{\beta, J, h}$ is continuous and thus measurable. The result follows since the composition of measurable functions is measurable.
\end{proof}
\noindent
The following lemma relates the limiting properties of two-dimensional random walks with step length $(h_0, h_0^2)$ to the limiting properties of the sequence $\{ m_n^h \}_{n \in \mathbb{N}}$.
\begin{lemma} \label{limit_theorems}  Let $h$ be a random external field which satisfies (A1) and (A2).
\\
\\
The sequence of random variables $\{ m_n^h \}_{n \in \mathbb{N}}$ has the following limit properties:
\begin{enumerate}
\item A strong law of large numbers of the form
\begin{align*}
\lim_{n \to \infty} m_n^h \to m = (\mathbb{E} h_0, \sqrt{\mathbb{V} h_0})
\end{align*}
almost surely.
\item A central limit theorem of the form
\begin{align*}
\lim_{n \to \infty}  (h, \sqrt{n} (m_n^h - m)) = (h,G) 
\end{align*}
weakly, where $G$ is a non-degenerate $2$-dimensional Gaussian random variable independent of $h$.
\item 
A recurrence result of the form
\begin{align*}
\left\{ \left( p_1, \frac{p_2}{2 m^\perp} - \frac{p_1 m^\parallel}{m^\perp} \right) : p \in P\right\} \subset L \left( \left\{ n(m_n^h - m) \right\}_{n=1}^\infty \right),
\end{align*}
almost surely, where $P$ is the set of possible values of the random walk with step length $(h_0 - \mathbb{E}h_0, h_0^2 - \mathbb{E} h_0^2)$
\end{enumerate}
\end{lemma}
\begin{proof}
Denote the random walk $\{ S_n \}_{n=1}^\infty$ with step length $(h_0, h_0^2)$ as a sequence of $\mathbb{R}^2$-valued random variables by setting
\begin{align*}
\left( S_n \right)_1 := \sum_{i=1}^n h_i, \ \left( S_n \right)_2 :=  \sum_{i=1}^n h_i^2 .
\end{align*}
The mean $\mathbb{E} S_n$ of $S_n$ is given by $\mathbb{E} S_n = n \left( \mathbb{E} h_0, \mathbb{E} h_0^2 \right)$. It follows that
\begin{align*}
m_n^\parallel =  \left( \frac{S_n}{n} \right)_1, \ m_n^\perp = \sqrt{\left( \frac{S_n}{n} \right)_2 - \left( \left( \frac{S_n}{n} \right)_1 \right)^2} .
\end{align*}
By the the strong law of large numbers, it follows that
\begin{align*}
\lim_{n \to \infty} m_n = \left( \mathbb{E} h_0, \sqrt{\mathbb{V} (h_0)}\right)
\end{align*}
almost surely. We can thus set $m^\parallel := \mathbb{E} h_0$ and $m^\perp := \sqrt{\mathbb{V}h_0}$. 
\\
\\
By the central limit theorem, it follows that
\begin{align*}
\lim_{n \to \infty} \left( h, \sqrt{n} \left( \frac{S_n}{n}  - (\mathbb{E} h_0, \mathbb{E} h_0^2) \right) \right):= (h, G_0)
\end{align*}
in distribution, where $G_0$ is a $2$-dimensional Gaussian random variable with mean $0$ and covariance matrix $\Sigma_0$ independent of $h$ given by
\begin{align*}
\Sigma_0 := \begin{bmatrix}
\mathbb{V} h_0 & \mathbb{E} h_0^3 - \mathbb{E} h_0 \mathbb{E} h_0^2 \\
\mathbb{E} h_0^3 - \mathbb{E} h_0 \mathbb{E} h_0^2  & \mathbb{V} h_0^2
\end{bmatrix} .
\end{align*}
Let $D := \{ (x,y) \in \mathbb{R}^2 : y - x^2 > 0 \}$. Define a function $f : D \to \mathbb{R}^2$ by setting $f(x,y) := (x, \sqrt{y - x^2})$. Observe that $f\left( \frac{S_n}{n} \right) = m_n$. By the multivariate delta method, it follows that
\begin{align*}
\lim_{n \to \infty} \left(h, \sqrt{n} (m_n - m) \right) &= \lim_{n \to \infty} \left(h, \sqrt{n} \left( f \left( \frac{S_n}{n} \right) - f(\mathbb{E} h_0, \mathbb{E} h_0^2) \right) \right) \\ &= \left( h, D[f](\mathbb{E} h_0, \mathbb{E} h_0^2) G_0 \right)
\end{align*}
weakly where $G_0$ is again independent of $h$. Let us then denote $G := D[f](\mathbb{E} h_0, \mathbb{E} h_0^2) G_0$ and note that $G$ is a $2$-dimensional Gaussian random variable with mean $0$ and covariance matrix $\Sigma$ defined by 
\begin{align*}
\Sigma := D[f](\mathbb{E} h_0, \mathbb{E} h_0^2) \Sigma_0 D[f](\mathbb{E} h_0, \mathbb{E} h_0^2)^T
\end{align*} 
independent of $h$. 
\\
\\
Denote the centred random walk $\{ S_n' \}_{n=1}^\infty$ defined by $S_n' := S_n - \mathbb{E} S_n$. A standard theorem of random walks, found for instance in \cite[Chapter 5]{Durrett2019}, states that if a $2$-dimensional scaled random walk $\{ \frac{S_n'}{\sqrt{n}} \}_{n=1}^\infty$ converges in distribution to a non-degenerate $2$-dimensional Gaussian random variable, then the random walk $\{ S_n' \}_{n=1}^\infty$ is recurrent. Thus the centred random walk $\{ S_n'\}_{n=1}^\infty$ is recurrent and we will denote its set of recurrent values by $P$. Note also that by the same standard theorem, the set $P$ is closed.
\\
\\
Since $P$ is a closed subset a separable space $\mathbb{R}^2$ it follows that there exists a sequence $\{ q_i \}_{i=1}^\infty$ of elements in $P$ which is dense in $P$. For $i,j \in \mathbb{N}$, define the set $\Omega_{i,j}$ by
\begin{align*}
\Omega_{i,j} := \left\{ || S'_{n} - q_i || < \frac{1}{j}  \text{ infinitely often}\right\} .
\end{align*}
By the definition of recurrence, we have $\mathbb{P}(\Omega_{i,j}) = 1$ for any $i,j \in \mathbb{N}$. It follows that the set $\Omega' \subset \Omega$ defined by
\begin{align*}
\Omega' := \bigcap_{i,j \in \mathbb{N}} \Omega_{i,j}
\end{align*}
satisfies $\mathbb{P}(\Omega') = 1$. Choose any realization of the random walk $\{ S_n \}_{n = 1}^\infty$ from the set $\Omega'$. For this realization, let $p \in P$ be any recurrent value. It follows that there exists a subsequence $\{ q_{i_k} \}_{k=1}^\infty$ such that $q_{i_k} \to p$. Using the sets $\Omega_{i_{k}, k}$, construct a subsequence $\{ n_k \}_{k=1}^\infty$ such that $|| S'_{n_k} - q_{i_k} || < \frac{1}{k}$. For such a subsequence, it is clear that $\lim_{k \to \infty} S'_{n_k} = p$. Such a construction is possible for any recurrent value $p$, and thus we have shown that
\begin{align*}
P \subset L \left( \{ S_n' \}_{n=1}^\infty \right)
\end{align*}
almost surely.
\\
\\
Returning now to the random variable $m_n - m$, let $\{ n_k \}_{k = 1}^\infty$ be a subsequence such that $S'_{n_k} \to p$. It follows that
\begin{align*}
n_k(m^\parallel_{n_k} - m^\parallel) = \left( S'_{n_k} \right)_1  
\end{align*}
and
\begin{align*}
n_k (m^\perp_{n_k} - m^\perp) = \frac{\left( S'_{n_k}\right)_2}{m^\perp_{n_k} + m^\perp} - \frac{\left( S'_{n_k} \right)_1 \left(m^\parallel_{n_k} + m^\parallel \right)}{m^\perp_{n_k} + m^\perp} .
\end{align*}
It follows that
\begin{align*}
\lim_{k \to \infty} n_k (m_{n_k} - m) = \left( p_1, \frac{p_2}{2 m^\perp} - \frac{p_1 m^\parallel}{m^\perp} \right) .
\end{align*}
Combining this with the previous result, we have 
\begin{align*}
\left\{ \left( p_1, \frac{p_2}{2 m^\perp} - \frac{p_1 m^\parallel}{m^\perp} \right) : p \in P\right\} \subset L \left( \left\{ n(m_n - m) \right\}_{n=1}^\infty \right)
\end{align*}
almost surely.
\end{proof}
\noindent
Using the strong law of large numbers for the sequence of random variables $\{ m_n^h \}_{n \in \mathbb{N}}$, we have the following self-averaging result for the limiting free energy.
\begin{proof}[Proof of \Cref{self_average}]
By reusing the proof of \Cref{main_result1}, and the strong law of large numbers from \Cref{limit_theorems}, it follows that
\begin{align*}
\lim_{n \to \infty} \frac{1}{n} \ln Z_n (\beta,J,h) = \sup_{z \in B(0,1)} \psi^{\beta,J,m} (z)
\end{align*}
almost surely. It is enough to prove that the sequence of random variables $\{ \frac{1}{n} \ln Z_n (\beta,J,h) \}_{n \in \mathbb{N}}$ is uniformly integrable. To that end, we immediately have the following inequality
\begin{align*}
\left| \frac{1}{n} \ln Z_n (\beta,J,h) - \frac{1}{n} \ln Z_n (\beta,J) \right| \leq \beta || m_n || ,
\end{align*}
where
\begin{align*}
Z_n (\beta,J) := \int_{B(0,1)} \frac{dz}{(1 - || z ||^2)^2} e^{n \left( \frac{\beta J x^2}{2} - \frac{1}{2} \ln (1 - || z ||^2) \right)} .
\end{align*}
We have
\begin{align*}
\sqrt{\mathbb{E} \left|  \frac{1}{n} \ln Z_n (\beta,J,h) \right|^2} \leq \frac{1}{n} \ln Z_n (\beta,J) + \beta \mathbb{E} h_0^2,
\end{align*}
since
\begin{align*}
\mathbb{E} m_n^2 = \mathbb{E} \frac{1}{n} \sum_{i=1}^n h_i^2  = \mathbb{E} h_0^2 .
\end{align*}
By using the Laplace method, it is clear that the sequence $\{ \frac{1}{n} \ln Z_n (\beta,J)\}_{n \in \mathbb{N}}$ is convergent and thus bounded. Combining all of these observations together, we find that the square moments of the sequence $\{ \frac{1}{n} \ln Z_n (\beta,J,h) \}_{n \in \mathbb{N}}$ are uniformly bounded and thus the sequence is uniformly integrable and the result follows.
\end{proof}
\noindent
Recall the phase characterization presented in \Cref{random_table} for the random external field. Using the first and the third properties from \Cref{limit_theorems}, we present the chaotic size dependence of the FVGS.
\begin{proof}[Proof of \Cref{main_result3}]
By \Cref{limit_theorems}, it follows that $m_n \to m = (\mathbb{E} h_0, \sqrt{\mathbb{V} h_0})$ almost surely. By applying \Cref{main_result1} to each realization in this set of probability $1$, all of the results follow except for the opposite inclusion of the limit points of $\{ \mu_n \}_{n \in \mathbb{N}}$ and the convex combinations of $\nu_\infty^{z^\pm}$.
\\
\\
For this result, by using the recurrence result from \Cref{limit_theorems}, for any $p \in P$, there exists a subsequence $\{ n_k \}_{k \in \mathbb{N}}$ such that
\begin{align*}
n_k(m_{n_k} - m) \to\left( p_1, \frac{p_2}{2 m^\perp} - \frac{p_1 m^\parallel}{m^\perp} \right) 
\end{align*} 
almost surely. It is not difficult to see that the collection of results leading to and including \Cref{conv_class} also hold for subsequences. The conditions of this subsequence version of \Cref{conv_class} are satisfied with $\delta = 1$ and $\gamma = \left( p_1, \frac{p_2}{2 m^\perp} - \frac{p_1 m^\parallel}{m^\perp} \right)$ which implies that
\begin{align*}
\mu_{n_k} \to \frac{1}{1 + e^{- 2 \beta x^+ p_1}} \nu_\infty^{z^+} + \frac{1}{1 + e^{ 2 \beta x^+ p_1}} \nu_\infty^{z^-} 
\end{align*}
almost surely. Since this holds for any $p \in P$, it follows that
\begin{align*}
\left\{ \frac{1}{1 + e^{- 2 \beta x^+ p_1}} \nu_\infty^{z^+} + \frac{1}{1 + e^{ 2 \beta x^+ p_1}} \nu_\infty^{z^-} \right\}_{p_1 \in \pi_1 (P)}  &\subset \mathcal{G}_\infty (\beta,J,h)
\end{align*}
almost surely. If $\pi_1 (P) = \mathbb{R}$, then for any $\alpha \in (0,1)$ there exists $p_1$ such that $\alpha = \frac{1}{1 + e^{- 2 \beta x^+ p_1}}$. It follows that
\begin{align*}
\left\{ \alpha \nu_\infty^{z^+} + (1 - \alpha) \nu_\infty^{z^-} \right\}_{\alpha \in (0,1)}  &\subset  \mathcal{G}_\infty (\beta,J,h) .
\end{align*}
Since the set of all limit points is closed, it follows that
\begin{align*}
\overline{\left\{ \alpha \nu_\infty^{z^+} + (1 - \alpha) \nu_\infty^{z^-} \right\}_{\alpha \in (0,1)}} =  \operatorname{conv} \left(  \nu_\infty^{z^+}, \  \nu_\infty^{z^-} \right) \subset \mathcal{G}_\infty (\beta,J,h)
\end{align*}
almost surely, from which the final result follows.
\end{proof}
\subsection{Weak convergence of the metastate probability measures} \label{aw_conv_sec}
\noindent
We have the following uniform tightness result.
\begin{proof}[Proof of \Cref{intensity_tight}]
Let $I \subset \mathbb{N}$ be a finite index set. For $n \geq \max(I)$, by finite permutation invariance of the distribution of $h$, the intensity measure $\mathbb{E} \mu_n$ satisfies the following permutation invariance property
\begin{align*}
\mathbb{E} \mu_n [f \circ \pi_{I}] = \mathbb{E} \mu_n [f \circ \pi_{[|I|]}]
\end{align*}
for any continuous bounded function $f : \mathbb{R}^{|I|} \to \mathbb{R}$, where $[|I|] := \{ 1,2,..., |I_0| \}$. If $J \subset \mathbb{N}$ is any other finite index set such that $|J| = |I|$, then, for $n \geq \max(I \cup J)$ it follows that
\begin{align*}
\mathbb{E} \mu_n [f \circ \pi_{I}] = \mathbb{E} \mu_n [f \circ \pi_{[|I|]}] = \mathbb{E} \mu_n [f \circ \pi_{J}]
\end{align*}
for any continuous bounded function $f : \mathbb{R}^{|I|} \to \mathbb{R}$.
\\
\\
Let $I \subset \mathbb{N}$ be any finite index set as before. For $j \in \mathbb{N} \cup \{ 0 \}$, define the sets $I_j \subset \mathbb{N}$ by $I_j := [|I|] + j |I| := \{ i + j |I| : i \in [|I|] \}$. The sets $I_j$ are disjoint and, for $n \geq \max(I) + |I|$, there exists $k \in \mathbb{N} $ such that $\bigcup_{j=0}^{k-1} I_j \subset [n] \subset \bigcup_{j=0}^{k} I_j$. The number $k+1$ corresponds to the smallest number of translates of $I_0$, including $I_0$, of the form $I_j$ required to cover $[n]$ as a union of such translates such that $[n]$ is a strict subset of the union. 
\\
\\
For $n \geq \max(I) + |I|$, and $k$ as before, it follows that $k |I| \leq n \leq (k + 1)|I|$. Combining together all of the above observations, it follows that
\begin{align*}
 \mathbb{E} \mu_n \left[ || \pi_{I}||^2 \right] = \mathbb{E} \mu_n \left[ || \pi_{I_{0}}||^2 \right] &= \frac{1}{k} \sum_{j=0}^{k-1} \mathbb{E} \mu \left[ || \pi_{I_{j}}||^2 \right] \\ &\leq \frac{1}{k} \sum_{i=1}^n \mathbb{E} \mu_n \left[ || \pi_{i} ||^2 \right] \\ &\leq |I| \left( 1 + \frac{1}{n - |I|}\right) .
\end{align*} 
In the above inequality, we used the spherical constraint to conclude that
\begin{align*}
\sum_{i=1}^n \mathbb{E} \mu_n \left[ || \pi_{i} ||^2 \right] = n ,
\end{align*} 
and the last line follows by considering the inequality $k |I| \leq n \leq (k + 1)|I|$. By Chebyshev's inequality, we have
\begin{align*}
\mathbb{E} \mu_n (|| \pi_{I} || > K) \leq \frac{|I|}{K^2} \left( 1 + \frac{1}{n - |I|}\right) ,
\end{align*}
where $K > 0$. It follows that
\begin{align*}
\lim_{K \to \infty} \lim_{n \to \infty} \mathbb{E} \mu_n (|| \pi_{I} || > K) = 0 
\end{align*}
which implies the uniform tightness of the marginal of $\mathbb{E} \mu_n$ on the index set $I$. Since $I$ was an arbitrary finite index set, and the proof above holds so long as 
$n$ is large enough, it follows that any finite marginal of the sequence of intensity measures $\{ \mathbb{E} \mu_n \}_{n \in \mathbb{N}}$ is uniformly tight and thus the sequence of intensity measures itself is uniformly tight.
\end{proof}
\noindent
We can now provide the proof of the convergence in distribution of the random variable $(h, \mu_n^{\beta,J,h})$.
\begin{proof}[Proof of \Cref{aw_conv1}]
First, recall from \Cref{limit_theorems} that $\left( h, \sqrt{n}(m_n - m) \right) \to (h,G)$ weakly, where $G$ is a non-degenerate 2-dimensional Gaussian random variable independent of $h$. By the Skorohod representation theorem, there exists a probability space $(\Omega', \mathcal{F}', \mathbb{P}')$ on which this weak convergence can be elevated to $\mathbb{P}'$-almost sure convergence and the distributions of these new random variables $(h',\sqrt{n}(m_n' - m))$ and $(h', G')$ are the same as the corresponding random variables in the original space. To be completely exact, the Skorohod representation theorem generates a random variable $Y_n'$ such that $\sqrt{n} (m_n - m)$ and $Y_n'$ agree in distribution and we then define $m_n'$ such that $m_n' = \frac{1}{\sqrt{n}} Y_n' + m$.  Define the set $\Omega'' \subset \Omega'$ by 
\begin{align*}
\Omega'' = \{ G'_1 \not = 0 \} \cap \left\{ \lim_{n \to \infty} n^\frac{1}{2} (m_n' - m) = G' \right\} .
\end{align*}
Since both sets in the intersection are sets of probability $1$, it follows that $\Omega''$ is a set of probability $1$, and we have 
\begin{align*}
\lim_{n \to \infty} n^\frac{1}{2} (m_n' - m) = G' .
\end{align*}
almost surely, where $G'_1 \not = 0$ almost surely. By applying \Cref{conv_class} with $\delta = \frac{1}{2}$ and $\gamma = G'$, it follows that
\begin{align*}
\lim_{n \to \infty} \left(h', \mu_n^{\beta,J,h'} \right) = \mathbbm{1}(G'_1 > 0) \left( h', \nu_\infty^{z^+,h'} \right) + \mathbbm{1}(G'_1 < 0)  \left( h', \nu_\infty^{z^-, h'} \right)
\end{align*}
almost surely. Finally, since $(h', G')$ and $(h,G)$ agree in distribution, it follows that
\begin{align*}
\lim_{n \to \infty} \mathbb{E} f \left(h, \mu_n^{\beta,J,h} \right) &= \lim_{n \to \infty} \mathbb{E}' f \left(h', \mu_n^{\beta,J,h'} \right) \\
&= \mathbb{P}'(G_1' > 0) \mathbb{E}' f \left( h', \nu_\infty^{z^+,h'} \right) +  \mathbb{P}'(G_1' < 0) \mathbb{E}' f \left( h', \nu_\infty^{z^-,h'} \right) \\
&=\frac{1}{2} \mathbb{E} f \left( h, \nu_\infty^{z^+,h} \right) +  \frac{1}{2} \mathbb{E} f \left( h, \nu_\infty^{z^-,h} \right)
\end{align*} 
for any $f \in C_b (\mathbb{R}^\mathbb{N} \times \mathcal{M}_1 (\mathbb{R}^\mathbb{N}))$. 
\end{proof}
\noindent
For the spin glass characterization given in \cite{Pastur1978}, we have the following calculation of the distribution of the magnetization density. 
\begin{lemma} \label{mag_dist} Let $h$ be a random external field which satisfies (A1) and (A2).
\\
\\
For the pure state parameter range, we have
\begin{align*}
\lim_{n \to \infty} \mathbb{E} f \left( \mu_n^{\beta,J,h} \left[ \frac{M_n}{n} \right]\right) = f(x^*) , 
\end{align*}
and, for the mixed state parameter range, we have
\begin{align*}
\lim_{n \to \infty} \mathbb{E} f \left( \mu_n^{\beta,J,h} \left[ \frac{M_n}{n} \right]\right) = \frac{1}{2} f(x^+) + \frac{1}{2} f(x^-) ,
\end{align*}
for any continuous $f \in C_b (\mathbb{R})$.
\end{lemma}
\begin{proof}
From \Cref{can_rep}, we see that
\begin{align*}
\mu_n \left[ f \left( \frac{M_n }{n} \right) \right] = \int_{B(0,1)} \rho_n^{\beta,J,h} (dz) \ f(x) . 
\end{align*}
By \Cref{mix_supp}, for the pure state parameter range, we have
\begin{align*}
\lim_{n \to \infty} \rho_n^{\beta,J,h} \to \delta_{z^*}
\end{align*}
almost surely. The first result follows. The second result follows by essentially repeating the proof of \Cref{aw_conv1} to show that
\begin{align*}
\lim_{n \to \infty} \rho_n^{\beta,J,h} = \mathbbm{1}(G > 0) \delta_{z^+} + \mathbbm{1}(G < 0) \delta_{z^-}
\end{align*}
in distribution, where $\rho_n^{\beta,J,h}$ is understood to be a random probability measure on $B(0,1)$, and $G$ is a standard $1$-dimensional Gaussian random variable independent of $h$.
\end{proof}
\subsection{Convergence of the Newman-Stein metastates} \label{ns_metastate}
\noindent
The following result concerns the almost sure convergence of the Cesàro sum of indicator functions of $A_{n, \delta}$.
\begin{lemma} \label{subseq} Let $h$ be a random external field which satisfies (A1), (A2), and (A4).
\\
\\
For the mixed state parameter range, it follows that
\begin{align*}
\lim_{N \to \infty} \frac{1}{N} \sum_{n=1}^{N} \mathbbm{1} \left( m_n^h - m \not \in  A_{n, \delta} \right) = 0 
\end{align*}
almost surely. 
\end{lemma}
\begin{proof}
Denote the sequence $C_N$ by
\begin{align*}
C_N := \frac{1}{N} \sum_{n=1}^{N} \mathbbm{1} \left( m_n - m \not \in  A_{n, \delta} \right) .
\end{align*}
For each $N$ there exists $K$ such that $2^K \leq N \leq 2^{K + 1}$. For such a $K$, we have 
\begin{align*}
\frac{C_{2^K}}{2} \leq C_N \leq 2 C_{2^{K+1}} .
\end{align*}
It follows that if $C_{2^{K}} \to 0$ almost surely in the limit as $K \to \infty$, then $C_N \to 0$ almost surely in the limit as $N \to \infty$. By Chebyshev's inequality, we have the following estimate
\begin{align*}
\mathbb{P}(C_N > \varepsilon) \leq \frac{1}{\varepsilon N} \sum_{n=1}^N \mathbb{P}(m_n - m \not \in A_{n, \delta}) .
\end{align*}
For the individual terms in the sum, we use the following rough union bound
\begin{align*}
\mathbb{P}(m_n - m \not \in A_{n, \delta}) \leq 3 \max \{ \mathbb{P}(m_n^{\parallel} - m^\parallel \not \in \pi_1 (A_{n, \delta})), \mathbb{P}(m_n^{\perp} - m^\perp \not \in \pi_2 (A_{n, \delta})) \} .
\end{align*}
Note that $\pi_1 (A_{n, \delta}) = \pi_2 (A_{n, \delta})$. To estimate the probabilities, we will use Berry-Essen type uniform bounds for non-linear smooth functions of independent random vectors provided in \cite[section 3]{Pinelis2016}. For $m_n^{\parallel} - m^\parallel$, one can apply standard Berry-Essen bounds with finite third moments, see \cite[Chapter 15]{Klenke2020}, to obtain
\begin{align*}
 \mathbb{P}(m_n^{\parallel} - m^\parallel \not \in \pi_1 (A_{n, \delta})) =  \mathbb{P} \left( |G| < \sqrt{\mathbb{E} h_0^2} n^{- \delta}\right) + \mathbb{P} \left( |G| > \sqrt{\mathbb{E} h_0^2} n^{ \delta}\right) + O \left( \frac{1}{\sqrt{n}}\right) ,   
\end{align*} 
where $G$ is a standard $1$-dimensional Gaussian random variable, and $O$ is the standard big-$O$ asymptotic notation. For $m_n^\perp - m^\perp$, we consider the function $f$ given by
\begin{align*}
f(x,y) := \sqrt{y - x^2 + \mathbb{E} h_0^2} - \sqrt{\mathbb{E} h_0^2} ,
\end{align*}
where $y - x^2 + \mathbb{E} h_0^2 > 0$, which is related to $m_n^\perp - m^\perp$ by 
\begin{align*}
m_n^\perp - m^\perp = f \left( \frac{1}{n} \sum_{i=1}^n h_i, \frac{1}{n} \sum_{i=1}^n (h_i^2 - \mathbb{E} h_0^2)\right) .
\end{align*}
This function is at least twice continuously differentiable in a neighborhood of the origin and its gradient is non-vanishing at the origin. The non-linear Berry-Essen bound with a finite $4 + \xi$:th moment gives us
\begin{align*}
\mathbb{P}(m_n^\perp - m^\perp \not \in \pi_2 (A_{n, \delta})) =  \mathbb{P} \left( |G| < \sqrt{\frac{\mathbb{E} h_0^4 - \left( \mathbb{E} h_0^2\right)^2}{4 \mathbb{E} h_0^2}} n^{- \delta}\right) + \mathbb{P} \left( |G| > \sqrt{\frac{\mathbb{E} h_0^4 - \left( \mathbb{E} h_0^2\right)^2}{4 \mathbb{E} h_0^2}} n^{ \delta}\right) + O \left( \frac{1}{n^{\frac{\xi}{4}}} \right) ,
\end{align*}
where $G$ is as before. For the Gaussian random variables, a rough standard estimate for small value probabilities and large value probabilities shows that
\begin{align*}
\mathbb{P}(|G| < a n^{- \delta}) = a O \left( \frac{1}{n^\delta} \right), \ \mathbb{P}(|G| > b n^\delta) = b O \left( \frac{1}{n^\delta} \right) 
\end{align*}
for $a,b > 0$. Denote $\chi = \min \{ \frac{1}{2}, \frac{\xi}{4}, \delta \}$. Compiling together all of the above bounds, we have
\begin{align*}
\mathbb{P}(m_n - m \not \in A_{n, \delta}) = O \left( \frac{1}{n^\chi}\right) .
\end{align*}
Returning to the Cesàro sum, we have
\begin{align*}
\mathbb{P}(C_N > \varepsilon) =  \frac{1}{\varepsilon} O \left( \frac{1}{N} \sum_{n=1}^N \frac{1}{n^\chi} \right) = \frac{1}{\varepsilon} \frac{1}{N} \sum_{n=1}^N \frac{1}{\left( \frac{n}{N} \right)^\chi} O \left( \frac{1}{N^\chi}\right)   .
\end{align*}
Using Riemann sums, it follows that
\begin{align*}
\lim_{N \to \infty} \frac{1}{N} \sum_{n=1}^N \frac{1}{\left( \frac{n}{N} \right)^\chi} = \int_0^1 dx \ \frac{1}{x^\chi} < \infty .
\end{align*} 
Combining together these results, we see that
\begin{align*}
\mathbb{P}(C_N > \varepsilon) = \frac{1}{\varepsilon} O\left( N^{- \chi} \right) .
\end{align*}
Now, we will apply the Borel-Cantelli lemma, see \cite[Chapter 2]{Klenke2020}. For any rational $\varepsilon > 0$, we have 
\begin{align*}
\sum_{K=1}^\infty \mathbb{P}(C_{2^K} > \varepsilon) = \frac{1}{\varepsilon} O \left( \sum_{K=1}^\infty \left( \frac{1}{2^\chi} \right)^K \right) .
\end{align*}
This implies that the sum on the left hand side is finite for any rational $\varepsilon > 0$. By the Borel-Cantelli lemma $\mathbb{P}(C_{2^K} > \varepsilon \text{ infinitely often}) = 0$, which implies that $C_{2^K} \to 0$ in the limit as $K \to \infty$ and thus $C_N \to 0$ almost surely in the limit as $N \to \infty$ as desired.
\end{proof}
\noindent
Along the sets $A_{n, \delta}$, we have explicit control over the convergence of realizations of the weights $W_n^{\beta, J, h, +}$ and the evaluation maps $\mu_n^{\beta, J, h, \pm} [f]$. This control is presented in the following result.
\begin{lemma} \label{ns_conv}
Let $h$ be a random external field which satisfies (A1) and (A2).
\\
\\
For the mixed state parameter range, it follows that 
\begin{align*}
\lim_{n \to \infty}  \mathbbm{1} \left( m_n^h - m \in A_{n, \delta} \right) \left| W_n^{\beta,J,h,+} - \mathbbm{1} \left( \sum_{i=1}^n h_i > 0 \right)\right| = 0
\end{align*}
and
\begin{align*}
\lim_{n \to \infty}  \mathbbm{1} \left(m_n^h - m \in A_{n, \delta} \right) \left| \mu_n^{\beta,J,h,\pm} [f] - \nu_\infty^{z^\pm,h} [f] \right| = 0
\end{align*}
for any $f \in \operatorname{LBL} (\mathbb{R}^\mathbb{N})$. 
\end{lemma}
\begin{proof}
Let us first split the set $A_{n, \delta}$ into two disjoint sets $A_{n,\delta}^+$ and $A_{n, \delta}^-$ satisfying $A_{n, \delta} = A_{n, \delta}^+ \cup A_{n, \delta}^-$ defined by
\begin{align*}
A_{n, \delta}^\pm = A_{n, \delta} \cap \left\{  \pm \sum_{i=1}^n h_i > 0 \right\} .
\end{align*}
From the proof of \Cref{weight_asymp}, recall that
\begin{align*}
W_n^+ = \frac{1}{1 + a_n e^{-n (\psi_n (z_n^+) - \psi (z^+) - (\psi_n(z_n^-) - \psi (z^-)))}},
\end{align*}
where $\{ a_n \}_{n \in \mathbb{N}}$ is a sequence which converges to $1$ so long as $m_n \to m$. If we consider a realization in the set $A_{n, \delta}^+$ for large enough $n$, then, by considering the asymptotics presented in \Cref{diff_eq2}, it follows that
\begin{align*}
&\psi_n (z_n^+) - \psi (z^+) - (\psi_n(z_n^-) - \psi (z^-)) \\ &\geq 2 \beta x^+ (m_n^\parallel - m^\parallel) - b_n \sum_{|\alpha| = 2} |(m_n - m)^\alpha| - c_n \sum_{|\alpha|=3} |(m_n - m)^\alpha| \\
&\geq 2 \beta x^+ n^{- \frac{1}{2} - \delta} - b_n \alpha_2 n^{-1 + 2 \delta} - c_n \alpha_3 n^{-\frac{3}{2} + 3 \delta} \\
&=n^{- \frac{1}{2} - \delta} (2 \beta x^+ - b_n \alpha_2 n^{- \frac{1}{2} + 3 \delta} - c_n \alpha_3 n^{-1 + 4 \delta}) \\
&\geq n^{- \frac{1}{2} - \delta} (2 \beta x^+ - b_n \alpha_2 n^{- \frac{1}{2} + 3 \delta} - c_n \alpha_3 n^{-1 + 4 \delta}),
\end{align*}
where the sequences $\{ b_n \}_{n \in \mathbb{N}}$ and $\{ c_n \}_{n \in \mathbb{N}}$ consist of terms arising from the matrices $Q_n$  and $H_n$ along with the error terms $R_\alpha$ present in \Cref{diff_eq2}. These sequences converge to some non-negative constants when $m_n \to m$. The terms $\alpha_2$ and $\alpha_3$ are the constants related to the number of 2 dimensional multi-indices of degree 2 and 3 respectively. Since $\delta \in (0, \frac{1}{6})$, it follows that
\begin{align*}
2 \beta x^+ - b_n \alpha_2 n^{- \frac{1}{2} + 3 \delta} - c_n \alpha_3 n^{-1 + 4 \delta} \to 2 \beta x^+ > 0
\end{align*}
when $m_n \to m$. Combining together these observations, for large enough $n$, we have
\begin{align*}
\mathbbm{1} \left( m_n - m \in A^+_{n, \delta} \right) \left| W_n^+ - \mathbbm{1} \left( \sum_{i=1}^n h_i > 0 \right) \right| \leq a_n e^{- n^{ \frac{1}{2} - \delta} (2 \beta x^+ - b_n \alpha_2 n^{- \frac{1}{2} + 3 \delta} - c_n \alpha_3 n^{-1 + 4 \delta}) },
\end{align*}
from which it follows that
\begin{align*}
\lim_{n \to \infty} \mathbbm{1} \left( m_n - m \in A^+_{n, \delta} \right) \left| W_n^+ - \mathbbm{1} \left( \sum_{i=1}^n h_i > 0 \right) \right| = 0 .
\end{align*}
A similar analysis done for $A_{n, \delta}^-$ shows that
\begin{align*}
\lim_{n \to \infty} \mathbbm{1} \left( m_n - m \in A^-_{n, \delta} \right) \left| W_n^+ - \mathbbm{1} \left( \sum_{i=1}^n h_i > 0 \right) \right| = 0 .
\end{align*}
Combining these observations, the first result follows. The second result follows from \Cref{weight_rep} since $\mu_n^\pm [f] \to \nu_\infty^{z^\pm} [f]$ when $m_n \to m$.
\end{proof}
\noindent
Combining \Cref{subseq} with \Cref{ns_conv}, we have the following asymptotic difference result.
\begin{proof}[Proof of \Cref{ns_conv2}]
First, let us denote 
\begin{align*}
\nu_n := \mathbbm{1} \left( \sum_{i=1}^n h_i > 0 \right) \nu_\infty^{z^+} + \left(1 -  \mathbbm{1} \left(\sum_{i=1}^n h_i > 0 \right) \right) \nu_\infty^{z^-} .
\end{align*}
By using the fact that $P$ is Lipschitz in its variables, we have
\begin{align*}
&\left| P \left( \mu_n [f_1],..., \mu_n [f_m] \right) - P \left( \nu_n [f_1],..., \nu_n [f_m] \right)  \right| \\ &\leq a \sum_{j=1}^m \left| \mu_n^+ [f_j] - \nu_\infty^{z^+} [f_j] \right| + b \sum_{j=1}^m \left| \mu_n^- [f_j] - \nu_\infty^{z^-} [f_j] \right| + c \left|W_n^+ - \mathbbm{1} \left( \sum_{i=1}^n h_i > 0\right) \right|
\end{align*}
where $a$, $b$, and $c$ are positive constants that depend on the coefficients of $P$ and the bounds of each $f_j$. By combining this inequality with \Cref{ns_conv}, we have 
\begin{align*}
\lim_{n \to \infty} \mathbbm{1} \left( m_n - m \in A^+_{n, \delta} \right) \left|  P \left( \mu_n [f_1],..., \mu_n [f_m] \right) - P \left( \nu_n [f_1],..., \nu_n [f_m] \right) \right| = 0 .
\end{align*}
Furthermore, one can see that
\begin{align*}
P \left( \nu_n[f_1],..., \nu_n[f_m] \right) &= \mathbbm{1} \left( \sum_{i=1}^n h_i > 0 \right) P \left( \nu_\infty^{z^+} [f_1],..., \nu_\infty^{z^+} [f_m] \right) \\ &+ \left(1 -  \mathbbm{1} \left(\sum_{i=1}^n h_i > 0 \right) \right) P \left( \nu_\infty^{z^-} [f_1],..., \nu_\infty^{z^-} [f_m] \right) .
\end{align*}
It follows that
\begin{align*}
\lim_{N \to \infty} \frac{1}{N} \sum_{n=1}^N \mathbbm{1} \left( m_n - m \in A^+_{n, \delta} \right) \left|  P \left( \mu_n [f_1],..., \mu_n [f_m] \right) - P \left( \nu_n [f_1],..., \nu_n [f_m] \right) \right| = 0 .
\end{align*}
By combining this result and \Cref{subseq}, the result follows.
\end{proof}
\noindent
We determine the limit in distribution of the sequence $\{ (h, T_N^+ )\}_{N \in \mathbb{N}}$ and the limit points of the sequence $\{ T_N^+\}_{N \in \mathbb{N}}$ almost surely.
\begin{proof}[Proof of \Cref{arcsine_dense}]
For the first result, note that the random variable $T^+_N$ converges in distribution to an arcsine distributed random variable $\alpha$ in the limit as $N \to \infty$, see \cite[Chapter 4]{Spitzer1964}. It follows that the sequence $\{ (h, T_N^+)\}_{N \in \mathbb{N}}$ is uniformly tight. Reusing the almost sure convergence argument from \Cref{subseq}, it follows that
\begin{align*}
\lim_{N \to \infty} \frac{1}{N} \sum_{n=1}^{N} \mathbbm{1} \left( \frac{1}{n} \sum_{i=1}^n h_i = m_n^\parallel \not \in \pi_1 (A_{n, \delta}) \right) = 0 
\end{align*}
almost surely, where $\delta > 0$. Note that assumption (A2) guarantees that $h_0$ has a finite absolute third moment. For the following steps, we fix $\delta \in (0, \frac{1}{2})$. If $I \subset \mathbb{N}$ is any finite index set, and $K \subset \mathbb{R}^I$ is a compact set, then it follows that along the sets $\pi_1 (A_{n, \delta})$,  we have
\begin{align*}
\lim_{n \to \infty} \mathbbm{1}(\pi_I (h) \in K) \mathbbm{1} \left( \frac{1}{n} \sum_{i=1}^n h_i \in \pi_1 (A_{n, \delta}) \right) \left| \mathbbm{1} \left(\sum_{i=1}^n h_i > 0 \right) - \mathbbm{1} \left(\sum_{i \in \{ 1,2,...,n\} \setminus I} h_i  > 0 \right)\right| = 0 .
\end{align*}
Let $f : \mathbb{R}^\mathbb{N} \to \mathbb{R}$ be a local continuous function on compact support, and let $L \in \operatorname{BL}(\mathbb{R})$. Denote $I$ to be the index set that $f$ is local on, and $K$ to be the compact set contained in $\mathbb{R}^I$ that it is supported by. By combining the above limits with the almost sure convergence of the Cesàro sum, it follows that
\begin{align*}
\lim_{N \to \infty} \left| \mathbb{E} f(h) L (T_{N}^+) - \mathbb{E} f(h) \mathbb{E} L \left(\frac{1}{N} \sum_{n=1}^{N} \mathbbm{1} \left( \sum_{i \in \{ 1,2,...,n \} \setminus I} h_i  > 0 \right)\right)  \right| = 0
\end{align*} 
Since $T_N^+$ converges in distribution to an arcsine distributed random variable, it follows that
\begin{align*}
\lim_{N \to \infty} \mathbb{E} f(h) L (T_{N}^+) = \mathbb{E} f (h) \mathbb{E} L (\alpha) .
\end{align*}
Finally, since the subalgebra generated by products of local continuous functions with compact support with bounded Lipschitz functions is separating, it follows that
\begin{align*}
\lim_{N \to \infty} (h, T^+_N) = (h, \alpha)
\end{align*}
in distribution, where $\alpha$ is independent of $h$. 
\\
\\
For the second result, let $\{ q_k \}_{k \in \mathbb{N}}$ be an enumeration of the rationals in $[0,1]$. Define the sets
\begin{align*}
A_{k,j} := \left\{ T_N^+ \in \left[ q_k - \frac{1}{j}, q_k + \frac{1}{j}\right]  \text{ i.o.} \right\} .
\end{align*}
The sets $A_{k,j}$ are measurable with respect to the exchangeable $\sigma$-algebra, and we have
\begin{align*}
\mathbb{P}(A_{k,j}) \geq \limsup_{N \to \infty} \mathbb{P} \left(T_N^+ \in \left[ q_k - \frac{1}{j}, q_k + \frac{1}{j}\right] \right) = \mathbb{P} \left(\alpha \in \left[ q_k - \frac{1}{j}, q_k + \frac{1}{j}\right] \right) > 0 .
\end{align*} 
By the Hewitt-Savage 0-1 law, see \cite[Chapter 12]{Klenke2020}, the sets $A_{k,j}$ are then sets of probability $1$, and, as a result, the intersection $\bigcap_{k,j \in \mathbb{N}} A_{k,j}$ is a set of probability $1$. Furthermore, by a similar random subsequence construction as in the proofs of \Cref{limit_theorems} for random walks, for any $\lambda \in [0,1]$, there exists a random subsequence $N_k$ such that $T^+_{N_k} \to \lambda$, which implies that
\begin{align*}
[0,1] \subset L \left( \{ T_N^+ \}_{N \in \mathbb{N}} \right) .
\end{align*}
For the opposite inclusion, it is enough to observe that any convergent subsequence of $T_N^+$ must belong to $[0,1]$. The result follows. 
\end{proof}
\noindent
We present the full almost sure divergence result and the random subsequence convergence result of the Newman-Stein metastates.
\begin{proof} [Proof of \Cref{main_result4}]
By \Cref{ns_tight_2}, we know that the N-S metastates are uniformly tight almost surely. Let $\{ \overline{\kappa}_{N_k} \}_{k \in \mathbb{N}}$ be any weakly convergent subsequence. By \Cref{ns_conv2}, it follows that 
\begin{align*}
\overline{\kappa}_{N_k} [P] = T_{N_k}^+ P (\nu_\infty^{z^+} [f_1],..., \nu_\infty^{z^+} [f_m]) +  (1 - T_{N_k}^+) P (\nu_\infty^{z^-} [f_1],..., \nu_\infty^{z^-} [f_m]) + o(1) , 
\end{align*}
almost surely, where $P$ and $f_i$ are as in \Cref{ns_conv2}. Since $T^+_{N_k} \in [0,1]$, it follows that there exists a subsubsequence $\{ T^+_{N_{k_j}}\}_{j \in \mathbb{N}}$ such that $T^+_{N_{k_j}} \to \lambda \in [0,1]$ in the limit as $j \to \infty$. Since this is a subsubsequence of a weakly convergent subsequence and the functions $P$ formed a separating subalgebra, it follows that 
\begin{align*}
\lim_{k \to \infty} \overline{\kappa}_{N_k} = \lambda \delta_{\nu_\infty^{z^+}} + (1 - \lambda )\delta_{\nu_\infty^{z^-}} 
\end{align*} 
almost surely. This shows that 
\begin{align*}
L \left( \left\{ \overline{\kappa}_N^{\beta,J,h} \right\}_{N \in \mathbb{N}} \right) \subset \operatorname{conv} (\delta_{\nu_\infty^{z^+, h}}, \delta_{\nu_\infty^{z^-, h}}) ,
\end{align*}
almost surely. For the opposite inclusion, by combining \Cref{arcsine_dense} and \Cref{ns_conv2}, for any $\lambda \in [0,1]$, there exists a subsequence $\{ N_{k} \}_{k \in \mathbb{N}}$ such that 
\begin{align*}
\overline{\kappa}_{N_k} [P] = \lambda P (\nu_\infty^{z^+} [f_1],..., \nu_\infty^{z^+} [f_m]) +  (1 - \lambda) P (\nu_\infty^{z^-} [f_1],..., \nu_\infty^{z^-} [f_m]) + o(1)
\end{align*}
almost surely. Since the N-S metastate are almost surely tight, and the functions $P$ formed a separating subalgebra, it follows that there exists a subsubsequence $\{ N_{k_j}\}_{j \in \mathbb{N}}$ such that
\begin{align*}
\lim_{j \to \infty} \overline{\kappa}_{N_{k_j}} = \lambda \delta_{\nu_\infty^{z^+}} + (1 - \lambda )\delta_{\nu_\infty^{z^-}} 
\end{align*}
almost surely. To be explicit, note that $\overline{\left\{ \overline{\kappa}_N \right\}}_{N \in \mathbb{N}}$ is compact almost surely by uniform tightness and $\overline{\left\{ \overline{\kappa}_{N_k} \right\}}_{k \in \mathbb{N}}$ is a closed subset of a compact set and it too is thus compact almost surely. This is the method of generating the subsequence $\{ N_{k_j}\}_{j \in \mathbb{N}}$. This shows that
\begin{align*}
\operatorname{conv} (\delta_{\nu_\infty^{z^+}}, \delta_{\nu_\infty^{z^-}}) \subset L \left( \left\{ \overline{\kappa}_N \right\}_{N \in \mathbb{N}} \right)
\end{align*}
almost surely, from which the results follow.
\end{proof}
\noindent
We have the following convergence in distribution result.
\begin{proof}[Proof of \Cref{main_result5}]
Since the collection of probability measures of N-S metastates is uniformly tight by \Cref{ns_tight_2}, it is enough to prove convergence in distribution of random variables of the form $\overline{\kappa}_N [P]$, where $P$ is as in \Cref{ns_conv2}. By the same result \Cref{ns_conv2}, by dominated convergence, it follows that
\begin{align*}
\mathbb{E} f \left( \overline{\kappa}_N [P] \right) = \mathbb{E} f \left( T_{N}^+ P (\nu_\infty^{z^+} [f_1],..., \nu_\infty^{z^+} [f_m]) +  (1 - T_{N}^+) P (\nu_\infty^{z^-} [f_1],..., \nu_\infty^{z^-} [f_m])  \right) + o(1) 
\end{align*}
for any $f \in \operatorname{BL}(\mathbb{R})$. By applying \Cref{arcsine_dense}, it follows that 
\begin{align*}
\lim_{N \to \infty}\overline{\kappa}_N = \alpha \delta_{\nu_\infty^{z^+}} + (1 - \alpha) \delta_{\nu_\infty^{z^-}} 
\end{align*}
in distribution, where $\alpha$ is an arcsine distributed random variable independent of $h$. 
\end{proof}
\section*{Acknowledgements}
\noindent
I thank my advisor Jani Lukkarinen for his support and encouragement during this project. For their technical comments and helpful discussions, I thank my colleagues Joona Oikarinen, Brecht Donvil, and Gerardo Barrera Vargas. Finally, I thank Christof Külske for providing some additional details on a proof in the paper \cite{Kuelske1997}. The research has been supported by the Academy of Finland, via an Academy project (project No. 339228) and the Finnish {\em centre of excellence in Randomness and STructures\/} (project No. 346306).
\appendix
\section{Appendix} \label{appendix}
\noindent
For this appendix, the spaces $X$ and $\mathcal{S}$ are Polish spaces.
\subsection{Metastate probability measures}
\noindent
We list here the three primary probability measures of interest. 
\begin{definition} Let $h$ be an $X$-valued random variable, and let $\mu$ be a random probability measure on $\mathcal{S}$.
\\
\\
The joint probability measure $K$ on the skew-space $X \times \mathcal{S}$ is defined via its action on $f \in C_b (X \times \mathcal{S})$ by
\begin{align*}
K[f] := \int_{\Omega} d \mathbb{P} \ \mu[f(h, \cdot)] .
\end{align*}
\end{definition}
\noindent
\begin{definition}  Let $h$ be an $X$-valued random variable and let $\mu$ be a random probability measure on $\mathcal{S}$.
\\
\\
The metastate probability measure $\mathcal{K}$ on $X \times \mathcal{M}_1(\mathcal{S})$ is defined via its action on $f \in C_b(X \times \mathcal{M}_1(\mathcal{S}))$ by
\begin{align*}
\mathcal{K}[f] := \int_{\Omega} d \mathbb{P} \ f (h, \mu) .
\end{align*}
Equivalently, we can define $\mathcal{K}$ to be the distribution of $(h, \mu)$.
\end{definition}
\noindent
\begin{definition} Let $\{ \mu_n \}_{n \in \mathbb{N}}$ define a collection of random probability measures on $\mathcal{S}$.
\\
\\
The Newman-Stein metastate is defined almost surely pointwise via its action on $f \in C_b(\mathcal{M}_1(\mathcal{S}))$ by
\begin{align*}
\overline{\kappa}_N [f] := \frac{1}{N} \sum_{n=1}^N f(\mu_n) .
\end{align*}
Equivalently, we can define it by
\begin{align*}
\overline{\kappa}_N := \frac{1}{N} \sum_{n=1}^N \delta_{\mu_n}
\end{align*}
almost surely.
\\
\\
We define the probability measure $\overline{\mathcal{K}}_N$ on $\mathcal{M}_1 (\mathcal{M}_1 (\mathcal{S}))$ to be the distribution of the Newman-Stein metastate $\overline{\kappa}_N$.
\end{definition}
\subsection{Uniform tightness}
\noindent
In this subsection, we consider the uniform tightness of the probability measures defined in the previous subsection. This first lemma shows that uniform tightness of the joint probability measures is equivalent to uniform tightness of the metastate probability measures.
\begin{lemma} \label{joint_tight}
Let $h$ be an $X$-valued random variable, let $\{ \mu_n \}_{n \in \mathbb{N}}$ define a collection of random probability measures on $\mathcal{S}$, let $\left\{ K_n \right\}_{n \in \mathbb{N}}$ be the associated collection of joint probability measures, and let $\left\{ \mathcal{K}_n \right\}_{n \in \mathbb{N}}$ be the associated collection of metastate probability measures.
\\
\\
The following conditions are equivalent:
\begin{enumerate}
\item The collection of intensity measures $\{ \mathbb{E} \mu_n \}_{n \in \mathbb{N}}$ is uniformly tight.
\item The collection of joint probability measures $\left\{ K_n \right\}_{n \in \mathbb{N}}$ is uniformly tight.
\item The collection of metastate probability measures $\left\{ \mathcal{K}_n \right\}_{n \in \mathbb{N}}$ is uniformly tight.
\end{enumerate}
\end{lemma}
\begin{proof}
Recall that to prove uniform tightness of a random vector, it is enough to prove the uniform tightness of its marginal distributions. For the joint probability measures and the metastate probability measures, the first marginal distribution is simply the distribution of $h$, which is automatically uniformly tight since it is the distribution of a Polish space-valued random variable. Thus it is enough to prove the uniform tightness of the second marginal distributions of the joint probability measures and the metastate probability measures.
\\
\\
The second marginal of the the joint probability measure $K_n$ is the intensity measure $\mathbb{E} \mu_n$, and the second marginal of the metastate probability measure $\mathcal{K}_n$ is the marginal distribution of the random measure $\mu_n$. Since the uniform tightness of the collection of intensity measures is equivalent to the uniform tightness of the associated collection of random probability measures, see \cite[Chapter 4]{Kallenberg2017}, the equivalence follows.
\end{proof}
\noindent
There are now two different notions of uniform tightness which concern the Newman-Stein metastates. This first lemma shows that the Newman-Stein metastates are uniformly tight almost surely if the sequence of probability measures defining them is uniformly tight almost surely.
\begin{lemma} \label{ns_tight}
Let $\{ \mu_n \}_{n \in \mathbb{N}}$ define a collection of random probability measures on $\mathcal{S}$ which is uniformly tight almost surely.
\\
\\
It follows that the collection of Newman-Stein metastates $\{ \overline{\kappa}_N \}_{N \in \mathbb{N}}$ is uniformly tight almost surely.
\end{lemma}
\begin{proof}
By Prokhorov's theorem, the closure $K := \overline{\left\{ \mu_n \right\}_{n \in \mathbb{N}}} \subset \mathcal{M}_1 (\mathcal{S})$ is compact almost surely. Since $\mu_n \in K$ for all $n \in \mathbb{N}$, it follows that $\overline{\kappa}_n (K) = 1$ almost surely which implies the uniform tightness of the Newman-Stein metastates.
\end{proof}
\noindent
This second lemma concerning the Newman-Stein metastates shows that the collection of probability distributions of Newman-Stein metastates is uniformly tight if the corresponding collection of either joint probability measures or metastate probability measures is uniformly tight.
\begin{lemma} \label{ns_distr_tight}
Let $h$ be a $X$-valued random variable, let $ \{ \mu_n \}_{n \in \mathbb{N}}$ be a collection of random probability measures on $\mathcal{S}$, let $\{ K_n \}_{n \in \mathbb{N}}$ be the collection of joint probability measures, let $\{ \mathcal{K}_n \}_{n \in \mathbb{N}}$ be the collection of metastate probability measures, and let $\{ \overline{\kappa}_N \}_{N \in \mathbb{N}}$ be the collection of  Newman-Stein metastate.
\\
\\
Suppose that either the joint probability measures $\{ K_n \}_{n \in \mathbb{N}}$ or the metastate probability measures $\{ \mathcal{K}_n \}_{n \in \mathbb{N}}$ are uniformly tight. It follows that the collection $\{ \overline{\mathcal{K}}_N \}_{N \in \mathbb{N}}$ is uniformly tight. 
\end{lemma}
\begin{proof}
By \Cref{joint_tight}, the supposed condition is equivalent to stating that the collection of metastate probability measures $\{ \mathcal{K}_n \}_{n \in \mathbb{}N}$ is uniformly tight, from which it immediately follows that the collection of random measures $ \{ \mu_n \}_{n \in \mathbb{N}}$ is uniformly tight. Let $\varepsilon > 0$. By uniform tightness, there exists a compact set $C \subset \mathcal{M}_1(\mathcal{S})$ such that $\mathbb{P}(\mu_n \in C) \geq 1 - \varepsilon$ uniformly in $n$. It follows that
\begin{align*}
\mathbb{E} \overline{\kappa}_N (C) = \frac{1}{n} \sum_{n=1}^N \mathbb{P} (\mu_n \in C) \geq 1 - \varepsilon 
\end{align*}
uniformly in $N$. This implies that the collection of intensity measures $\{ \mathbb{E} \overline{\kappa}_N \}_{N \in \mathbb{N}}$ is uniformly tight which is equivalent to the uniform tightness of the collection $\{ \overline{\mathcal{K}}_N\}_{N \in \mathbb{N}}$.
\end{proof}
\subsection{Weak convergence}
\noindent
Let us first give the following result on the triviality of the limiting metastates given almost sure convergence.
\begin{lemma} \label{trivial_conv} Let $h$ be a $X$-valued random variable, let $ \{ \mu_n \}_{n \in \mathbb{N}}$ be a collection of random probability measures on $\mathcal{S}$, and let $\{ \overline{\kappa}_N \}_{N \in \mathbb{N}}$ be the collection of  Newman-Stein metastates.
\\
\\
If $\mu_n \to \mu$ weakly in the limit as $n \to \infty$ almost surely, then
\begin{align*}
\lim_{n \to \infty} (h, \mu_n) = (h, \mu)
\end{align*}
in distribution, and 
\begin{align*}
\lim_{N \to \infty} \overline{\kappa}_N = \delta_{\mu}
\end{align*}
almost surely. 
\end{lemma}
\begin{proof}
For $f \in C_b (X \times \mathcal{S})$, it immediately follows that
\begin{align*}
\lim_{n \to \infty} \mathbb{E} f (h, \mu_n) = \mathbb{E} f (h, \mu),
\end{align*}
from which the first result follows. For the second result, if $g \in C_b (\mathcal{M}_1 (\mathcal{S}))$, then $g(\mu_n) \to g(\mu)$ almost surely in the limit as $n \to \infty$. The Cesàro sum of a convergent sequence converges to the same limit as the convergent sequence itself. It follows that
\begin{align*}
\lim_{N \to \infty}\overline{\kappa}_N [g] = g(\mu)
\end{align*}
almost surely, from which the second result follows.
\end{proof}
\noindent
Note that if a collection of probability measures is uniformly tight, then any separating subset of continuous bounded functions is convergence determining, see \cite[Chapter 3]{Ethier1986}. We list and develop a number of useful approximation results concerning such separating subsets.
\\
\\
This first lemma concerns the joint probability measures.
\begin{lemma} \label{subalgebra0}
Let $\mathcal{A}_X \subset C_b(X)$ be a subalgebra that separates points, let $\mathcal{A}_\mathcal{S} \subset C_b(\mathcal{S})$ be a subalgebra that separates points, and let $C \subset C_b(X \times \mathcal{S})$ be the collection of maps of the form $(h,\phi) \mapsto f(h) g(\phi)$ where $f \in \mathcal{A}_X$ and $g \in \mathcal{A}_{\mathcal{S}}$.
\\
\\
It follows that $C$ is separating.
\end{lemma}
\begin{proof}
See \cite[Chapter 3]{Ethier1986}.
\end{proof}
\noindent
The following lemma will be used as an auxiliary result.
\begin{lemma} \label{subalgebra1}
Let $\mathcal{A} \subset C_b(\mathcal{S})$ be a subalgebra which separates points, and let $\mathcal{A}_{\mathcal{M}} \subset C_b(\mathcal{M}_1(\mathcal{S}))$ be the subalgebra generated by the evaluation maps $\mu \mapsto \mu[f]$ where $f \in \mathcal{A}$.
\\
\\
It follows that $\mathcal{A}_{\mathcal{M}}$ separates points.
\end{lemma}
\begin{proof}
Let $\mu$ and $\nu$ be two probability measures on $\mathcal{S}$ such that $\mu \not = \nu$. Toward a contradiction, suppose that $\mu[f] = \nu[f]$ for all $f \in \mathcal{A}$. Since $\mathcal{A}$ separates points, it follows that $\mathcal{A}$ is separating and thus $\mu = \nu$ which is a contradiction. It follows that there must exist $f \in \mathcal{A}$ such that $\mu[f] \not = \nu[f]$ which implies that $\mathcal{A}_{\mathcal{M}}$ separates points.
\end{proof}
\noindent
Note that the previous theorem concerns polynomials of evaluation maps. The following result now refines the previous one to a simpler subset of functions.
\begin{lemma} \label{subalgebra2}
Let $\mathcal{A} \subset C_b(\mathcal{S})$ be a subalgebra which separates points, and let $C \subset C_b(\mathcal{M}_1(\mathcal{S}))$ be the collection of maps of the form $\mu \mapsto L (\mu[g])$ where $g \in \mathcal{A}$ and $L \in \operatorname{BL}(\mathbb{R})$.
\\
\\
It follows that $C$ is separating.
\end{lemma}
\begin{proof}
This proof is a modification of a proof that can be found for a similar object in \cite[Chapter 4]{Kallenberg2017}. First, by the Cramer-Wold theorem, the distribution of a random vector $Y \in \mathbb{R}^k$ is completely determined by the distribution of its linear combinations $\sum_{i=1}^k t_i Y_i$. If we consider the random vector $(\mu[g_1],..., \mu[g_k])$ where $\mu$ is a random probability measure and $g_1,...,g_k \in \mathcal{A}$, by linearity of the evaluation maps and the subalgebra $\mathcal{A}$, we see that this random vector is completely determined by the distribution of $\mu \left[ \sum_{i=1}^k t_i g_i \right]$. It follows that if $\mu$ and $\nu$ are two random  probability measures such that $\mathbb{E} L (\mu[g]) = \mathbb{E} L(\nu[g])$ for all $L \in \operatorname{BL}(\mathbb{R})$, then, by the previous observation, it follows that the distributions of the random vectors $(\mu[g_1],..., \mu [g_k])$ and $(\nu[g_1],..., \nu[g_k])$ are the same. Because finite degree monomials $P : \mathbb{R}^k \to \mathbb{R}$ are continuous maps, it follows that $\mathbb{E} P (\mu[g_1],..., \mu[g_k]) = \mathbb{E} P (\nu[g_1],..., \nu[g_k])$ for any $g_1,...,g_k \in \mathcal{A}$ and any finite degree monomial $P$. Subsequently, the distributions of all finite polynomials of evaluation maps are also the same which implies that $C$ is separating.
\end{proof}
\noindent
This result will be used for both the metastate probability measures and the Newman-Stein metastates along with their distributions. Since the metastate probability measures also contain the disorder variable, we have the following modification to the previous lemma.
\begin{lemma} \label{subalgebra3}
Let $\mathcal{A}_X \subset C_b(X)$ be a subalgebra that separates points, let $\mathcal{A}_\mathcal{S} \subset C_b(\mathcal{S})$ be a subalgebra that separates points, and let $C \subset C_b(X \times \mathcal{M}_1(\mathcal{S}))$ be the collection of maps of the form $(h,\mu) \mapsto f(h) L(\mu[g])$ where $f \in \mathcal{A}_X$, $g \in \mathcal{A}_{\mathcal{S}}$, and $L$ is a bounded Lipschitz function.
\\
\\
It follows that $C$ is separating.
\end{lemma}
\begin{proof}
See \cite[Chapter 3]{Ethier1986}.
\end{proof}
\section*{Declarations}
\noindent
The author has no competing interests to declare that are relevant to the content of this article. This work was supported by the Academy of Finland, via an Academy project (project No. 339228) and the Finnish {\em centre of excellence in Randomness and STructures\/} (project No. 346306). Data sharing not applicable to this article as no datasets were generated or analysed during the current study.

\end{document}